\documentclass{article}

\usepackage[frenchb,english]{babel}

\usepackage{amsmath}
\usepackage{amssymb}
\usepackage{amsthm}
\usepackage{amsfonts} %% Afin d'avoir mathfrak et les lettres gothiques
\usepackage{stmaryrd}
\usepackage{graphicx}
\usepackage{comment}
\usepackage{dsfont}
\usepackage{xcolor}
\usepackage{mathtools}
\usepackage{bigints}
\usepackage{yhmath}
\usepackage{enumitem} %% Pour supprimer le décalage induit automatiquement par l'environnement itemize ou enumerate avec la commande leftmargin=*
\usepackage{subcaption} %% Pour faire des sous-figures dans les figures
\usepackage[title]{appendix} %% Comme son nom l'indique, une fois n'est pas coutume
\usepackage{geometry} %% Pour ajuster les marges
\usepackage{authblk} %% Pour les détails sur l'auteur
\usepackage{empheq} %% Afin de faire des systèmes d'équations parenthésés
\usepackage{tikz} %% Pour faire des dessins, et en particulier pour écrire des nombres dans des cercles

\newtheorem{lemma}{Lemma}
\newtheorem{propo}{Proposition}
\newtheorem{theor}{Theorem}

\newtheorem{defin}{Definition}
\newtheorem{remar}{Remark}

\newcommand{\vertii}[1]{{\left\vert\kern-0.3ex\left\vert #1 
    \right\vert\kern-0.3ex\right\vert}}

\newcommand{\vertiii}[1]{{\left\vert\kern-0.3ex\left\vert\kern-0.3ex\left\vert #1 
    \right\vert\kern-0.3ex\right\vert\kern-0.3ex\right\vert}}
    
\newcommand*\dd{\mathop{}\!\mathrm{d}}

\newcommand*\circled[1]{\tikz[baseline=(char.base)]{\node[shape=circle,draw,inner sep=2pt] (char) {#1};}}

\title{Collapse of inelastic hard spheres in dimension $d \geq 2$}
\author{Th\'eophile Dolmaire}
\author{Juan J. L. Vel\'azquez}
\affil{Institute for Applied Mathematics, University of Bonn, Endenicher Allee 60, D-53115 Bonn, Germany}

\begin{document}

\maketitle

\begin{abstract}
\noindent
We investigate the collapse of three inelastic particles in dimension $d \geq 2$. We obtain general results of convergence and asymptotics concerning the variables of the dynamical system describing a collapsing system of particles. We prove a complete classification of the singularities when a collapse of three particles takes place, obtaining only two possible orders of collisions between the particles. In the first case we recover that the particles arrange in a nearly-linear chain, already studied by Zhou and Kadanoff, and in the second case we obtain that the particles arrange in a triangle, and we show that, after sufficiently many collisions, the particles collide according to a unique order of collisions, which is periodic. Finally, we construct an initial configuration leading to a nearly-linear collapse, stable under perturbations, and such that the angle between the particles at the time of collapse can be chosen a priori, with an arbitrary precision.
\end{abstract}


\tableofcontents

\section{Introduction}

\numberwithin{equation}{section}

We consider in this article a system of three inelastic hard spheres, with fixed restitution coefficient, in dimension $d \geq 2$. We assume that the particles are spherical, of diameter $1$, and that a pair collides if and only if the distance between the respective centers of the particles of the pair is equal to $1$. The restitution coefficient, denoted by $r$, is a positive real number such that $0 < r < 1$, which quantifies the dissipation of kinetic energy during collisions between pairs of particles. More precisely, when a collision between two particles of respective velocities $v_1$ and $v_2$ occurs, the velocities are immediately modified into $v_1'$ and $v_2'$ such that:
\begin{align*}
\left\{
\begin{array}{cccc}
v_1'+v_2' &=& v_1 + v_2, & \text{(conservation of the momentum)}\\
\left[(v_1'-v_2') \cdot \omega_{1,2}\right] &=& -r \left[(v_1-v_2) \cdot \omega_{1,2}\right], &\text{(dissipation of a fixed fraction of the normal component)}
\end{array}
\right.
\end{align*}
where $\omega_{1,2}$ is the unitary vector pointing from the center of the first particle towards the center of the second. We see that along the collisions the momentum is conserved, as well as the tangential components of the relative velocities, but a \emph{fixed} fraction, equal to $1-r$, of the normal component of such relative velocities, is dissipated. Therefore, a positive amount of kinetic energy is dissipated during each collision.\\
\newline
Our motivation to study such a system comes from kinetic theory, and more precisely, granular materials, where inelastic hard spheres are used to model such materials at the microscopic level. As usual in kinetic theory, one seeks information on matter at a macroscopical scale by understanding the behaviour of a large number of elementary particles. The elastic case, corresponding to $r=1$, is relatively well understood. The system of particles, at the microscopic level, is known as the system of \emph{elastic hard spheres} (see for instance \cite{Szas000} where various questions concerning such a system are discussed). For a system of $N$ elastic hard spheres, considering the first marginal of the $N$-particles density leads to the Boltzmann equation when $N$ goes to infinity. The solutions of the Boltzmann equation, that describes dilute gases, are then in close relation with the behaviour of the elementary components of such gases. This relation is rigorously described by Lanford's theorem (\cite{Lanf975}).\\
In the case of inelastic hard spheres, the BBGKY hierarchy and the correlation functions associated to the particle system are studied in \cite{Petr009} and \cite{BoGe014}, and the associated kinetic equation is well-known: one recovers the inelastic Boltzmann equation, which is the analog of the classical Boltzmann equation for granular materials. A review on the recent results concerning such an equation, stretching from the Cauchy theory to the long time behaviour, and including also numerical investigations, can be found in \cite{CHMR021}. The applications of the inelastic Boltzmann equations are numerous and varied, for granular materials appears everywhere in our environment: snow, dust, wheat in silos, sand can all be described with such an equation. In general, every system composed with a large number of elementary components that interact with friction (or more generally, with dissipation of any sort of energy) is susceptible to be described with the help of the theory of granular materials.\\
Granular materials exhibit behaviours that are typical, in some situations, from liquids (they may flow), from solids (they might form a stable pile), or from gases (as interstellar dust). In addition, some features of granular materials are peculiar (for instance, the pressure in a silo filled with wheat is an increasing function of the depth, until a certain limit, from which it remains constant if one goes deeper). The interested reader may consult the very informative review \cite{BeJN996} about the extreme variety of the behaviour of granular media.\\
However, the rigorous derivation of the kinetic equations to describe such granular materials is still an open question: an analogous version of Lanford's theorem for inelastic hard spheres is missing. In the case of elastic hard spheres, the main difficulty is to prove the convergence of the marginals of the $N$-particles density towards the solutions of the Boltzmann hierarchy. Such a density is described relying on the fact that the dynamics of the elastic hard spheres is globally well-posed (almost everywhere), this is Alexander's theorem (\cite{Alex975}, see also Proposition 4.1.1 page 28 in \cite{GSRT013} for a modern presentation). In the case of inelastic particles, even this very first step of Lanford's program is not under control. This open question is the original motivation that led us to conduct the present work.\\
\newline
The original Alexander's theorem relies on a crucial property of the flow of elastic hard spheres: it preserves the measure in the phase space. Although there exist models of inelastic particles that preserve the measure in the phase space (that is, even if some kinetic energy is dissipated during the collisions, see for example \cite{DoVeNot}), such a property is not true in the case of the model we consider here. Therefore, in addition to the problems already present in the elastic case (such as the possibility of triple collisions, that is, when three or more particles collide together at the same time), one might expect additional difficulties.\\
One of the first consequences of the dissipation of the kinetic energy during the collisions is the spontaneous apparition of inhomogeneities in clouds of inelastic particles, as it was numerically observed, for instance, in \cite{GoZa993}. This behaviour is typical from inelastic particles, for a system of elastic particles would have the opposite trend: the entropy is increasing, so that a system of elastic particles tends to fill the space in a homogeneous way (at least on time scales smaller than the enormous Poincar\'e's return time). Concerning the rich variety of the surprising behaviours of large systems of inelastic particles, the reader may refer to \cite{PoSc005}. When the restitution coefficient is taken even smaller than the range producing spontaneous inhomogeneities, the system of particles tends to concentrate in smaller and smaller regions, and an infinite number of collisions take place in finite time. In some sense, this is a regime of inhomogeneity pushed to its climax. Such a phenomenon is called \emph{inelastic collapse}. In particular, when an inelastic collapse takes place, the numerical simulations of the system break down, because the computers cannot determine the position of the particles after the time of the collapse, which is a sort of horizon. We will see that the inelastic collapse is the main obstruction to obtain an Alexander's theorem for inelastic particles, and we will study this phenomenon in great detail in the present article.\\
\newline
The collapse of inelastic particles was first observed in \cite{BeMa990}, considering a pair of one-dimensional inelastic particles in a domain with a vibrating boundary. It remained to determine if such a phenomenon was peculiar to the dimension one: in \cite{McYo993} was observed the first inelastic collapse of particles in dimension $2$. From this point, it was clear that the collapse was a constitutive feature of the general systems of inelastic particles.\\
\newline
Concerning the one-dimensional collapse, many independent directions were investigated, although an important number of open questions remain, even for this very simple case. In what follows, we describe only one-dimensional systems evolving in a domain without boundary.\\
As a matter of fact, the collapse cannot take place for a system of only two particles, but it was observed in \cite{McYo991} for three particles. With \cite{CoGM995}, where the authors represented the system as a two-dimensional billiard in a corner with appropriate reflection laws, the case of three one-dimensional particles can be considered as fully understood.\\
Concerning the case of four particles, the two-particles system considered in \cite{BeMa990} can be interpreted as a symmetric system of four particles, for which the collapse is also completely characterized. For general systems of four particles though, a lot remains to be understood. If we label the collisions between \circled{1} and \circled{2}, and between \circled{2} and \circled{3}, by respectively $a$ and $b$, for three particles the order of the collisions is completely determined: it has to be an infinite repetition of the pattern $ab$. With four particles or more, such a property is not true anymore, and virtually all the orders are possibly achieved by a well-chosen initial configuration. Inspired by preliminary numerical simulations, the opposite approach is proposed in \cite{CDKK999}, where the authors look for self-similar solutions, achieving patterns of the form $(ab)^n(cb)^n$ (using the natural extension of the notations we introduced for three particles). They observe first that the pattern $abcb$ can be achieved, but that it is not stable, in the sense that a perturbation of the initial data would eventually break this sequence of collisions. Then, they observed that $(ab)^2(cb)^2$ is also feasible, and it is in addition stable. For a thorough discussion of this result, the reader may consult \cite{HuRo023}. Similarly, it seems that all the patterns of the form $(ab)^n(cb)^n$, with $n\geq 2$ are feasible, and stable. But it is remarkable that, for each of these patterns, corresponds a single interval $I_n \subset\, ]0,1[$ of restitution coefficients: for $(ab)^n(cb)^n$ to be stable, $r$ has to lie in $I_n$, these intervals $I_n$ being disjoint, and accumulating as $n$ goes to infinity towards the upper critical value of $r$ for which the collapse of three particles is feasible.\\
From these results, new questions emerge: what about the restitution coefficients outside the intervals $I_n$? Is there stable periodic collision patterns, different from $(ab)^n(cb)^n$? What about stable, but non-periodic patterns? What is the maximal restitution coefficient for which a collapse of four particles is possible? And in addition stable? Outside the intervals $I_n$, it seems that when a collapse takes place, the order of collisions is chaotic. These intervals $I_n$ seem to be the analogs of the windows of stability for the logistic map. However, none of these last statements are proved, nor disproved.\\
For larger numbers of particles, the picture is even less clear. As for the case of four particles, the question of the largest restitution coefficient $r$ for which an inelastic collapse is possible is still an open question. In \cite{BeCa999}, explicit lower and upper bounds on such a critical restitution coefficient are obtained. In the same article, the authors produce an explicit collapse, with an arbitrary number of particles. Nevertheless, such a collapse is likely not stable under perturbations, for it is symmetric. Therefore, the question of finding stable collapses for systems with five or more particles remains open. Let us also mention \cite{ChKZ022}, in which collapses similar to the one of \cite{BeCa999} are constructed, with an elegant geometrical approach.\\
Finally, in \cite{GrMu996}, instead of considering inelastic particles on a line, the system on a circle is considered, exhibiting a quite subtle dynamics.\\
\newline
The case of higher dimensions is less studied. There is essentially a single reference dealing with our model: in \cite{ZhKa996}, Zhou and Kadanoff provide the first, and to the best of our knowledge, only study of a system of three inelastic particles in dimension $d\geq 2$. In particular, two necessary conditions for a collapse to be stable are determined. At the time of the collapse, two particles being in contact with the third one (which is then in a central position) and forming an angle $\overline{\theta}$, such a final angle $\overline{\theta}$ has to be obtuse, and one needs:
\begin{align*}
(-\cos\overline{\theta}) > \frac{2r^{1/3}(1+r^{1/3})}{1+r}\cdotp
\end{align*}
In particular, if $r$ is too close to $1$, no collapse can take place. The critical restitution coefficient preventing the apparition of the collapse corresponds naturally to the optimal upper bound $r_c = 7-4\sqrt{3}$, found in \cite{CoGM995}. On the other hand, if $r$ is small enough, collapses with final angles arbitrarily close to $\pi/2$ are in theory feasible: the one-dimensional result can be deformed in a notable way. The article of Zhou and Kadanoff will be the starting point of our investigations.\\
Let us mention two articles concerning variations of our inelastic system: \cite{ScZh996} proves that three rotating inelastic particles can collapse even for values of $r$ close to $1$, whereas \cite{GSBM998} proves the impossibility of a collapse of three particles on the line or on the circle, when the restitution coefficient depends on the relative velocities, and tends to $1$ when these velocities vanish.\\
To conclude this review of the literature, back to the classical model (with fixed restitution coefficient), let us quote \cite{McYo993}. In this reference, presenting numerical simulations for two-dimensional systems with a large number of particles, it was observed that, a large number of collision take place in a finite time, and the particles involved in the last collisions tend to arrange into one-dimensional-like chains.
%numerically observed
%for the first time that when a collapse takes place in systems with a large number of particles, the particles involved in such a collapse are often arranged in linear structures(
This motivates very strongly the study of the one-dimensional collapse. Let us emphasize that the numerical approach of \cite{McYo993} does not provide information on the systems beyond the time of collapse. In \cite{PoSc005} are presented further investigations on that question, in particular about the length of these linear structures, that they call ``collision chains'' (see Section ``Collision Chains'' page 185).\\
\newline
As we mentioned already, our original motivation was to obtain an Alexander's theorem for inelastic particles. But it turns out that obtaining an Alexander's theorem, even for a system with a very small number of particles, is a very challenging problem. Considering indeed the properties of the collapse, that are intrinsic difficulties of the dynamics of inelastic particles, it becomes clear that an Alexander's theorem has to take into account the possibility of an inelastic collapse, especially because such a phenomenon may be stable. Therefore, the inelastic collapse cannot be associated to a negligible set of initial data.\\
The configuration of a collapsing system of particles cannot be computed explicitly in the neighbourhood of a collapse, due to the infinite number of collisions. However, it does not mean one cannot overcome the time of the collapse, since as it was already noticed by Zhou and Kadanoff (\cite{ZhKa996}), it is possible, in theory, to compute the limiting velocities of the particles at the time of collapse, and in principle, to continue the dynamics. Indeed, on the one hand, the normal components of the relative velocities vanish. On the other hand, if the tangential components of all the particles in contact are non zero, we expect the particles to separate, and to be able to continue the dynamics of the system in a unique way. If now the limiting relative velocity of a pair of particles is zero, we expect that the dynamics cannot be continued, for two particles would remain attached, and any further collisions involving one of these two particles would yield an ill-posed problem.\\
%Of course, when an inelastic collapse takes place in a system of particles, there is no hope to compute directly the configuration of such particles for times larger than the collapse, at least by the naive construction of the flow. However, it does not mean one cannot overcome the collapse, since as it was already noticed by Zhou and Kadanoff (\cite{ZhKa996}), it is possible, in theory, to compute the limiting velocities of the particles at the time of collapse, and although the normal components of the relative velocities are completely dissipated, the tangential components (that is, in the plane of contact between a pair of particles in contact) might be non zero, so that the particles would then separate spontaneously, and one would recover again a system with a well-posed dynamics. We just described a continuation principle of inelastic collapse.\\
\newline
In the present article, in order to understand this continuation principle, we attached ourselves to establish a mathematical description of systems of collapsing inelastic particles, and we hope that our results will  be useful to develop further the theory of such systems.\\
\newline
The main results are the following. We obtain general results of convergence for collapsing systems of three inelastic particles, proving that the pairs of particles that collide infinitely many times are in contact at the time of collapse, and that the normal components of the relative velocities of such pairs vanish. Such results are obtained without assuming any order on the collisions, and hold therefore in full generality, under the only assumption that an inelastic collapse takes place. We also obtain asymptotics on the different variables describing the inelastic particles when the collapse takes place, in particular we prove that all the vanishing variables can be estimated only with the vanishing normal components of the relative velocities, and the times between two consecutive collisions. We hope that such an approach will prove helpful to understand better the collapsing particle systems.\\
We establish a complete classification of the singularities when a collapse of three inelastic particles takes place, proving that, after a sufficiently large number of collisions, only two orders of collisions are possible, namely:
\begin{itemize}
\item either one recovers (up to relabel the particles) the infinite repetition of \circled{0}-\circled{1}, \circled{0}-\circled{2} (corresponding to the setting of Zhou and Kadanoff in \cite{ZhKa996}), in that case the particles arrange in a nearly-linear chain, and we call such a situation a \emph{nearly-linear collapse},
\item or the three pairs are involved in infinitely many collisions, and then (up to relabel the particles) the order becomes eventually the infinite repetition of \circled{0}-\circled{1}, \circled{0}-\circled{2}, \circled{1}-\circled{2}. In that case, the three particles are in contact at the time of collapse, forming a triangular structure. To the best of our knowledge, this is the first description of such a collision order, that we called a \emph{triangular collapse}.
\end{itemize}
Finally, we construct explicitly a nearly-linear collapse, which is stable under perturbation, where in addition the limiting angle between the particles can be chosen a priori in $]\pi/2,\pi]$, with an arbitrary precision. To the best of our knowledge, this is the first construction of an explicit, stable collapse in dimension $d\geq2$.\\
\newline
The plan of the article is as follows. In the second section of the paper we introduce in full detail the dynamical system we will study, and the relevant variables we will need.\\
The third section of the paper is devoted to establish general convergence and asymptotic results concerning the different variables of a system of three collapsing particles. Then, we describe a particular regime concerning the asymptotic behaviour of the different variables of the system when a collapse takes place, that was the setting that Zhou and Kadanoff investigated in \cite{ZhKa996}. We show that in such a regime, that we naturally called the \emph{Zhou-Kadanoff}-regime (or, in short, the \emph{ZK}-regime), the study of the dynamical system simplifies. We expect such a regime to be the only stable one in the phase space.\\
In the fourth section, we study in full detail the nearly-linear collapse, presenting strengthened convergence results, we recall the results of Zhou and Kadanoff concerning this collapse, and we finally construct explicitly a set of initial data, with a positive measure, leading to a nearly-linear collapse, in such a way that the limiting angle between the external particles belongs to any arbitrary interval (as small as one wants), around any angle $\overline{\theta} \in\, ]\pi/2,\pi]$.\\
In the fifth section, we investigate the triangular collapse. We prove that, eventually, only the infinite repetition of the collisions \circled{0}-\circled{1}, \circled{0}-\circled{2}, \circled{1}-\circled{2} can take place, and we study completely the limiting matrix that prescribes the evolution of the velocities of the particles when such a collapse occurs. Such a matrix encodes in particular important information concerning the possible final configuration of the system when the time of the collapse is reached.\\
\newline
In order to obtain an Alexander's theorem for a system of three particles, it seems in particular necessary to characterize the ZK-regime for linear collapses. We obtained formal results in this directions, compiled in the companion paper \cite{DoVeArt}, where we present the study of simplified systems, that govern, we believe, the behaviour of the full dynamical system, and that encode in particular the characterization of the Zhou-Kadanoff regime. Such a study is composed with the derivation of formal limiting systems, rigorous investigations on such systems, concerning for instance a complete description of all the fixed points (including also the unstable fixed points) or the long time behaviour of the orbits of the simplified systems, as well as conjectures suggested by numerical results. We conjecture in particular the existence of a separatrix that characterizes the ZK-regime for the formal limiting systems, and we hope to solve such a conjecture in the future, which should provide an Alexander-type result for three inelastic particles.

\section{Obtaining the complete dynamical system}
\label{SECTION__2EcrirSystmDynam}

\subsection{The model}
\label{SSECTI02.1_Le_Modele}

We consider a system of three inelastic particles in the Euclidean space $\mathbb{R}^d$ ($d \geq 2$). Such particles will be denoted by \circled{i} ($i \in \llbracket1,3\rrbracket$). Let us denote $x_i \in \mathbb{R}^d$ and $v_i \in \mathbb{R}^d$ the respective positions and velocities of the three particles ($i \in \llbracket 1,3 \rrbracket$). We assume that the particles are identical spheres (that is, of the same diameter, equal to $1$) that cannot overlap, as a consequence, our system will evolve in the phase space $\mathcal{D}_3$ defined as
\begin{equation}
\label{}
\mathcal{D}_3 = \left\{Z=\left(x_1,x_2,x_3,v_1,v_2,v_3\right) \in (\mathbb{R}^d_x)^3 \times (\mathbb{R}^d_v)^3 = \mathbb{R}^{6d}\ /\ \vert x_i - x_j \vert \geq 1 \text{ for all } i<j \right\}.
\end{equation}

\subsubsection{The inelastic hard sphere flow}

Let us denote as $t$ the time variable. We will assume that the particles evolve according to the \emph{inelastic hard sphere flow}. Let us describe such a dynamics. Starting, at $t=0$, from an initial configuration
$$
Z_0 = \left(x_1(0),x_2(0),x_3(0),v_1(0),v_2(0),v_3(0)\right)
$$
in the phase space $\mathcal{D}_3$, the evolution laws of the system are defined as follows. When no pair of particles are in contact, that is, when $\vert x_i - x_j \vert > 1$ for all $1 \leq i<j \leq 3$, and as long as such a condition holds true, we will assume that the particles move in straight line, with constant velocities (Newton's law without any external force), that is
\begin{align}
\label{EQUATSS2.1_Loi_de_Newton_}
\frac{\dd}{\dd t} v_i(t) &= 0, \nonumber\\
x_i(t) &= x_i(0) + t v_i(0).
\end{align}
At the first moment two particles collide, that is, when there exists a time $t_1 \geq 0$ such that $\vert x_i(t_1) - x_j(t_1) \vert = 1$ for a certain pair of indices $i<j$, the velocities of the two particles \circled{$i$} and \circled{$j$} are immediately changed from $v_i = v_i(0)$ and $v_j = v_j(0)$ into
\begin{equation}
\label{EQUATSS2.1VitesPost-Colli}
\left\{
\begin{array}{rl}
v_i' &= v_i - \frac{(1+r)}{2}(v_i-v_j)\cdot \omega_{i,j} \omega_{i,j},\\
v_j' &= v_j + \frac{(1+r)}{2}(v_i-v_j)\cdot \omega_{i,j} \omega_{i,j},
\end{array}
\right.
\end{equation}
where $r \in [0,1]$ is the \emph{restitution coefficient}, measuring the ``inelasticity'' of the collision, and $\omega_{i,j}$ is the normalized vector $\left(x_i(t_1)-x_j(t_1)\right)/\vert x_i(t_1)-x_j(t_1) \vert$ joining the centers of the two colliding particles \circled{$i$} and \circled{$j$}, at the time of the collision $t_1$. After the first collision happened, and the velocities of the colliding particles are changed according to the laws \eqref{EQUATSS2.1VitesPost-Colli}, the particles will again move in straight line, with constant velocities, until the next collision happens, and so on. Such a process defines the inelastic hard sphere dynamics.

\begin{defin}[$r$-inelastic hard sphere flow]
\label{DEFINSS2.1FlotrInelaSphDu}
Let $r \in \ ]0,1[$ be a positive number smaller than $1$, and let us consider the mapping $\mathcal{D}_3\times\mathbb{R}_+ \rightarrow \mathcal{D}_3$, $\left(Z_0,t\right) \mapsto Z(t)$, where $Z(t)$ is defined recursively using the free transport \eqref{EQUATSS2.1_Loi_de_Newton_} as long as $Z(t)$ belongs to $\overset{\circ}{\mathcal{D}}_3$, and the collision mapping \eqref{EQUATSS2.1VitesPost-Colli} for the times $t_n$ such that $Z(t) \in \partial\mathcal{D}_3$.\\
Such a mapping is called the \emph{$r$-inelastic hard sphere flow}.
\end{defin}
 
\subsubsection{Dissipation of the kinetic energy}

In the rest of this work, it will be useful to decompose the relative velocities into their projections on the lines of contact $\text{span}\{\omega_{i,j}\}$, and their projections on the orthogonal of these lines.

\begin{defin}[Normal and tangential components of the relative velocities]
\label{DEFINSS2.1DecmpVitesNrmTg}
For two particles \circled{i} and \circled{j} in contact, with respective positions $x_i$ and $x_j$, and respective velocities $v_i$ and $v_j$, we will denote by $W_{i,j}$ the relative velocity $v_j-v_i$, and often by $W_j$ when \circled{i} is the particle \circled{0}.\\
We define the \emph{normal component of the relative velocity of the pair} \circled{i}-\circled{j} as the orthogonal projection $\eta_{i,j}$ of the relative velocity $v_j-v_i$ on the line of contact $\omega_{i,j} = \frac{(x_j-x_i)}{\vert x_j-x_i \vert} = x_j-x_i$, that is:
\begin{align}
\eta_{i,j} = \left(v_j-v_i\right) \cdot \left(x_j-x_i\right).
\end{align}
We define also the \emph{tangential component of the relative velocity of the pair} \circled{i}-\circled{j} as the orthogonal projection $W_{i,j}^\perp$ of the relative velocity $v_j-v_i$ on the orthogonal to the line of contact $\omega_{i,j} = \frac{(x_j-x_i)}{\vert x_j-x_i \vert} = x_j-x_i$, that is:
\begin{align}
W_{i,j}^\perp = (v_j-v_i) - \left(v_j-v_i\right) \cdot \left(x_j-x_i\right) \left(x_j-x_i\right) = (v_j-v_i) - \left( \left(v_j-v_i\right) \cdot \omega_{i,j} \right) \omega_{i,j}
\end{align}
\end{defin}

\begin{remar}
We defined the normal component $\eta$ as a scalar, and the tangential component $W^\perp$ as a vector of $\mathbb{R}^d$, even though it corresponds to a projection on a $(d-1)$-dimensional space.
\end{remar}
\begin{remar}
\label{REMARSS2.1DissiCompoNorma}
For a pair \circled{i} and \circled{j} of colliding particles, \eqref{EQUATSS2.1VitesPost-Colli} can be rewritten as:
\begin{align}
\eta_{i,j}' = -r \eta_{i,j}.
\end{align}
A \emph{fixed, constant} proportion of the normal component of the relative velocity is dissipated during each collision.
\end{remar}
\noindent
In the model we are considering, the momentum is clearly conserved along the collisions, but we are prescribing a loss of kinetic energy during each collision. The restitution coefficient is prescribing such a loss, quantified in the following lemma.

\begin{lemma}
\label{LEMMESS2.1DissiEnergCinet}
Let $r \in \ ]0,1[$ be a positive real number smaller than $1$, and let us consider a pair of particles \circled{i} and \circled{j}, with respective velocities $v_i$ and $v_j$, that collide according to the law \eqref{EQUATSS2.1VitesPost-Colli}.\\
Then the dissipation of the kinetic energy during the collision is quantified as follows. If we denote by $\mathcal{E}_c$ the kinetic energy of the pair before the collision, and $\mathcal{E}_c'$ the kinetic energy after the collision, we have:
\begin{equation}
\label{EQUATSS2.1DissiEnergCinet}
\mathcal{E}_c' = \mathcal{E}_c - \frac{(1-r^2)}{4} \eta_{i,j}^2 = \frac{1}{4} \vert v_i + v_j \vert^2 + \frac{1}{4} \vert W_{i,j}^\perp \vert^2 + \frac{r^2}{4} \eta_{i,j}^2.
\end{equation}
\end{lemma}

\noindent
By definition, the particles lose a certain ratio of their \emph{normal} relative velocity in each collision, fixed by the positive coefficient $r$, as observed in Remark \ref{REMARSS2.1DissiCompoNorma}. %Note that there exist also other models, in which the dissipated normal relative velocity is not a linear function of the pre-collisional normal relative velocity.\\
When $r=1$, we recover the case of the elastic hard sphere model. In that case, the kinetic energy is preserved. When $r=0$, all the normal relative velocity is dissipated. In particular, in the one dimensional case $d=1$, the particles remain ``glued'' after such a collision, in the literature authors often refer to the case $r=0$ as ``sticky particles''. Note that there is no reason for the particles to remain attached to each other after such a collision in higher dimension: if the pre-collisional tangential relative velocity is non zero, it will remain unchanged after the collision, and the particles will then immediately separate.

\subsubsection{The problem of the well-posedness of the inelastic hard sphere flow}

It is important to emphasize that this inelastic hard sphere flow is well-defined only when binary collisions take place: if the three particles are in contact at the same time, the laws of the dynamics \eqref{EQUATSS2.1VitesPost-Colli} defining the post-collisional velocities are ill-posed (in general, the system defining the post-collisional velocities becomes over-determined). We have therefore to restrict ourselves to study dynamics emerging from initial conditions that never lead to triple collisions.\\
Another problem that can also appear in the sticky particles regime, linked to the triple collision problem, is the following: if $r=0$, and if a pair of colliding particles has no relative tangential velocity at the time of collision, then these two particles will stay glued together after the collision. Therefore, if one of these two particles collides with a third particle in the future, the dynamics will be again ill posed, due to a triple collision. In the case $r=0$, one has therefore also to restrict oneself to collisions between particles such that the pre-collisional relative velocity is not colinear to the direction between the two centers.\\
However, in the present work, we shall not discuss the elastic regime ($r=1$), already addressed by Alexander (see \cite{Alex975}, and Proposition 4.1.1, in Section ``The $N$-particle flow'' in the more recent reference \cite{GSRT013}), nor the sticky particles, which is, to the best of our knowledge, left completely unstudied so far.\\
In addition to the reasons of ill-posedness just described, a last one can (and will) occur in the case we are interested in here ($0 < r < 1$): the existence of accumulation points of the sequence of collision times between the particles. In that case, there is no finite sequence of time intervals on which it is possible to apply the free transport on the particles, between two collisions, in order to reach any time beyond an accumulation point. Therefore, the dynamics is not clearly globally well-defined. Such a phenomenon is the central object of the present article.

\subsection{From the dynamics of the particles to a discrete dynamical system}

\subsubsection{Reducing the number of variables}

In the general case, our system of particles is described by a trajectory $Z:t\mapsto Z(t)$, a time dependent function taking its values in the phase space, which is a subset of $\mathbb{R}^{6d}$. With such a representation, we keep track at any time of the positions and the velocities of the three particles, each of these vectors being an element of $\mathbb{R}^d$. We have defined a continuous dynamical system.\\
Nevertheless, it is possible to restrict the dimension of the dynamical system we are studying. An approach, used in a fruitful way by Zhou and Kadanoff (\cite{ZhKa996}), consists in working in a frame attached to one of the particles. The immediate drawback of this approach is that the momentum is not conserved in this frame, which is not Galilean anymore, but we will be able to describe the evolution of the dynamical system in a particularly efficient way. More precisely, Zhou and Kadanoff focus on the case when the infinite sequence of collisions between the particles is \circled{0}-\circled{1}, then \circled{0}-\circled{2}, and again \circled{0}-\circled{1}, and so on, over and over again. %We will see that, in the case of infinite collisions between three particles, there are indeed excellent reasons to believe that only this sequence of collisions is possible.
Therefore, in this case \circled{0} has a central role: it will collide infinitely many times with the two others, while both of the two other particles will collide only with the central one.\\
As a consequence, it will make sense to choose this central particle as the origin of the frame, and to measure quantities such as relative velocities and distances from the velocity and position of the central particle. In other words, we will assume that the central particle \circled{0} is at rest, lying at the center of our frame, and we will measure the relative positions and velocities of the two other particles \circled{1} and \circled{2} in this frame.\\
On the other hand, when a collision occurs, the central particle is systematically involved, one of the two other particles is of course a colliding particle, while the third one remains ``spectator'' to this collision. But when the next collision takes place, the role of the colliding and the spectator particles are exchanged. Following Zhou and Kadanoff, this observation motivates to consider another description of the dynamical system: a first set of variables will describe the movement, with respect to the central particle, of the colliding particle, whereas a second set of variables will describe the behaviour of the spectator particle. Therefore, when the sequence of collisions is the infinite repetition of the period \circled{0}-\circled{1}, \circled{0}-\circled{2}, after each collision, the first set will describe, alternatively, the movement of the particle \circled{1} then the particle \circled{2}, and vice versa concerning the second set of variables.\\
Finally, without any assumption on the order of the collisions, we see that we can generalize the approach of Zhou and Kadanoff: in the end, there is no need to have a central particle along the whole sequence of collisions. What matters is the order of the collisions, that defines which colliding particle will become spectator for the next collision, and which particle that just collided will replace it. In other words, it is enough to describe a collision between \circled{0} and \circled{2}, assuming that it follows directly a collision between \circled{0} and \circled{1}. The complete evolution of the system is then simply a composition of this collision mapping, with the operation consisting only in relabelling of the particles at each collision.

\subsubsection{Writing the evolution laws of the dynamical system}

Now that the general ideas in order to parametrize the system have been introduced, let us now introduce the discrete dynamical system that we will investigate. %Let us assume that the particle \circled{0} will play the role of the central particle, that is, it will collide infinitely many times with the two particles \circled{1} and \circled{2}. The positions and velocities of these two last particles will be measured from the central particle \circled{0}. 

\paragraph{The initial data.} Let us start with the initial data. We will assume that, at the initial time $t = 0$, the two particles \circled{0} and \circled{1} are in contact, and just collided, while the particle \circled{2} is at a positive distance from the two other particles \circled{0} and \circled{1}. Therefore, if we denote $v_0$ and $v_1$ the respective velocities of the particles \circled{0} and \circled{1} just after this initial collision, and if we denote $x_0 = 0$ and $x_1$ the respective positions of these particles at the initial time, we have:
\begin{align}
\left\vert x_1-x_0 \right\vert = 1,
\end{align}
translating the fact that \circled{0} and \circled{1} are in contact. Let us then denote by $\omega_1$ the difference $x_1 - x_0 \in \mathbb{S}^{d-1}$. We have also
\begin{align}
\label{EQUATSS2.2LoiEvW1Om1Init1}
\left(v_1 - v_0\right) \cdot \left(x_1 - x_0\right) = \left( v_1 - v_0 \right) \cdot \omega_1 > 0,
\end{align}
%% Équations d'évolution, lien entre la vitesse relative W1 avec omega 1, première version
describing the fact that the two particles \circled{0} and \circled{1} are in a post-collisional configuration, that is, the distance between these two particles is increasing for small positive times, since we have:
\begin{align*}
\frac{\dd}{\dd t} \left[ t \mapsto \left\vert x_1(t) - x_0(t) \right\vert \right] &= \frac{2 \left( \frac{\dd}{\dd t} \left( x_1(t) - x_0(t) \right) \right) \cdot \left( x_1(t) - x_0(t) \right)}{2 \left\vert x_1(t) - x_0(t) \right\vert} \\
&= \frac{\left( v_1 - v_0 \right) \cdot \left( x_1 + tv_1 - x_0 - tv_0 \right)}{\left\vert x_1(t) - x_0(t) \right\vert} \cdotp
\end{align*}
%so the time derivative of the distance function $t \mapsto \left\vert x_1(t) - x_0(t) \right\vert$ at the initial time is given by $\left( v_1 - v_0 \right) \cdot \left( x_1 - x_0 \right)$.
Introducing then the relative velocity $W_1 = v_1 - v_0$ of the particle \circled{1} with respect to \circled{0}, \eqref{EQUATSS2.2LoiEvW1Om1Init1} can be rewritten as
\begin{align}
\label{EQUATSS2.2LoiEvW1Om1Init2}
W_1 \cdot \omega_1 > 0.
\end{align}
We assumed also that the particle \circled{2} is at a positive distance from the two others at the initial time, and in particular, we have
\begin{align}
\label{EQUATSS2.2LoiEvDistaInit1}
\left\vert x_2-x_0 \right\vert > 1.
\end{align}
If we denote by $\omega_2$ the vector $\displaystyle{\frac{x_2-x_0}{\vert x_2 - x_0 \vert}} \in \mathbb{S}^{d-1}$, and $1+d = \vert x_2 - x_1 \vert$, $x_2 - x_0 = (1+d) \omega_2$, and \eqref{EQUATSS2.2LoiEvDistaInit1} can be rewritten as
\begin{align}
d > 0.
\end{align}
Without loss of generality we can assume that \circled{0} is also involved in the next collision, together with \circled{2}, and $\left( v_2 - v_0 \right) \cdot \left( x_2 - x_0 \right) < 0$, or again, if we denote by $W_2$ the relative velocity $v_2 - v_0$:
\begin{align}
\label{EQUATSS2.2LoiEvW2Om2Init1}
W_2 \cdot \omega_2 < 0.
\end{align}
In the end, the initial configuration of the system is described with the help of the following collection of variables:
\begin{align}
\big( \hspace{-2mm} \underbrace{\omega_1,W_1}_{\substack{\text{position and}\\ \text{velocity of}\\ \text{the particle}\\ \circled{1}}},\underbrace{d,\omega_2,W_2}_{\substack{\text{position and}\\ \text{velocity of}\\ \text{the particle}\\ \circled{2}}} \hspace{-0.5mm} \big) \in \mathbb{S}^{d-1} \times \mathbb{R}^d \times \mathbb{R}_+^* \times \mathbb{S}^{d-1} \times \mathbb{R}^d.
\end{align}
We need in total $4d-1$ real variables in order to describe the initial configuration of the system.

\paragraph{The first collision after the initial configuration.} The next collision involves the particles \circled{0} and \circled{2}. The condition \eqref{EQUATSS2.2LoiEvW2Om2Init1} is necessary, but not sufficient in order to have a collision between the particles \circled{0} and \circled{2} at a positive time. Let us first discuss a characterization for such a collision to happen in the future.\\
The distance, along time, between the two particles \circled{0} and \circled{2}, if no collision takes place on the time interval $]0,t[$, is given by
\begin{align*}
\left\vert x_2(t) - x_0(t) \right\vert &= \left\vert (x_2-x_0) + t(v_2-v_0) \right\vert \\
&= \left\vert (1+d)\omega_2 + tW_2 \right\vert \\
&= \sqrt{ (1+d)^2 + 2(1+d) \left( \omega_2 \cdot W_2 \right) t + \left\vert W_2 \right\vert^2 t^2 }.
\end{align*}
The collision between \circled{0} and \circled{2} will take place if and only if there exists a positive time $\tau$ such that
\begin{align}
\label{EQUATSS2.2CondiColliFutur}
\vert W_2 \vert^2 \tau^2 + 2(1+d)\left( \omega_2 \cdot W_2 \right) \tau + d^2 + 2d = 0,
\end{align}
and the particles \circled{1} and \circled{2} do not collide before (\circled{0} and \circled{1} cannot, according to \eqref{EQUATSS2.2LoiEvW1Om1Init2}). The existence of a (real) solution to \eqref{EQUATSS2.2CondiColliFutur} is equivalent to $\Delta \geq 0$, where:
\begin{align}
\label{EQUATSS2.2DiscrTrinox2-x0}
\Delta &= \left[ 2(1+d) \left( \omega_2 \cdot W_2 \right) \right]^2 - 4 \vert W_2 \vert^2 d(2+d) \nonumber \\
&= 4 \left[ (1+d)^2 \left( \omega_2 \cdot W_2 \right)^2 - d(2+d) \vert W_2 \vert^2 \right].
\end{align}
\noindent
In our case, we will require the positivity of the discriminant $\Delta$. In that case, there exist two solutions $\tau$ to the equation \eqref{EQUATSS2.2CondiColliFutur}. These two solutions will be both positive, or both negative, and the sign of the solutions will be the opposite of the sign of the coefficient of the first degree, and according to the condition \eqref{EQUATSS2.2LoiEvW2Om2Init1}, the two roots will be positive. Naturally, the only ``physical'' solution corresponds to the smallest solution of the equation, describing the first time the two balls representing \circled{0} and \circled{2} intersect. The expression of the time of collision is then:
\begin{align}
\label{EQUATSS2.2TempsColliTauV1}
\tau = - \frac{(1+d)\left( \omega_2 \cdot W_2 \right)}{\vert W_2 \vert^2} - \frac{\sqrt{ (1+d)^2 \left( \omega_2 \cdot W_2 \right)^2 - d(2+d) \vert W_2 \vert^2 }}{\vert W_2 \vert^2} \cdotp
\end{align}
%% Expression du temps de collision tau, version 1.
\noindent
Let us finally use the decomposition of the relative velocities introduced in Definition \ref{DEFINSS2.1DecmpVitesNrmTg}, writing %Since the normal direction (that is, the line between the two centers of colliding particles) is playing a crucial role in the inelastic collisions (we recall that the dispersion of the kinetic energy takes place only along that direction), we will decompose the relative velocities as follows.
$W_i = \eta_i \omega_i + W_i^\perp$, with $\omega_i \cdot W_i^\perp = 0$. %$\eta_i \omega_i$ is the normal component of the relative velocity $W_i$, while $W_i^\perp$ is its tangential component.
Note that the sign of $\eta_i$ is describing if the pair of particles (if $i=1$, then we consider the pair \circled{0}-\circled{1}, \circled{0}-\circled{2} if $i = 2$) is in a pre- or in a post-collisional configuration.\\
%Similarly, we will decompose $W'_i = \eta_i' \omega_i' + (W_2^\perp)'$ into normal and tangential components, where $\omega_i'$ is the direction of the pair of particles at the time $\tau$ of the collision between \circled{0} and \circled{2}.\\
Thanks to these notations, it will be possible to simplify the expressions of the different quantities involved in the dynamics of the particle system. In particular, we can rewrite the expression \eqref{EQUATSS2.2TempsColliTauV1} of $\tau$ as
\begin{align}
\label{EQUATSS2.2TempsColliTauV2}
\tau = - \frac{(1+d)\eta_2}{\vert W_2 \vert^2} - \frac{\sqrt{ (1+d)^2 \eta_2^2 - d(2+d) \left( \eta_2^2 + \vert W_2^\perp \vert^2 \right) }}{\vert W_2 \vert^2} \nonumber\\
= \frac{(1+d)(-\eta_2)}{\vert W_2 \vert^2} - \frac{\sqrt{ (1+d)^2 \eta_2^2 - d(2+d) \vert W_2 \vert^2 }}{\vert W_2 \vert^2} \cdotp
\end{align}
\noindent
Considering \eqref{EQUATSS2.2TempsColliTauV2}, let us introduce an important quantity, named after \cite{ZhKa996}, since this parameter allows to characterize the main regime studied in this reference.

\begin{defin}[Zhou-Kadanoff parameter]
Let $d$ be a positive number, and $W_2$ be a vector in $\mathbb{R}^d$. We denote by $\zeta$ the positive quantity:
\begin{align}
\label{EQUATSS2.2DefinParZK_zeta}
\zeta = \frac{d(2+d) \vert W_2 \vert^2}{(1+d)^2 \eta_2^2} \geq 0.
\end{align}
The number $\zeta$ will be called the \emph{Zhou-Kadanoff parameter} (in short, the \emph{ZK parameter}).
%% Definition du paramètre Zhou et Kadanoff, la variable zêta.
\end{defin}
\noindent
We will see that the evolution of the Zhou-Kadanoff parameter along the different collisions is fundamental to understand the long time behaviour of the particles. Using the last definition, \eqref{EQUATSS2.2TempsColliTauV2} writes:
\begin{align}
\label{EQUATSS2.2TempsColliTauV3}
\tau &= \frac{(1+d)(-\eta_2)}{\vert W_2 \vert^2} \left[ 1 - \sqrt{1 - \zeta} \right].
\end{align}

%using the assumption \eqref{SECT3EquatLoiEvW2Om2Init1} on the sign of $\eta_2 = \omega_2 \cdot W_2$, in order to deal with positive quantities. The non negative number $\zeta$, defined in \eqref{SECT3EquatDefinParZK_zeta}, will be called the \emph{Zhou and Kadanoff parameter}, in reference to the work of these two authors in \cite{}. We will see the key role of such a parameter in order to describe the behaviour of the system.\\
\noindent
In order to write completely the expressions defining the evolution of our dynamical system, let us first notice that the distance at time $\tau$ between the particle \circled{0} and the spectator particle \circled{1} is
\begin{align*}
1 + d' = \vert x_1(\tau) - x_0(\tau) \vert = \vert \omega_1 + \tau W_1 \vert,
\end{align*}
which can be rewritten as
\begin{align}
\label{EQUATSS2.2LoiEvDista_d'v1}
d' &= \sqrt{ \vert \omega_1 + \tau W_1 \vert^2 } - 1 \nonumber\\
&= \sqrt{ 1 + 2 \omega_1 \cdot W_1 \tau + \vert W_1 \vert^2 \tau^2} - 1 \nonumber\\
&= \sqrt{ 1 + 2 \eta_1 \tau + \vert W_1 \vert^2 \tau^2} - 1.
\end{align}
Let us now rewrite the equations of the post-collisional velocities in a more synthetic way. In order to do so, let us define more precisely the quantities $\omega_1'$ and $\omega_2'$. These two unit vectors represent, respectively, the directions from the central particle \circled{0} to the particle \circled{1}, respectively to the particle \circled{2}, at time $\tau$, that is:
\begin{equation}
\label{EQUATSS2.2DefinOmeg1Prime}
\omega_1' = \frac{x_1(\tau)-x_0(\tau)}{\vert x_1(\tau)-x_0(\tau) \vert} = \frac{\omega_1 + \tau W_1}{(1+d')},
\end{equation}
and
\begin{equation}
\label{EQUATSS2.2DefinOmeg2Prime}
\omega_2' = \frac{x_2(\tau)-x_0(\tau)}{\vert x_2(\tau)-x_0(\tau) \vert} = (1+d)\omega_2 + \tau W_2.
\end{equation}
In particular, using these last notations, \eqref{EQUATSS2.1VitesPost-Colli} can be rewritten as
\begin{align}
\label{EQUATSS2.2VitesPostCVers2}
\left\{
\begin{array}{rl}
v_0' &= v_0 - \frac{(1+r)}{2} \left( - W_2 \cdot \omega_2' \right) \omega_2',\\
v_1' &= v_1,\\
v_2' &= v_2 + \frac{(1+r)}{2} \left( - W_2 \cdot \omega_2' \right) \omega_2',
\end{array}
\right.
\end{align}
which provides, in terms of \emph{relative} post-collisional velocities:
\begin{align}
\label{EQUATSS2.2ViteRPostCVers2}
\left\{
\begin{array}{rlll}
W_1' &= v_1'-v_0' &= (v_1-v_0) + \frac{(1+r)}{2}\left( - W_2 \cdot \omega_2' \right)\omega_2' &= W_1 - \frac{(1+r)}{2} \left( W_2 \cdot \omega_2' \right) \omega_2',\\
W_2' &= v_2'-v_0' &= (v_2-v_0) + (1+r) \left( -W_2 \cdot \omega_2' \right)\omega_2' &= W_2 - (1+r) \left( W_2 \cdot \omega_2' \right) \omega_2'.
\end{array}
\right.
\end{align}
Since we have defined $W_i'=\eta_i'\omega_i'+\left( W_i^\perp \right)'$, we find, for the normal component of the first post-collisional relative velocity:
\begin{align*}
\eta_1' = W_1'\cdot \omega_1' = \left( W_1 - \frac{(1+r)}{2} \left( W_2 \cdot \omega_2' \right) \omega_2' \right) \cdot \omega_1' = W_1 \cdot \omega_1' - \frac{(1+r)}{2} \left( W_2 \cdot \omega_2' \right) \omega_1' \cdot \omega_2'.
\end{align*}
Using \eqref{EQUATSS2.2DefinOmeg1Prime} and \eqref{EQUATSS2.2DefinOmeg2Prime} we find
\begin{align*}
W_1 \cdot \omega_1' = \frac{W_1 \cdot (\omega_1 + \tau W_1) }{(1+d')} = \frac{\eta_1 + \tau \vert W_1 \vert^2}{(1+d')},
\end{align*}
and
\begin{align*}
W_2 \cdot \omega_2' = W_2 \cdot \left( (1+d)\omega_2 + \tau W_2 \right) = (1+d)\eta_2 + \tau \vert W_2 \vert^2,
\end{align*}
therefore we have:
\begin{equation}
\label{EQUATSS2.2Eta_1PrimeVers1}
\eta_1' = \frac{1}{(1+d')} \left( \eta_1 + \tau \vert W_1 \vert^2 \right) - \frac{(1+r)}{2} \left( \omega_1'\cdot\omega_2' \right) \left( (1+d)\eta_2 + \tau \vert W_2 \vert^2 \right).
\end{equation}
In the same way, we have:
\begin{align}
\eta_2' = W_2'\cdot \omega_2' = \left( W_2 - (1+r)\left( W_2\cdot \omega_2'\right)\omega_2' \right) \cdot \omega_2' = W_2 \cdot \omega_2' - (1+r) \left( W_2\cdot\omega_2' \right) = -r W_2\cdot\omega_2'.
\end{align}
Here, we see how much the expression simplifies in this frame. Using again \eqref{EQUATSS2.2DefinOmeg2Prime}, we find
\begin{equation}
\label{EQUATSS2.2Eta_1PrimeVers2}
\eta_2' = -r W_2 \cdot \left( (1+d)\omega_2 + \tau W_2 \right) = -r(1+d) \eta_2 - r\tau \vert W_2 \vert^2.
\end{equation}
We can therefore write (using the intermediate expression we found for $\eta_1'$):
\begin{align}
\label{EQUATSS2.2W1prpPrimeVers1}
\left(W_1^\perp\right)' &= W_1' - \eta_1'\omega_1' = W_1 - \frac{(1+r)}{2} \left( W_2\cdot\omega_2' \right) \omega_2' - \left( W_1 \cdot \omega_1' - \frac{(1+r)}{2}(W_2\cdot\omega_2')\omega_1'\cdot\omega_2' \right) \omega_1' \nonumber\\
&= W_1 - W_1\cdot \omega_1' \omega_1' + \frac{(1+r)}{2} \left( W_2\cdot\omega_2' \right) \big[ (\omega_1'\cdot\omega_2') \omega_1' - \omega_2' \big] \nonumber\\
&= W_1 - \frac{(\eta_1 + \tau \vert W_1 \vert^2)}{(1+d')} \omega_1' + \frac{(1+r)}{2} \left( (1+d)\eta_2 + \tau \vert W_2 \vert^2 \right) \big[ (\omega_1'\cdot\omega_2') \omega_1' - \omega_2' \big],
\end{align}
%% W_1 perp
and, similarly:
\begin{align}
\label{EQUATSS2.2W2prpPrimeVers1}
\left( W_2^\perp \right)' &= W_2' - \eta_2'\omega_2' \nonumber\\
&= W_2 - (1+r) \left( W_2\cdot\omega_2' \right) \omega_2' - \eta_2'\omega_2' \nonumber\\
&= W_2 - (1+r)(W_2\cdot\omega_2')\omega_2' + r W_2\cdot\omega_2'\omega_2' \nonumber\\
&= W_2 - (W_2\cdot\omega_2')\omega_2' \nonumber\\
&= W_2 - \left( (1+d)\eta_2 + \tau \vert W_2 \vert^2 \right) \left( (1+d)\omega_2 + \tau W_2 \right).
\end{align}
%% W_2 perp
Equations \eqref{EQUATSS2.2DefinOmeg1Prime}, \eqref{EQUATSS2.2DefinOmeg2Prime}, \eqref{EQUATSS2.2W1prpPrimeVers1} and \eqref{EQUATSS2.2W2prpPrimeVers1} together with \eqref{EQUATSS2.2LoiEvDista_d'v1} provide almost a complete description of the dynamical system involving only the initial datum described in term of the variables $d$, $\omega_1$, $\omega_2$, $W_1$ and $W_2$. It remains only the describe $\omega_1'\cdot\omega_2'$, which is equal to the cosine of the angle between the particles \circled{1} and \circled{2} at time $\tau$. This can be done writing:
\begin{align}
\label{EQUATSS2.2CosinAngleVers1}
\omega_1'\cdot\omega_2' &= \left( \frac{\omega_1+\tau W_1}{(1+d')} \right) \cdot \left( (1+d)\omega_2+\tau W_2 \right) \nonumber\\
&= \frac{(1+d)}{(1+d')} \left( \omega_1 \cdot \omega_2 \right) + \left[ \frac{1}{(1+d')} \omega_1 \cdot W_2 + \frac{(1+d)}{(1+d')}\omega_2 \cdot W_1 \right] \tau + \frac{W_1 \cdot W_2}{(1+d')} \tau^2 \cdotp
\end{align}
%\eqref{EQUATSS2.2CosinAngleVers1}, that described essentially the perturbation of the angle between the relative positions of the particles \circled{1} and \circled{2} from \circled{0} along the inelastic hard sphere flow, can be written only in terms of the initial data in a straightforward way, but there is no need to provide such a long and complicated expression here.\\
%\newline
\noindent
We can therefore convert the continuous dynamical system associated to the $r$-inelastic hard sphere flow into a discrete dynamical system. More precisely, let us write entirely the evolution of the variables $\left(\omega_1,W_1,d,\omega_2,W_2\right)$, describing a transition from a post-collisional configuration just after a collision of type \circled{0}-\circled{1}, to a post-collisional configuration just after a collision of type \circled{0}-\circled{2}.

\begin{defin}[Complete single-collision mapping]
Let us consider three particles \circled{0}, \circled{1}, \circled{2} in $\mathbb{R}^d$, of respective positions $x_i$ and velocities $v_i$ ($x_i,v_i \in \mathbb{R}^d \ \forall i \in \llbracket 1,3 \rrbracket$) described by the configuration:
\begin{align}
\left(\omega_1,W_1,d,\omega_2,W_2\right) \in \mathbb{S}^{d-1} \times \mathbb{R}^d \times \mathbb{R}_+^* \times \mathbb{S}^{d-1} \times \mathbb{R}^d,
\end{align}
where:
\begin{itemize}
\item $\omega_1 = x_1-x_0$ ( \circled{0} and \circled{1} are initially in contact),
\item $W_1 = v_1-v_0$,
\item $1+d = \vert x_2 - x_0 \vert > 0$ ( \circled{0} and \circled{2} are initially separated),
\item $(1+d) \omega_2 = x_2-x_0$,
\item $W_2 = v_2-v_0$.
\end{itemize}
Let us assume in addition that \circled{0} and \circled{1} are in a post-collisional configuration, and that a collision of type \circled{0}-\circled{2} is the next collision that will take place:
\begin{align}
\eta_1 = W_1 \cdot \omega_1 > 0,\hspace{5mm} \eta_2 = W_2 \cdot \omega_2 < 0, \hspace{2mm} \text{and} \hspace{2mm} \zeta = \frac{d(2+d) \vert W_2 \vert^2}{(1+d)^2\eta_2^2} < 1.
\end{align}
We define then \emph{complete single-collision mapping} as the function:
\begin{align*}
\mathfrak{C} :
\left\{
\begin{array}{clc}
\mathbb{S}^{d-1} \times \mathbb{R}^d \times \mathbb{R}_+^* \times \mathbb{S}^{d-1} \times \mathbb{R}^d &\rightarrow &\mathbb{S}^{d-1} \times \mathbb{R}^d \times \mathbb{R}_+^* \times \mathbb{S}^{d-1} \times \mathbb{R}^d, \\
\left(\omega_1,W_1,d,\omega_2,W_2\right) &\mapsto &\left(\omega_1',W_1',d',\omega_2',W_2'\right),
\end{array}
\right.
\end{align*}
where:
\begin{align}
\label{EQUATSS2.2IterationSystm1}
\left\{
\begin{array}{rl}
\omega_1' &= \displaystyle{\frac{\omega_1 + \tau W_1}{(1+d')}},\vspace{1mm} \\
\eta_1' &= \displaystyle{\frac{1}{(1+d')} \left( \eta_1 + \tau \vert W_1 \vert^2 \right) - \frac{(1+r)}{2} \left( \omega_1'\cdot\omega_2' \right) \left( (1+d)\eta_2 + \tau \vert W_2 \vert^2 \right)}, \vspace{3mm}\\
(W_1^\perp)' &=W_1 - W_1\cdot \omega_1' \omega_1' + \displaystyle{\frac{(1+r)}{2}} \left( W_2\cdot\omega_2' \right) \big[ (\omega_1'\cdot\omega_2') \omega_1' - \omega_2' \big] \vspace{1mm}\\
&\hspace{20mm}= W_1 - \displaystyle{\frac{(\eta_1 + \tau \vert W_1 \vert^2)}{(1+d')}} \omega_1' + \displaystyle{\frac{(1+r)}{2}} \left( (1+d)\eta_2 + \tau \vert W_2 \vert^2 \right) \big[ (\omega_1'\cdot\omega_2') \omega_1' - \omega_2' \big], \vspace{3mm}\\
d' &= \displaystyle{\sqrt{ 1 + 2 \eta_1 \tau + \vert W_1 \vert^2 \tau^2} - 1}, \vspace{3mm}\\
\omega_2' &= (1+d)\omega_2 + \tau W_2, \vspace{3mm}\\
\eta_2' &= -r(1+d) \eta_2 - r\tau \vert W_2 \vert^2, \vspace{3mm}\\
(W_2^\perp)' &= W_2 - (W_2\cdot\omega_2')\omega_2' \vspace{2mm}\\
&\hspace{20mm}= W_2 - \left( (1+d)\eta_2 + \tau \vert W_2 \vert^2 \right) \left( (1+d)\omega_2 + \tau W_2 \right),
\end{array}
\right.
\end{align}
% (W_1^\perp)' &= \displaystyle{\frac{(1+r)}{2} \left( W_2\cdot\omega_2' \right) \big[ (\omega_1'\cdot\omega_2') \omega_1' - \omega_2' \big]} \vspace{1mm}\\
%&\hspace{25mm}= \displaystyle{\frac{(1+r)}{2} \left( (1+d)\eta_2 + \tau \vert W_2 \vert^2 \right) \big[ (\omega_1'\cdot\omega_2') \omega_1' - \omega_2' \big]}, \vspace{3mm}\\
\noindent
with $\omega_1'\cdot \omega_2'$ defined in \eqref{EQUATSS2.2CosinAngleVers1}, and $\tau$ is defined as
\begin{align*}
\tau = \frac{(1+d)(-\eta_2)}{\vert W_2 \vert^2} \left[ 1 - \sqrt{1 - \zeta} \right], \text{  and  } 
\zeta = \frac{d(2+d) \vert W_2 \vert^2}{(1+d)^2 \eta_2^2} \geq 0.
\end{align*}
\end{defin}
\noindent
The mapping $\mathfrak{C}$ introduced in the previous definition encodes completely the dynamics of the system of three inelastic particles, describing the evolution of the particles between a collision of type \circled{0}-\circled{1} and a collision of type \circled{0}-\circled{2}. All the distances and velocities are measured from the position $x_0$ and velocity $v_0$ of the particle \circled{0}, central for the pair of collisions \circled{0}-\circled{1}, \circled{0}-\circled{2}.

\begin{remar}
$\mathfrak{C}$ defines a mapping from $\mathbb{S}^{d-1} \times \mathbb{R}^d \times \mathbb{R}_+^* \times \mathbb{S}^{d-1} \times \mathbb{R}^d$ into itself: we have a $(4d-1)$-dimensional dynamical system, $7$-dimensional for $d=2$, $11$-dimensional for $d=3$.\\
Nevertheless, we mentioned already that $W_i^\perp$ belong by definition to $\omega_i^\perp$, so it would be possible to reduce the dimension of the dynamical system by $2$ in theory. In any case, from the expressions of \eqref{EQUATSS2.2IterationSystm1}, it is clear that we are facing a complicated dynamical system.
\end{remar}

\section{General properties of the inelastic collapse of three particles}
\label{SECTION__3ProprGenerColla}

This section is devoted to establish and list elementary properties that are satisfied by the system of particles when a collapse takes place. In particular, we will concentrate on obtaining qualitative behaviours of the sequences of variables used to describe the system of inelastic particles.\\
The thorough discussion of the present section sets a mathematically rigorous framework to study inelastic collapses. To the best of our knowledge, this is the first study of this sort in the literature.

\subsection{Definition of the inelastic collapse}
\label{SSECTIO3.1DefinCollaInela}

\begin{defin}[Inelastic collapse]
\label{DEFINSS3.1Collapse_Inelas}
Let $r \in\ ]0,1[$ be a positive real number smaller than $1$, and let us consider a system of three particles \circled{0}, \circled{1} and \circled{2} evolving according to the $r$-inelastic hard sphere flow \eqref{EQUATSS2.1_Loi_de_Newton_}, \eqref{EQUATSS2.1VitesPost-Colli}, on a time interval $[0,\tau^*[$, with $\tau^* > 0$.\\
We say that the system undergoes an inelastic collapse at the time $\tau^*$, called the \emph{time of the collapse}, or the \emph{collapsing time}, if there exists an increasing sequence of positive times $(t_n)_{n\in\mathbb{N}}$, with $t_0 = 0$, such that $t_n \xrightarrow[n \rightarrow + \infty]{}\tau^*$, when the times $t_n$ correspond exactly to the times of collisions between the particles, and such that $\sup_{n}t_n$ is the only accumulation point of the sequence $\left(t_n\right)_{n\in\mathbb{N}}$.\\
In particular, for any $t \in \ ]t_n,t_{n+1}[$, the system of particles evolves according to the free flow.
\end{defin}

\begin{remar}
It is important here to notice that we considered only \emph{well-defined} trajectories of particles. In particular, we implicitly assume that no triple collision takes place on the time interval $[0,\tau^*[$, nor grazing collision. We also assumed that $\tau^*$ is the \emph{first} time of inelastic collapse, by requiring that $\tau^*$ is the only accumulation point of times of collision. Under the last assumption, requiring that the system does not present triple of grazing collisions does not harm: such events, obtained after a finite number of collisions, correspond to a set of initial data of measure zero.
\end{remar}

\subsection{Converging quantities and vanishing quantities}
\label{SSECTIO3.2ConvergingQuant}

We now dedicate ourselves to establish general properties, holding for all kinds of collapsing systems of particles.

\subsubsection{Convergence of the distances}

In what follows, it will be convenient to consider the time intervals between two collisions. In particular, we will denote $\tau_n = t_n - \tau_{n-1}$ the times between two consecutive collisions. With these notations, the times $\tau_n$ correspond to the time $\tau$ that is used in \eqref{EQUATSS2.2IterationSystm1}.\\
There is already a direct consequence of the definition, stated in the Proposition that follows.

\begin{propo}[Convergence of the times between two collisions]
\label{PROPOSS3.2ConveInterTemps}
Let $r \in\ ]0,1[$ be a positive number smaller than $1$, and let us consider a system of three particles \circled{0}, \circled{1} and \circled{2} evolving according to the $r$-inelastic hard sphere flow \eqref{EQUATSS2.1_Loi_de_Newton_}, \eqref{EQUATSS2.1VitesPost-Colli}, on a time interval $[0,\tau^*[$, and undergoing an inelastic collapse at time $\tau^* > 0$.\\
Then the sequence $\left( \tau_n \right)_n$ of the intervals $\tau_n = t_n - \tau_{n-1}$ between two consecutive collisions is summable, that is,  $\left( \tau_n \right)_n \in\ \ell^1$. In particular, we have:
\begin{align}
\tau_n \xrightarrow[n \rightarrow +\infty]{} 0.
\end{align}
\end{propo}

\noindent
Any collision is dissipating a positive amount of kinetic energy. Such a dissipation was quantified in Lemma \ref{LEMMESS2.1DissiEnergCinet}. In particular, the kinetic energy of the system remains bounded, as all the velocities of the particles. Therefore, since the positions of the particles at the times of the collisions $t_n$ are given by:
\begin{align*}
x_i(t_n) = x_i(0) + \sum_{k=1}^n (t_k-t_{k-1}) v_i(t_{k-1}),
\end{align*}
(for $i \in\ \{1,2,3\}$, corresponding to the index of the particle we are considering), we clearly see that
\begin{align*}
\left\vert x_i(t_p) - x_i(t_q) \right\vert \leq \sum_{k=p+1}^q \left\vert (t_k-t_{k-1}) v_i(t_{k-1}) \right\vert \leq \left( \sup_{n \in \mathbb{N}} \vert v_i(t_n) \vert \right) \underbrace{\sum_{n \geq p+1} \tau_n}_{ \xrightarrow[p \rightarrow +\infty]{} 0 }
\end{align*}
for $p<q$, providing that the sequence of positions of the particles are Cauchy sequences, hence converging at the time $\tau^*$ of the inelastic collapse. As an immediate consequence, the directions between the particles are also converging at the time $\tau^*$ of the inelastic collapse.\\
%On the other hand, the sequence of distances between the central particle and the spectator particle is also converging to zero. Indeed, when a particle is spectator during a collision, it will be colliding during the next collision, and the distance separating this particle from the central particle has to be travelled until this next collision. A priori, we cannot claim that the sequence of distances $\left( d_n \right)_{n \in \mathbb{N} }$ is decreasing, but it is bounded, since the sequences of the positions of the three particles are bounded. Therefore, we can consider
%\begin{align*}
%\limsup_{n \in \mathbb{N}} d_n < +\infty.
%\end{align*}
%If such a quantity is positive, say, equal to $\overline{d} > 0$, then there exists a subsequence $\left( \varphi(n) \right)_{n \in \mathbb{N}}$ of collisions such that $d_{\varphi(n)} \xrightarrow[n \rightarrow +\infty]{} \overline{d}$. However, this would lead to a contradiction, for this distance has to be travelled during time intervals $\tau_{\varphi(n)}$ that are smaller and smaller, converging to zero, with a velocity that remains bounded.
In summary, we can state the following result.

\begin{propo}[Convergence of the angular parameters $\omega$ and the distance between colliding particles]
\label{PROPOSS3.2ConveAngleDista}
Let $r \in\ ]0,1[$ be a positive real number smaller than $1$, and let us consider a system of three inelastic particles \circled{0}, \circled{1} and \circled{2} evolving according to the $r$-inelastic hard sphere flow \eqref{EQUATSS2.1_Loi_de_Newton_}, \eqref{EQUATSS2.1VitesPost-Colli}, on a time interval $[0,\tau^*[$, and undergoing an inelastic collapse at time $\tau^* > 0$, such that the particle \circled{0} is implied in infinitely many collisions with the two particles \circled{1} and \circled{2}.\\
Then the sequences of the angular parameters $\left(\omega_{1,n}\right)_{n\in\mathbb{N}}$ and $\left(\omega_{2,n}\right)_{n\in\mathbb{N}}$, defined as $\omega_{i,n} = x_i(t_n) - x_0(t_n)$ (for $i \in\ \{1,2\}$) and the sequence of the distances $\left( d_n \right)_{n \in\ \mathbb{N}}$ between the central particle \circled{0} and the spectator particle are converging, and we have:
\begin{align}
d_n \xrightarrow[n \rightarrow +\infty]{} 0.
\end{align}
The limits of the sequences $\left(\omega_{1,n}\right)_{n\in\mathbb{N}}$ and $\left(\omega_{2,n}\right)_{n\in\mathbb{N}}$ will be respectively denoted by $\overline{\omega}_1$ and $\overline{\omega}_2$.
\end{propo}

\begin{remar}
The assumption in the previous proposition concerning the particle \circled{0} %, which makes then this particle as a central particle,
is necessary in order to be sure that the two pairs \circled{0}-\circled{1} and \circled{0}-\circled{2} are involved in infinitely many collisions. Under this assumption, it is not possible to deduce, a priori, anything about the number of collisions between the two particles \circled{1} and \circled{2}.\\
Let us also note that, up to rename the particles, such an assumption is equivalent to assume that an infinite number of collisions take place: in this case at least two pairs are involved in infinitely many collisions, so that at least one particle is involved in infinitely many collisions.
\end{remar}
\noindent
We have then a clear geometrical image of the system at the final time $\tau^*$ of the inelastic collapse: the pairs of particles that underwent an infinite number of collisions are in contact.\\
It remains to study the limiting velocities at the time of collapse of such a system of particles.

\subsubsection{Convergence of the velocities}
\label{SSSCT3.2.2QuanConveGVites} % Quantités Convergentes Générales : cas des Vitesses

We can already establish, in full generality, the result that follows. Let us start with introducing a convenient description of the order of collisions for our purpose.\\
When the system of three inelastic particles \circled{0}, \circled{1} and \circled{2} experiences a collapse at time $\tau^* > 0$, infinitely many collisions take place, at the respective times $t_n$, such that $t_n$ converges towards $\tau^*$. Each of these collisions is either of type \circled{0}-\circled{1}, \circled{0}-\circled{2} or \circled{1}-\circled{2}. Therefore, considering any of the three pairs $(i,j) = (0,1), (0,2)$ or $(1,2)$, we can associate a strictly increasing function:
\begin{align*}
\varphi^{(i,j)}:\mathbb{N}^* \rightarrow \mathbb{N}^*
\end{align*}
(or $\varphi^{(i,j)}:\{1,\dots,N\} \rightarrow \mathbb{N}^*$ if the system experiences only $N$ collisions of type \emph{\circled{i}-\circled{j}}), where $\varphi^{(i,j)}(n)$ is the index of the $n$-th collision of type \emph{\circled{i}-\circled{j}}. Of course, the images of the functions $\varphi^{(i,j)}$, for the three pairs $(i,j)$, form a partition of $\mathbb{N}^*$. An important notion is the following.

\begin{defin}[Counting collision function, maximal gaps between collisions of type $(i,j)$]
\label{DEFINSS3.2FonctCompt_Gaps}
Let $r \in\ ]0,1[$ be a positive real number smaller than $1$, and let us consider a system of three inelastic particles \circled{0}, \circled{1} and \circled{2}  evolving according to the $r$-inelastic hard sphere flow \eqref{EQUATSS2.1_Loi_de_Newton_}, \eqref{EQUATSS2.1VitesPost-Colli}, on a time interval $[0,\tau^*[$, and undergoing an inelastic collapse at time $\tau^* > 0$.\\
For any of the pairs $(i,j) \in\ \left\{(0,1),(0,2),(1,2)\right\}$, we consider the function $\varphi^{(i,j)}:\mathbb{N}^* \rightarrow \mathbb{N}^*$ or $\varphi^{(i,j)}:\{1,\dots,N\} \rightarrow \mathbb{N}^*$, where $N$ is the number of collisions of type \circled{i}-\circled{j} experienced by the system of particles (with the convention that $\{1,\dots,N\} = \emptyset$ if $N=0$), and where $\varphi^{(i,j)}(n)$ is defined as:
\begin{align}
\varphi^{(i,j)}(n) = k \ \forall n\in\ \mathbb{N}^* (\text{or }n\in\ \{1,\dots,N\}),
\end{align}
where $k$ is the index (counting all the collisions of the three possible types) of the $n$-th collision of type \circled{i}-\circled{j}. The function $\varphi^{(i,j)}$ is called the \emph{counting collision function of type \circled{$i$}-\circled{$j$}}.\\
In particular, $\varphi^{(i,j)}$ is a strictly increasing function.\\
Let $\varphi^{(i,j)}:\mathbb{N}^* \rightarrow \mathbb{N}^*$ or $\varphi^{(i,j)}:\{1,\dots,N\} \rightarrow \mathbb{N}^*$ be the counting collision function of type \circled{i}-\circled{j} of the collapsing system of particles \circled{0}, \circled{1} and \circled{2}. We call the \emph{maximal gap between collisions of type $(i,j)$}, or the \emph{maximal gap of $\varphi^{(i,j)}$}, the positive number:
\begin{align}
\sup_{n \in \mathbb{N}^*} \left[ \varphi^{(i,j)}(n+1) - \varphi^{(i,j)}(n) \right] \text{  or  } \max_{n \in \{1,\dots,N-1\}} \left[ \varphi^{(i,j)}(n+1) - \varphi^{(i,j)}(n) \right],
\end{align}
that we will denote by $K_{(i,j)}$.\\
For a system of collapsing particles, we say that the system has a \emph{finite maximal gap of type $(i,j)$} if the counting collision function $\varphi^{(i,j)}$ has a finite maximal gap.\\
Finally, we say that the collapsing system of particles \emph{has finite maximal gaps between collisions} if all of its counting functions have a finite maximal gap.
\end{defin}
\noindent
For example, if we consider the sequence of collisions: \circled{0}-\circled{1}, \circled{1}-\circled{2}, \circled{0}-\circled{1}, \circled{0}-\circled{2}, \circled{0}-\circled{1}, \circled{1}-\circled{2}, \circled{0}-\circled{1}, \circled{1}-\circled{2}, \circled{0}-\circled{1}, \circled{0}-\circled{2}, \circled{0}-\circled{1}... and which goes on with infinitely many repetitions of the period \circled{0}-\circled{2}, \circled{0}-\circled{1}, the counting functions $\varphi^{(i,j)}$ are:
\begin{align*}
\varphi^{(0,1)}(k) = 1, 3, 5, 7, 9, 11,\dots \forall k \in \mathbb{N}^*,\hspace{3mm} \varphi^{(0,2)}(k) = 4, 10, 12, 14,\dots \forall k \in \mathbb{N}^*, \varphi^{(1,2)}(k) = 2, 6, 8 \text{ for } k = 1, 2 \text{ or } 3,
\end{align*}
and the maximal gaps, respectively, between the collisions of type $(0,1)$, $(0,2)$ and $(1,2)$, are $2$, $6$ and $4$.\\
In what follows, we will start with restricting ourselves to study collapsing systems of particles with finite maximal gaps. But prescribing assumptions on the maximal number of collisions of a certain type, we will actually show in the two next sections that any collapsing system of three particles has finite maximal gaps.

\begin{propo}[Summability and convergence of the normal components $\eta$ of the relative velocities]
\label{PROPOSS3.2SommaCompNormlV}
Let $r \in\ ]0,1[$ be a positive real number smaller than $1$, and let us consider a system of three inelastic particles \circled{0}, \circled{1} and \circled{2}  evolving according to the $r$-inelastic hard sphere flow \eqref{EQUATSS2.1_Loi_de_Newton_}, \eqref{EQUATSS2.1VitesPost-Colli}, on a time interval $[0,\tau^*[$, and undergoing an inelastic collapse at time $\tau^* > 0$. Let us also assume that, for some $(i,j) \in\ \left\{ (0,1),(0,2),(1,2)\right\}$, the pair of particles \circled{i}-\circled{j} collides infinitely many times, and that the system has a finite maximal gap of type $(i,j)$. Let us denote by $\eta_{(i,j),n}$ the normal component of the relative velocity between the particles \circled{i} and \circled{j} just before the $n$-th collision.\\
Then the sequences of normal components $\left( \eta_{(i,j),n} \right)_n$ of the relative velocities are square summable, that is $\left( \eta_{(i,j),n} \right)_n \in\ \ell^2$. In particular, for the pairs of particles \circled{i}-\circled{j} that collide infinitely many times, we have:
\begin{align}
\eta_{(i,j),n} \xrightarrow[n \rightarrow +\infty]{} 0.
\end{align}
\end{propo}

\begin{proof}
The result of Proposition \ref{PROPOSS3.2SommaCompNormlV} is a direct consequence of the quantification \eqref{EQUATSS2.1DissiEnergCinet} of the dissipation of the kinetic energy.\\
Indeed, let $(i_p,j_p)$ be the type of the $p$-th collision of the system. Since the system can dissipate only a finite amount of kinetic energy, the square of the normal component $\eta_{i_p,j_p,p}$ of the relative velocity between the particles \circled{$i_p$} and \circled{$j_p$} at the time $t_p$ of this $p$-th collision cannot be larger than any given constant, more than a finite number of times, and even more, the sequence of these squares has to be summable, that is:
\begin{align}
\label{EQUATSS3.2EtaCollisiSumma}
\sum_{p=1}^{+\infty} \vert \eta_{(i_p,j_p),p} \vert^2 < +\infty.
\end{align}
Now, the order of the collisions being not prescribed, it is not clear a priori how the normal component of the relative velocity of a prescribed pair $(i,j)$ is subject to the same regularity described by \eqref{EQUATSS3.2EtaCollisiSumma} for $\eta_{i_p,j_p}$. But for a prescribed pair $(i,j)$, either the $p$-th collision is such that the two particles \circled{$i$} and \circled{$j$} are involved in such a collision, or only one of these two particles is colliding, say particle \circled{$i$}, while particle \circled{$j$} is not colliding at time $t_p$.\\
In the first case (\circled{$i$} and \circled{$j$} colliding), we have:
\begin{align}
\eta_{(i,j),p} = \eta_{(i_p,j_p),p},
\end{align}
where $\left( \eta_{(i_p,j_p),p} \right)_p$ is a square summable sequence.\\
In the second case (\circled{$i$} colliding, \circled{$j$} not colliding), we have:
\begin{align*}
v_j'-v_i' = v_j(t_p^-)-v_i(t_p^-) + \frac{(1+r)}{2}\left(v_i(t_p^-)-v_k(t_p^-)\right)\cdot\omega_{(i,k),p}\omega_{(i,k),p},
\end{align*}
where $k \neq i,j$ is the second particle colliding, $v_i'$ and $v_j'$ are the post-collisional velocities of the particles \circled{$i$} and \circled{$j$}, that is, the velocities just after the $p$-th collision, and $\omega_{(i,k),p}=x_i(t_p)-x_k(t_p)$ (of norm $1$ at $t_p$, since \circled{$i$} and \circled{$k$} are in contact at $t_p$). Therefore, we find:
\begin{align}
\eta_{(i,k),p}' = \eta_{(i,k),p+1} = \eta_{(i,k),p} + \frac{(1+r)}{2} \eta_{(i_p,j_p),p} \omega_{(i_p,j_p),p} \cdot \omega_{(i,k),p},
\end{align}
so, in all the cases, at time $t_p^+$, either we have:
\begin{align}
\label{EQUATSS3.2eta_ij_CasColli}
\eta_{(i,j),n+1} = a_n
\end{align}
or we have
\begin{align}
\label{EQUATSS3.2eta_ij_CasNColl}
\eta_{(i,j),n+1} = \eta_{(i,j),n} + \alpha_n a_n,
\end{align}
where $\left(a_n\right)_n$ and $\left(\alpha_n\right)_n$ are two sequences of real numbers such that
\begin{align}
\sum_{n \geq 0} a_n^2 < +\infty \hspace{5mm} \text{and} \hspace{5mm} 0 \leq \alpha_n \leq \alpha_0 < 1 \hspace{3mm} \forall n \in\ \mathbb{N}^*.
\end{align}
Now, by definition of the counting collision function (see Definition \ref{DEFINSS3.2FonctCompt_Gaps} above), \eqref{EQUATSS3.2eta_ij_CasColli} holds if and only if $p = \varphi^{(i,j)}(n)$, and by assumption the system has a finite maximal gap of type $(i,j)$, that is, there exists $K_{(i,j)} \in\ \mathbb{N}^*$ such that
\begin{align*}
0 \leq \varphi^{(i,j)}(n+1) - \varphi^{(i,j)}(n) \leq K_{(i,j)}.
\end{align*}
Therefore, we can write:
\begin{align*}
\sum_{n\geq0} \eta_{(i,j),n}^2 &= \sum_{n\geq0} \left[ \eta_{(i,j),\varphi^{(i,j)}(n)}^2 + \hspace{-4mm} \sum_{k = \varphi^{(i,j)}(n)+1}^{\varphi^{(i,j)}(n+1)-1} \hspace{-6mm} \eta_{(i,j),k}^2 \right] + \hspace{-4mm} \sum_{n=0}^{\varphi^{(i,j)}(0)-1} \hspace{-4mm} \eta_{(i,j),n}^2 \\
&= \sum_{n\geq0} \left[ a^2_{\varphi^{(i,j)}(n)-1} + \hspace{-4mm} \sum_{k = 1}^{\varphi^{(i,j)}(n+1)-\varphi^{(i,j)}(n)-1} \hspace{-6mm} a^2_{\varphi^{(i,j)}(n)-1}(1+k\alpha_n)^2 \right] + \hspace{-4mm} \sum_{n=0}^{\varphi^{(i,j)}(0)-1} \hspace{-4mm} \eta_{(i,j),n}^2 \\
&\leq \sum_{n\geq0} \left( \sum_{k=0}^{K_{i,j}-1} (1+k \alpha_n)^2 \right) a^2_{\varphi^{(i,j)}(n)-1} \leq \left( \sum_{k=0}^{K_{i,j}-1} (1+k \alpha_n)^2 \right) \left( \sum_{n\geq0} a^2_n \right) < +\infty,
\end{align*}
which concludes the proof of the proposition.
\end{proof}

\subsection{Asymptotic behaviour of the variables}
\label{SSECTI03.3CompoAsympVaria}

Let us carry on the study of general collapsing systems of three particles. The purpose of this section is to compare the asymptotic behaviours of the different variables of the dynamical system \eqref{EQUATSS2.2IterationSystm1} when an inelastic collapse takes place. These additional results will be useful in order to identify the leading order terms in the (quite complicated) expression \eqref{EQUATSS2.2IterationSystm1} of the dynamical system.\\
Let us emphasize that in this section neither, we will not assume any order on the sequence of the collisions.\\
\newline
Around the regime of the collapse, the variables $\tau$ and $\eta_i$ (for the pairs of particles involved in infinitely many collisions) are vanishing. We will then compare the decay of the vanishing variables, namely:
\begin{align*}
d,d', \eta_1' \text{   and   } \eta_2', \left(\omega_1'-\omega_1\right),\left(\omega_2'-\omega_2\right),\left(\left(W_1^\perp\right)' - W_1^\perp\right) \text{   and   } \left(\left(W_2^\perp\right)' - W_2^\perp\right)
\end{align*}
with the reference variables $\tau$ and $\eta$.
%Finally, we can also study in more details the rate of convergence of the normal components $\eta_i$. To do so, we will finally consider the differences
%\begin{align*}
%\eta_1' - \eta_1 + \frac{(1+r)}{2}\cos\overline{\theta} \eta_2 \text{   and   } \eta_2'-r\eta_2,
%\end{align*}
%where $\overline{\theta}$ is the limiting angle between the particles \circled{1} and \circled{2}, measured from the particle \circled{0}.
Let us start with introducing some useful notations.

\subsubsection{Notations}

In the rest of this section, concerning a variable $y$, we will denote by
\begin{align*}
y = \mathcal{O}\left(x_1,\dots,x_n\right)
\end{align*}
when this variable $y$ is bounded by the $n$ other variables $x_1,\dots,x_n$, that is, if there exist $n$ positive constants $C_1,\dots,C_n > 0$ such that $
\vert y \vert \leq C_1 \vert x_1 \vert + \dots C_n \vert x_n \vert$. If the two variables $y$ and $z$ are bounded by each other, that is, if there exist two positive constants $C_1,C_2 > 0$ such that $\vert y \vert \leq C_1 \vert z \vert \hspace{5mm}$ and $\vert z \vert \leq C_2 \vert y \vert$, such a situation will be denoted by
\begin{align*}
y \lessgtr z.
\end{align*}
In the same way, we will denote by
\begin{align*}
y = o\left(x_1,\dots,x_n\right)
\end{align*}
when the variable $y$ is negligible with respect to the $n$ variables $x_1,\dots,x_n$, that is, if there exist $n$ functions $\varphi_1,\dots,\varphi_n:\mathbb{R}_+\rightarrow \mathbb{R}_+$ vanishing at $0$ and such that $\vert y \vert \leq \varphi_1\left(\vert x_1 \vert\right) + \dots + \varphi_n\left(\vert x_n \vert\right)$.

\subsubsection{Comparison of the variables I: the general case}
\label{SSSCT3.3.2CompaVariaGener}

In the case of a general collapsing system of three collapsing, we have the following result.

\begin{propo}[Asymptotic comparison of the variables in the collapsing regime]
\label{PROPOSS3.3AsympCompa__I__}
Let $r \in\ ]0,1[$ be a positive real number smaller than $1$, and let us consider a system of three inelastic particles \circled{0}, \circled{1} and \circled{2} evolving according to the $r$-inelastic hard sphere flow \eqref{EQUATSS2.1_Loi_de_Newton_}, \eqref{EQUATSS2.1VitesPost-Colli}, on a time interval $[0,\tau^*[$, and undergoing an inelastic collapse at time $\tau^* > 0$. Let us assume that the pair of consecutive collisions \circled{0}-\circled{1}, \circled{0}-\circled{2} takes places infinitely many times, and that the system has finite maximal gaps of types $(0,1)$ and $(0,2)$. Let us denote by $t_{\varphi(n)}$ the different times when a collision of type \circled{0}-\circled{1} takes place, and when such a collision is immediately followed at times $t_{\varphi(n)+1} = t_{\varphi(n)}+\tau_n$ by a collision of type \circled{0}-\circled{2}. At time $t_{\varphi(n)}$, let us denote by $d_n$ the distance between the particles \circled{0} and \circled{2}, by $\omega_{1,n}$ and $\omega_{2,n}$ the respective angular parameters between the pairs \circled{0} and \circled{1}, and \circled{0} and \circled{2} respectively, by $\eta_{1,n}$ and $\eta_{2,n}$ the respective normal components of the relative velocities, and by $W_{1,n}^\perp$ and $W_{2,n}^\perp$ the respective tangential components of the relative components of the pairs \circled{0} and \circled{1}, and \circled{0} and \circled{2} respectively. At time $t_{\varphi(n)+1} = t_{\varphi(n)}+\tau_n$, let us denote by $d'_n$ the distance between the particles \circled{0} and \circled{1}, and by $\omega_{1,n}'$ and $\omega_{2,n}'$ the angular parameters, by $\eta_{1,n}'$ and $\eta_{2,n}'$ the normal components, and by $(W_{1,n}^\perp)'$ and $(W_{2,n}^\perp)'$ the respective tangential components of the relative components of the pairs \circled{0} and \circled{1}, and \circled{0} and \circled{2} respectively.\\
Then, we have the following asymptotic relations, as $n \rightarrow +\infty$:
\begin{align}
\label{EQUATSS3.3PROPOComp1Tau_d}
\tau_n \lessgtr \frac{d_n}{(-\eta_{2,n})}, \hspace{0.5mm} \text{which implies:} \hspace{2mm} d_n = \mathcal{O}\left(\tau_n (-\eta_{2,n}) \right), \hspace{0.5mm} \text{and in particular:} \hspace{2mm} d_n = o(\tau_n) \text{   and   } d_n = o(-\eta_{2,n}),
\end{align}
\begin{align}
\label{EQUATSS3.3PROPOComp1_d'__}
d'_n = \mathcal{O}\left(\eta_{1,n}\tau_n,\tau_n^2 \vert W_{1,n}^\perp \vert^2 \right) = \mathcal{O}\left(\tau_n\left(\eta_{1,n} + \tau_n\right)\right) = o(\tau_n),
\end{align}
\begin{align}
\label{EQUATSS3.3PROPOComp1_eta'}
\vert \eta_{1,n}' \vert = \mathcal{O}\left(\eta_{1,n},-\eta_{2,n},\tau_n \vert W_{1,n}^\perp \vert^2 \right), \hspace{1mm} \text{and} \hspace{5mm} \eta_{2,n}' = \mathcal{O}(-\eta_{2,n}),
\end{align}
\begin{align}
\label{EQUATSS3.3PROPOComp1omga'}
\left\vert \omega_{1,n}'-\omega_{1,n} \right\vert = \mathcal{O} \left( \tau_n \eta_{1,n},\tau_n \vert W_{1,n}^\perp \vert \right), \hspace{1mm} \text{and} \hspace{5mm} \left\vert \omega_{2,n}'-\omega_{2,n} \right\vert = \mathcal{O}\left(\tau_n(-\eta_{2,n}),\tau_n \vert W_{2,n}^\perp \vert \right)
\end{align}
and finally:
\begin{align}
\label{EQUATSS3.3PROPOComp1PerpD} %% composante Perpendiculaire, Différences
\left\vert (W_{1,n}^\perp)'-W_{1,n}^\perp \right\vert &= \mathcal{O} \left( \tau_n \eta_{1,n}^2, \tau_n^2\eta_{1,n} \vert W_{1,n}^\perp \vert, \tau_n \vert W_{1,n}^\perp \vert^2, -\eta_{2,n} \right), \nonumber\\
& \hspace{30mm} \text{and} \hspace{2mm} \left\vert (W_{2,n}^\perp)'-W_{2,n}^\perp \right\vert = \mathcal{O}\left(\tau_n\eta^2_{2,n},\tau_n \vert W_{2,n}^\perp \vert^2\right).
\end{align}
In addition, the relations:
\begin{align}
\label{EQUATSS3.3PROPOComp1PerpG} %% Comparaison I, composante Perpendiculaire assez Grande
\vert W_{2,n}^\perp \vert = o(-\eta_{2,n}) \hspace{2mm} \text{or} \hspace{2mm} \vert W_{2,n}^\perp \vert = \mathcal{O} (-\eta_{2,n})
\end{align}
cannot hold.
\end{propo}

\begin{proof}
All the results we will obtain are consequences of the equations \eqref{EQUATSS2.2IterationSystm1}. % defining the dynamical system composed with the three inelastic particles
Let us observe that the assumption stating that the system has finite maximal gaps of types $(0,1)$ and $(0,2)$ ensures that the two sequences $\left(\eta_{1,n}\right)_n$ and $\left(\eta_{2,n}\right)_n$ are both vanishing as $n$ goes to infinity.\\
Let us start with the general conditions relative to the collision time $\tau_n$. First, this time $\tau_n$ has to vanish at the limiting time of the collapse, so that either
\begin{align}
\label{EQUATSS3.3DisjoCVTauEta/W}
\frac{(-\eta_{2,n})}{\vert W_{2,n} \vert^2} = \frac{(-\eta_{2,n})}{\eta_{2,n}^2 + \vert W_{2,n}^\perp \vert^2} {\xrightarrow[n \rightarrow +\infty]{}0},
\end{align}
or
\begin{align}
\label{EQUATSS3.3DisjoCVTau_Zeta}
1 - \sqrt{1-\zeta_n} {\xrightarrow[n \rightarrow +\infty]{} 0},
\end{align}
that is $\zeta_n {\xrightarrow[n \rightarrow +\infty]{} 0}$, where $\zeta_n = \displaystyle{\frac{d_n\left(2+d_n\right) \left\vert W_{2,n} \right\vert^2}{\left(1+d_n\right)^2\eta_{2,n}^2}}\cdotp$
If \eqref{EQUATSS3.3DisjoCVTauEta/W} holds, we can deduce that \eqref{EQUATSS3.3PROPOComp1PerpG} cannot hold. Besides, $\zeta_n$ has always to be smaller than $1$, so for $n$ large enough $d_n \leq \eta_{2,n}^2/(2 \vert W_{2,n} \vert^2)$. Since $d_n$ vanishes asymptotically, we recover that $\vert W_{2,n}^\perp \vert = o(-\eta_{2,n})$ or $\vert W_{2,n}^\perp \vert = \mathcal{O}(-\eta_{2,n})$ cannot hold. The estimate \eqref{EQUATSS3.3PROPOComp1PerpG} is proved in full generality.\\
Since for an argument $x \in\ [0,1]$ we have always $
x/2 \leq 1 - \sqrt{1-x} \leq x$, we deduce that we have
\begin{align}
\label{EQUATSS3.3AsympCompa_Tau_Zeta_}
\frac{(1+d_n)(-\eta_{2,n})}{\vert W_{2,n} \vert^2} \cdot \frac{d_n(d_n+2) \vert W_{2,n} \vert^2}{2(1+d_n)^2 \eta_{2,n}^2} = \frac{d_n(2+d_n)}{2(1+d_n)(-\eta_{2,n})} \leq \tau_n \leq \frac{d_n(2+d_n)}{(1+d_n)(-\eta_{2,n})} \cdotp
\end{align}
In particular, in the regime of the collapse we deduce \eqref{EQUATSS3.3PROPOComp1Tau_d}.\\
We can now turn to the consequence of the evolution equations \eqref{EQUATSS2.2IterationSystm1}.
Concerning the distance $d'_n$ between \circled{0} and \circled{1} at the time $t_{\varphi(n)}+\tau_n$ of collision between \circled{0} and \circled{2}, we have already:
\begin{align*}
d'_n = \eta_{1,n} \tau_n + \frac{1}{2} \vert W_{1,n} \vert^2 \tau_n^2 + \mathcal{O}\left( \left(\eta_{1,n}\tau_n + \vert W_{1,n} \vert^2\tau_n^2\right)^2 \right).
\end{align*}
As a direct consequence of the fact that $\vert W_1 \vert$ remains bounded and that $\eta_1$ is vanishing, we obtain \eqref{EQUATSS3.3PROPOComp1_d'__}.\\
Concerning now the normal components of the relative velocities, we find:
\begin{align}
\vert \eta_{1,n}' \vert = \mathcal{O} \left( \eta_{1,n},-\eta_{2,n},\tau_n\vert W_{1,n} \vert^2,\tau_n\vert W_{2,n} \vert^2 \right)
\end{align}
and
\begin{align}
\vert \eta_{2,n}' \vert = \mathcal{O} \left( -\eta_{2,n},\tau_n\vert W_{2,n} \vert^2 \right).
\end{align}
Keeping in mind the expression \eqref{EQUATSS2.2TempsColliTauV3} of the collision time $\tau_n$, we can replace $\tau_n \vert W_{2,n} \vert^2$ by $\mathcal{O}\left( (-\eta_{2,n}) \zeta_n \right)$ and then by $\mathcal{O}(\eta_{2,n})$, because $\zeta_n$ is always smaller than $1$, so that we can rewrite:
\begin{align}
\vert \eta_{1,n}' \vert = \mathcal{O} \left( \eta_{1,n},-\eta_{2,n},\tau_n\vert W_{1,n} \vert^2 \right)
\end{align}
and
\begin{align}
\vert \eta_{2,n}' \vert = \mathcal{O} \left( -\eta_{2,n} \right),
\end{align}
so that \eqref{EQUATSS3.3PROPOComp1_eta'} is proved.\\
Let us now turn to the difference of the converging variables. Concerning the first angular parameters, we have $\omega_{1,n}' - \omega_{1,n} = -d_n'\omega_{1,n}/(1+d_n')  + \tau_n W_{1,n}/(1+d_n')$ so that $\omega_{1,n}'-\omega_{1,n} = \mathcal{O} \left( d_n',\tau_n W_{1,n} \right)$, which, combined with \eqref{EQUATSS3.3PROPOComp1_d'__}, provides
\begin{align*}
\omega_{1,n}'-\omega_{1,n} = \mathcal{O}\left( \tau_n \eta_{1,n},\tau_n \vert W_{1,n}^\perp \vert \right).
\end{align*}
In the same way, for the other angular parameter we find $\omega_{2,n}'-\omega_{2,n} = d_n \omega_{2,n} + \tau_n W_{2,n}$, so that
\begin{align}
\omega_{2,n}'-\omega_{2,n} = \mathcal{O}\left(d_n,\vert \tau_n W_{2,n} \vert \right) = \mathcal{O}\left(\tau_n(-\eta_{2,n}),\tau_n \vert W_{2,n}^\perp \vert \right),
\end{align}
and \eqref{EQUATSS3.3PROPOComp1omga'} is proved.\\
Concerning the differences between the consecutive tangential components of the relative velocities, we obtain:
\begin{align*}
\left(W_{1,n}^\perp\right)' - W_{1,n}^\perp &= \eta_{1,n} \omega_{1,n} - \frac{(\eta_{1,n}+\tau_n \vert W_{1,n} \vert^2)}{(1+d_n')}\omega_{1,n}' \\
&\hspace{5mm}+ \frac{(1+r)}{2}\left((1+d_n)\eta_{2,n} + \tau_n \vert W_{2,n} \vert^2\right) \left[(\omega_{1,n}'\cdot\omega_{2,n}')\omega_{1,n}'-\omega_{2,n}'\right]
\end{align*}
so that $\left(W_{1,n}^\perp\right)' - W_{1,n}^\perp = \mathcal{O}\left( \eta_{1,n}(\omega_{1,n}'-\omega_{1,n}),\eta_{1,n} d_n',\tau_n \vert W_{1,n} \vert^2,-\eta_{2,n},\tau_n \vert W_{2,n} \vert^2 \right)$, which can be simplified using \eqref{EQUATSS3.3PROPOComp1omga'}, \eqref{EQUATSS3.3PROPOComp1_d'__} and the explicit expression \eqref{EQUATSS2.2TempsColliTauV3} of $\tau_n$ as
\begin{align}
\left(W_{1,n}^\perp\right)' - W_{1,n}^\perp &= \mathcal{O}\left( \tau_n \eta_{1,n}^2, \tau_n^2 \eta_{1,n} \vert W_{1,n}^\perp \vert, \tau_n^2 \eta_{1,n} \vert W_{1,n}^\perp \vert^2, \tau_n \vert W_{1,n} \vert^2, -\eta_{2,n} \right) \nonumber \\
&= \mathcal{O} \left( \tau_n \eta_{1,n}^2, \tau_n^2\eta_{1,n} \vert W_{1,n}^\perp \vert, \tau_n \vert W_{1,n}^\perp \vert^2, -\eta_{2,n} \right),
\end{align}
and from $\left( W_{2,n}^\perp \right)' - W_{2,n}^\perp = \eta_{2,n}\omega_{2,n} - \left( (1+d_n)\eta_{2,n} + \tau_n \vert W_{2,n} \vert^2 \right) \left( (1+d_n)\omega_{2,n} + \tau_n W_{2,n} \right)$ we find
\begin{align}
\left( W_{2,n}^\perp \right)' - W_{2,n}^\perp = \mathcal{O} \left( d_n(-\eta_{2,n}),\tau_n \vert W_{2,n} \vert^2 \right) = \mathcal{O} \left( \tau_n\eta_{2,n}^2,\tau_n \vert W_{2,n}^\perp \vert^2 \right),
\end{align}
which completes the proof of \eqref{EQUATSS3.3PROPOComp1PerpD}, and of Proposition \ref{PROPOSS3.3AsympCompa__I__}.
\end{proof}

\begin{remar}
Let us observe that we could also study in the same way the second order terms of the normal components of the relative velocities. To be more explicit, we have:
\begin{align*}
\eta_{1,n}' - \eta_{1,n} - \frac{(1+r)}{2} \cos\overline{\theta}\, \eta_{2,n} = \mathcal{O} \left( d_n'\eta_{1,n},\tau_n \vert W_{1,n} \vert^2, d_n\eta_{2,n}, \tau_n \vert W_{2,n} \vert^2 \right) = \mathcal{O} \left( \tau_n \eta_{1,n}^2, \tau_n \vert W_{1,n}^\perp \vert^2, -\eta_{2,n} \right)
\end{align*}
where $\overline{\theta}$ is the limiting angle between the particles \circled{1} and \circled{2}, measured from the particle \circled{0}, and
\begin{align*}
\eta_{2,n}' - r\eta_{2,n} = \mathcal{O} \left( d_n\eta_{2,n},\tau_n \vert W_{2,n} \vert^2 \right).
\end{align*}
\end{remar}
\noindent
As a general comment concerning the estimates obtained in Proposition \ref{PROPOSS3.3AsympCompa__I__}, let us emphasize that all the variables are bounded only in terms of the normal components $\eta_1$ and $\eta_2$, and the time of collision $\tau$. Therefore, understanding the asymptotic behaviour of these last variables will enable to understand completely the full dynamical system. On the other hand, it does not seem possible to compare, in full generality, the asymptotic behaviour of the normal components $\eta$ with the time of collision $\tau$. In other words, a natural problem appears: we need to compare the asymptotics of the variables $\eta$ and $\tau$, that govern completely the evolution equations \eqref{EQUATSS2.2IterationSystm1}.

\subsubsection{Comparison of the variables II: the generic case, with non vanishing tangential velocities}

Let us now revisit the inequalities comparing the asymptotics of the variables obtained in Proposition \ref{PROPOSS3.3AsympCompa__I__}, considering that none of the tangential velocities vanish asymptotically. To be more accurate, we will assume that the norms of the relative velocities are bounded from below by a positive constant, say $\overline{w} > 0$, after a sufficiently large number of collisions.\\
A first consequence of the assumption $\vert W_{2,n}^\perp \vert \geq \overline{w} > 0$, and perhaps the most important, is that in this case the Zhou-Kadanoff parameter $\zeta_n$ satisfies
\begin{align}
\zeta_n \lessgtr \frac{d_n}{\eta_{2,n}^2} \cdotp
\end{align}
Therefore, since in order to have a collapse this parameter has always to remain below $1$, we deduce that
\begin{align}
\label{EQUATSS3.3CompaAsymp_d__3}
d_n = \mathcal{O}(\eta_{2,n}^2),
\end{align}
which refines the estimate $d_n = o\left(-\eta_{2,n}\right)$ of \eqref{EQUATSS3.3PROPOComp1Tau_d}. Now, considering again the formula \eqref{EQUATSS2.2TempsColliTauV3} of the time of collision $\tau_n$, we find:
\begin{align}
\label{EQUATSS3.3CompaAsympTauV2}
\tau_n \leq \frac{2d_n}{(-\eta_{2,n})} = \underbrace{\frac{2d_n}{\eta_{2,n}^2}}_{\text{bounded}} \cdot (-\eta_{2,n}),
\end{align}
so that in particular we have the important estimate:
\begin{align}
\tau_n = \mathcal{O}(\eta_{2,n}).
\end{align}
As a consequence, in the ``generic'' case $\vert W_{1,n} \vert, \vert W_{2,n} \vert \geq \overline{w} > 0$, we can assert that the time of collision $\tau_n$ is bounded by the normal components $\eta_n$, and therefore, when the relative velocities do not vanish at the final time of the collapse, only the normal components of the relative velocities govern the evolution of the dynamical system \eqref{EQUATSS2.2IterationSystm1}.

\begin{propo}[Asymptotic comparison of the variables in the collapsing regime, assuming that the relative velocities do not vanish]
\label{PROPOSS3.3AsympCompa__II_}
Let $r \in\ ]0,1[$ be a positive real number smaller than $1$, and let us consider a system of three inelastic particles \circled{0}, \circled{1} and \circled{2} evolving according to the $r$-inelastic hard sphere flow \eqref{EQUATSS2.1_Loi_de_Newton_}, \eqref{EQUATSS2.1VitesPost-Colli}, on a time interval $[0,\tau^*[$, and undergoing an inelastic collapse at time $\tau^* > 0$. Let us assume that the pair of consecutive collisions \circled{0}-\circled{1}, \circled{0}-\circled{2} takes places infinitely many times, and that the system has finite maximal gaps of types $(0,1)$ and $(0,2)$. Let us assume also that there exists $\overline{w}>0$ such that $\vert W_{2,n} \vert \geq \overline{w}$ for all $n$ large enough.\\
Then, with the same notations as in Proposition \ref{PROPOSS3.3AsympCompa__I__}, we have the following asymptotic relations, as $n \rightarrow +\infty$:
\begin{align}
\label{EQUATSS3.3PROPOComp2Tau_d}
\tau_n = \mathcal{O}(-\eta_{2,n}), \hspace{3mm} \text{and} \hspace{5mm} d_n = \mathcal{O}(\eta_{2,n}^2),
\end{align}
\begin{align}
\label{EQUATSS3.3PROPOComp2_d'__}
d'_n = \mathcal{O}\left(\eta_{1,n}\tau_n\right),
\end{align}
\begin{align}
\label{EQUATSS3.3PROPOComp2_eta'}
\vert \eta_{1,n}' \vert = \mathcal{O}\left(\eta_{1,n},-\eta_{2,n}\right), \hspace{3mm} \text{and} \hspace{5mm} \eta_{2,n}' = \mathcal{O}(-\eta_{2,n}),
\end{align}
\begin{align}
\label{EQUATSS3.3PROPOComp2omga'}
\left\vert \omega_{1,n}'-\omega_{1,n} \right\vert = \mathcal{O} \left( \tau_n  \right), \hspace{3mm} \text{and} \hspace{5mm} \left\vert \omega_{2,n}'-\omega_{2,n} \right\vert = \mathcal{O}\left(\tau_n\right)
\end{align}
and finally:
\begin{align}
\label{EQUATSS3.3PROPOComp2PerpD} %% composante Perpendiculaire, Différences
\left\vert (W_{1,n}^\perp)'-W_{1,n}^\perp \right\vert &= \mathcal{O} \left(-\eta_{2,n} \right), \hspace{3mm} \text{and} \hspace{5mm} \left\vert (W_{2,n}^\perp)'-W_{2,n}^\perp \right\vert = \mathcal{O}\left(\tau_n\right).
\end{align}
\end{propo}

\begin{remar} Even under the assumption on the positive bound from below for the norm of the relative velocities, it seems that one cannot conclude a priori that $\tau$ will be eventually negligible with respect to the normal components: it might perfectly remain of the same order asymptotically.\\
Note also that, in order to establish the estimates of Proposition \ref{PROPOSS3.3AsympCompa__II_}, we needed \emph{only} the information that the \emph{second} relative velocity $\vert W_2 \vert$ does not vanish, which is not surprising, since the equations \eqref{EQUATSS2.2IterationSystm1} are describing a collision of type \circled{0}-\circled{2}.
\end{remar}

\subsubsection{Comparison of the variables III: the consequences of the regime $\tau \ll \eta$}
\label{SSSCT3.3.2ConsqRegimTau<E}

In this section, let us make a formal observation concerning a further simplification of the full dynamical dynamical. We observed in Proposition \ref{PROPOSS3.3AsympCompa__II_} that the regime $\tau = \mathcal{O}(\eta)$ provides much simpler estimates on the variables of the dynamical system \eqref{EQUATSS2.2IterationSystm1} than in the general case, addressed in Proposition \ref{PROPOSS3.3AsympCompa__I__}. In this section, let us assume $\tau \ll \eta$. Actually, this setting is also the one studied by Zhou and Kadanoff in \cite{ZhKa996}. These authors call this regime the ``flat surface approximation''.\\
So, considering \eqref{EQUATSS2.2IterationSystm1}, if the time $\tau$ of the next collision becomes small with respect to the normal components $\eta_1$ and $\eta_2$ of the relative velocities, we find:
\begin{align*}
\eta_1' \simeq \frac{1}{(1+d')} \eta_1 - \frac{(1+r)}{2}\left( \omega_1' \cdot \omega_2' \right) (1+d)\eta_2,
\end{align*}
neglecting in a first time the terms involving $\tau$ in the sums $\eta_1+ \tau \vert W_1 \vert^2$ and $(1+d)\eta_2 + \tau \vert W_2 \vert^2$ (of course, the norms $\vert W_1 \vert$ and $\vert W_2 \vert$ remain bounded along the evolution of the system), and then keeping only the leading order terms (since $d$ and $d'$ converge to zero in the case of a collapse):
\begin{align*}
\eta_1' \simeq \eta_1 - \frac{(1+r)}{2}\left( \omega_1' \cdot \omega_2' \right) \eta_2,
\end{align*}
and, in the same way, for the second relative velocity:
\begin{align*}
\eta_2' \simeq -r \eta_2.
\end{align*}
%We observe in particular the dissipation of the normal component of the relative velocity between the particles \circled{0} and \circled{2}, as the main effect of the model, and in the same time, we recover indeed that pre-collisional velocities are transformed into post-collisional velocities by the collision process (according to the change of sign between $\eta_2$ and $\eta_2'$).
Considering only the leading order for the cosine $\left( \omega_1' \cdot \omega_2' \right)$, we find in the end:

\begin{align}
\label{SECT6EquatSysteApproxEta_}
\left\{
\begin{array}{rl}
\eta_1' &\simeq \displaystyle{\eta_1 - \frac{(1+r)}{2} \cos\overline{\theta} \eta_2}, \\
\eta_2' &\simeq -r \eta_2,
\end{array}
\right.
\end{align}
where $\overline{\theta}$ corresponds to the limiting angle between the two particles \circled{1} and \circled{2} at the final time of the collapse.\\
In particular, the evolution of the normal components $\eta_1$ and $\eta_2$ are completely determined without using any other variables. We recovered formally, the system studied by Zhou and Kadanoff in \cite{ZhKa996}. In turn, since the evolution of all the other variables are driven by the normal components, we can deduce the complete evolution of the dynamical system \eqref{EQUATSS2.2IterationSystm1}.

\subsubsection{Comparison of the variables IV: sufficient condition for the regime $\tau \ll \eta$}
\label{SSSCT3.3.5CondiSuffiT<<Et}

We saw in the previous section that the idealized regime $\tau \ll \eta$ enables to simplify very much the full, complicated, dynamical system \eqref{EQUATSS2.2IterationSystm1}. Therefore, a natural and important question consists in determining when such a regime takes place. It turns out that this is the main difficulty.\\
Our main observations concerning the expression \eqref{EQUATSS2.2TempsColliTauV3} of $\tau$ are the following: first, in the collapse regime, since in particular $d$ vanishes, we can write for the Zhou-Kadanoff parameter:
\begin{align*}
\zeta = \frac{d(2+d)\vert W_2 \vert^2}{(1+d)^2 \eta_2^2} \simeq 2 \frac{d \vert W_2 \vert^2}{\eta_2^2} \cdotp
\end{align*}
Second, and it is the crucial observation, when the ZK-parameter $\zeta$ is small, the expression of the collision time $\tau$ becomes:
\begin{align}
\label{EQUATSS3.3CompaAsympTauV3}
\tau \sim \frac{(-\eta_2)}{\vert W_2 \vert^2} \left( \frac{\zeta}{2} \right) \sim \left( \frac{(-\eta_2)}{\vert W_2 \vert^2} \right) \frac{ \left( \displaystyle{ 2 \frac{d \vert W_2 \vert^2}{\eta_2^2} } \right) }{2} = \frac{d}{(-\eta_2)} \cdotp
\end{align}
Note that we made no assumption concerning bounds from below on $\vert W_2 \vert$ while obtaining \eqref{EQUATSS3.3CompaAsympTauV3}. The relation \eqref{EQUATSS3.3CompaAsympTauV3} refines the estimate \eqref{EQUATSS3.3PROPOComp1Tau_d} obtained in Proposition \ref{PROPOSS3.3AsympCompa__I__} on the one hand, and \eqref{EQUATSS3.3CompaAsympTauV2} on the other hand, for this time we obtained better than an upper bound on $\tau$. In the regime $\zeta \rightarrow 0$, the expression of $\tau$ takes then a particularly simple (approximated) form. Moreover, assuming that the tangential velocities do not vanish at the time of the collapse, that is, assuming that $\vert W_2 \vert^2$ converges towards a strictly positive limit, we see that the assumption $\zeta \rightarrow 0$ is equivalent to $d \ll \eta_2^2$. But then, in that case, since we found $\tau \simeq d/(-\eta_2)$, we deduce that:
\begin{align*}
\tau \sim \frac{d}{(-\eta_2)} = (-\eta_2) \frac{d}{\eta_2^2} \sim C(-\eta_2) \zeta,
\end{align*}
where $C$ is a positive constant, which implies that $\tau$ is negligible with respect to $\eta_2$ in the regime $\zeta \rightarrow 0$. Therefore, the regime $\zeta \rightarrow 0$ implies the regime $\tau \ll \eta$ briefly studied in the last section. It is now clear that the estimate of the asymptotic behaviour of the quantity $\zeta$ is crucial in order to understand the dynamics of the particles. Then, let us devote a definition to this crucial regime.

\begin{defin}[Zhou-Kadanoff regime]
Let $r \in\ ]0,1[$ be a positive real number smaller than $1$, and let us consider a system of three inelastic particles \circled{0}, \circled{1} and \circled{2} evolving according to the $r$-inelastic hard sphere flow \eqref{EQUATSS2.1_Loi_de_Newton_}, \eqref{EQUATSS2.1VitesPost-Colli}, on a time interval $[0,\tau^*[$, and undergoing an inelastic collapse at time $\tau^* > 0$.\\
We say that the \emph{Zhou-Kadanoff regime takes place} (or \emph{ZK-regime}, in short) takes place, if the sequence $\left( \zeta_n \right)_n$ of the $n$-th Zhou-Kadanoff parameter, defined for the $n$-th collision, for any $n \in\ \mathbb{N}^*$, is converging towards $0$, that is, if for any $n\in\ \mathbb{N}^*$, considering that the collision is of type \circled{$i_n$}-\circled{$k_n$}, we have:
\begin{align}
\zeta_n = \frac{d_n(2+d_n) \left\vert W_{n} \right\vert^2}{(1+d_n)^2 \eta_{n}^2} {\xrightarrow[n \rightarrow +\infty]{} 0},
\end{align}
where $d_n = d_{(i_n,k_n),n}$ is the distance between the particles \circled{$i_n$} and \circled{$k_n$}, $W_{2,n} = W_{2,(i_n,k_n),n}$ is the relative velocity between the particles \circled{$i_n$} and \circled{$k_n$}, and $\eta_{n} = \eta_{(i_n,k_n),n}$ is the normal component of this relative velocity, when the collision \circled{$i_n$}-\circled{$j_n$} takes place.
\end{defin}
\noindent
In addition, according to the estimate \eqref{EQUATSS3.3PROPOComp1Tau_d}, holding in full generality, since $\tau$ behaves asymptotically as $d/(-\eta_2)$, the condition $d/\eta_2^2 \rightarrow 0$ is not only a sufficient condition to have $\tau \ll \eta_2$, but it is also a necessary condition. Therefore, the ratio
\begin{align*}
\varphi_2 = \frac{d}{\eta_2^2}
\end{align*}
turns out to be a natural measure of the defect from the regime $\tau \ll \eta$.\\
\newline
However, there are several limitations and difficulties: on the one hand, we will see that the regime $\zeta \rightarrow 0$ is not always true when a collapse occurs. Another natural question that arises is then: is this regime at least the only one we can observe, that is, the only one associated to a set of initial data of positive measure? On the other hand, we proved that $\tau \ll \eta$ under the assumption that $\zeta \rightarrow 0$ \emph{and} that $\vert W_2 \vert^2$ is not vanishing at the final time of the collapse. The last assumption cannot be true in full generality, although one might expect that it is almost always true. When such an assumption does not hold, it implies in particular that at least two particles of the system remain attached after the final time of the collapse. This configuration is preventing to get well-posedness of the system of particles, and is a major obstruction to obtain an Alexander theorem for the system of three inelastic particles. As a consequence, understanding the cases when one or two relative velocities completely vanish at the final time of the collapse is of particular interest for the study of collapsing systems of particles.\\
In the companion paper \cite{DoVeArt}, we study the two-dimensional system obtained by keeping only the leading order terms of \eqref{EQUATSS2.2IterationSystm1}, using the information of Proposition \ref{PROPOSS3.3AsympCompa__II_}. In the end, the system reduces in studying only the two variables $\varphi_1 = \eta_1/(-\eta_2)$ and $\varphi_2$. We investigate in particular when the ZK-regime (equivalent to $\varphi_2 \rightarrow 0$ for non vanishing tangential velocities) takes place. We prove in this paper that the regime $\varphi_2$ is stable in a non trivial region of the plane $(\varphi_1,\varphi_2)$. We also conjecture the existence of a separatrix in this plane, that enables to characterize the ZK-regime.

\section{Sequence of collisions involving only two pairs of particles: the nearly-linear collapse}
\label{SSECTIO3.4SuiteColli01-02}

In this section, we consider a collapsing system of particles colliding (eventually) with the following sequence of collisions: the infinite repetitions of the pairs of collisions \circled{0}-\circled{1}, \circled{0}-\circled{2}. For reasons that will become clear later in this section, such a collapse will be called a \emph{nearly-linear collapse}. This is the configuration already investigated in \cite{ZhKa996}. In the first part of this section, we perform a rigorous study of the convergence of the variables of the dynamical system \eqref{EQUATSS2.2IterationSystm1} in the case of the linear collapse, which was, to the best of our knowledge, still not done in the literature. In a second part, we present the results already existing in the literature, namely, in \cite{ZhKa996}.\\
Let us note that for a collapsing system of three particles, there are only two possible cases: either only two pairs are involved in infinitely many collisions, or the three pairs are all involved in infinitely many collisions. The latter case is discussed in Section \ref{SSECTIO3.5SuiteColli_012_}. We study in this section the former.\\
Finally, let us observe that the present situation covers the one-dimensional case, which is completely understood for $3$ particles (see \cite{CoGM995}). Indeed, if the particles evolve along a line, the external particles cannot collide with each other, so if the system experiences infinitely many collisions, it has to be between the two pairs formed, respectively, by the two external particles, and the central one. In this work we do not restrict ourselves to the one-dimensional case, but we will see that the results we can gather concerning the linear collapse can be seen as a spatial perturbation of the one-dimensional case, especially when considering the results of \cite{ZhKa996}.

\subsection{Converging quantities in the case of the nearly-linear collapse}
\label{SSSCT3.4.1ConveQuantLinea}

\begin{defin}[Nearly-linear inelastic collapse]
\label{DEFINSS3.4CollapseLineair}
Let $r \in\ ]0,1[$ be a positive real number smaller than $1$, and let us consider a system of three inelastic particles \circled{0}, \circled{1} and \circled{2} evolving according to the $r$-inelastic hard sphere flow \eqref{EQUATSS2.1_Loi_de_Newton_}, \eqref{EQUATSS2.1VitesPost-Colli}, on a time interval $[0,\tau^*[$, and undergoing an inelastic collapse at time $\tau^* > 0$.\\
We say that the system experiences a \emph{nearly-linear collapse} if the sequence of collisions becomes eventually the infinite repetition of the pairs of collisions \circled{0}-\circled{1}, \circled{0}-\circled{2}.\\
The particles \circled{0}, involved in infinitely collisions with the two other particles \circled{1} and \circled{2}, will be called the \emph{central particle}. The particles \circled{1} and \circled{0}, involved in infinitely many collisions only with the particle \circled{0}, will be called, alternatively, the \emph{colliding}, and the \emph{spectator} particles.
\end{defin}

\begin{remar}
The nomenclature of Definition \ref{DEFINSS3.4CollapseLineair}, concerning the name of the particles according to their respective roles along the different collisions, is taken from \cite{ZhKa996}. Let us note that this definition is not arbitrary, up to rename the particles.
\end{remar}

\noindent
Let us start with the geometrical description of the nearly-linear collapse, following as an immediate corollary of Proposition \ref{PROPOSS3.2ConveAngleDista}.

\begin{propo}[Geometry of the nearly-linear collapse]
\label{PROPOSS3.4GeomeCollaLinea}
Let $r \in\ ]0,1[$ be a positive real number smaller than $1$, and let us consider a system of three inelastic particles \circled{0}, \circled{1} and \circled{2} evolving according to the $r$-inelastic hard sphere flow \eqref{EQUATSS2.1_Loi_de_Newton_}, \eqref{EQUATSS2.1VitesPost-Colli}, on a time interval $[0,\tau^*[$, and undergoing a nearly-linear inelastic collapse at time $\tau^* > 0$.\\
Then, the particles \circled{0} and \circled{1}, resp. \circled{0} and \circled{2}, are in contact at the collapsing time $\tau^*$, that is, we have:
\begin{align}
\vert \overline{\omega}_1 \vert = 1 \hspace{5mm} \text{and} \hspace{5mm} \vert \overline{\omega}_2 \vert = 1.
\end{align}
\end{propo}
\noindent
The particles are then in contact in the limit, forming a linear structure, in the sense that the particles \circled{1}-\circled{0}-\circled{2} form a sort of a queue, possibly curved, hence the name of this collapse. Such a structure is typical from the collapsing systems of inelastic particles, as it was first observed by \cite{McYo993}. The interested reader may also consult \cite{PoSc005} for more intriguing observations about this phenomenon.\\%\footnote{One may for instance consider Figure 1.4 page 5, representing the concentration of the pressure inside a rotating ball mill, where the grist filling the mill is modelized by inelastic spheres, and the section ``Collision Chains'' page 185, where in particular the distribution of the length of the chains of colliding particles is discussed.}.\\
\newline
We can also describe the convergence of the normal components of the relative velocities. By assumption the nearly-linear collapse has finite maximal gaps (see Definition \ref{DEFINSS3.2FonctCompt_Gaps}). In particular, the results of Proposition \ref{PROPOSS3.2SommaCompNormlV} hold true for the linear collapse. We can improve this result in the present case and provide an \emph{explicit} rate of convergence.

\begin{propo}[Convergence of the normal components of the relative velocities for the nearly-linear collapse]
\label{PROPOSS3.4CLineConveNorma}
Let $r \in\ ]0,1[$ be a positive real number smaller than $1$, and let us consider a system of three inelastic particles \circled{0}, \circled{1} and \circled{2} evolving according to the $r$-inelastic hard sphere flow \eqref{EQUATSS2.1_Loi_de_Newton_}, \eqref{EQUATSS2.1VitesPost-Colli}, on a time interval $[0,\tau^*[$, and undergoing a nearly-linear inelastic collapse at time $\tau^* > 0$.\\
Then we have:
\begin{align}
\eta_{1,n},\eta_{2,n} \xrightarrow[n \rightarrow +\infty]{} 0,
\end{align}
where $\eta_{i,n}$ denotes the normal component of the relative velocity between the particles \circled{0} and \circled{i}, that is $\eta_{i,n} = W_i\cdot \omega_i = (v_i-v_0)\cdot(x_i-x_0)$.\\
In addition, the normal components $\eta_{1,n}$ and $\eta_{2,n}$ converge exponentially fast to zero, at a rate at least equal to
\begin{align*}
\max\left(\frac{(1+r)}{2} \big\vert \cos\overline{\theta} \big\vert,r\right),
\end{align*}
where $\overline{\theta}$ denotes the angle between the angular parameters $\omega_1 = x_1-x_0$ and $\omega_2 = x_2-x_0$ at the limiting time $\tau^*$ of the collapse. As a consequence, the series of the normal components $\sum_{n \geq 0} \vert \eta_{1,n} \vert$ and $\sum_{n \geq 0}\vert \eta_{2,n} \vert$ are both converging.
\end{propo}

\begin{proof}
When the sequence of collisions is prescribed and given by the infinite repetition of \circled{0}-\circled{1}, \circled{0}-\circled{2}, we can compute explicitly the normal components of the relative velocities: first we compute the post-collisional velocities according to \eqref{EQUATSS2.2IterationSystm1}, (the system describes a collision between \circled{0} and \circled{2}, while \circled{1} remains spectator), and then we have to exchange the roles between the particles \circled{1} and \circled{2}, in order to perform the next collision. We can then repeat the process in order to obtain the velocities after an arbitrary number of collisions. This process writes explicitly, if we denote by $\eta_c$ the normal component of the relative velocity with the colliding particle, while $\eta_s$ corresponds to the spectator particle, the particles \circled{1} and \circled{2} being alternatively colliding and spectator:

\begin{equation}
\label{EQUATSS3.4EvoluVNormLinea}
\left\{
\begin{array}{rl}
\eta_{s,n+1} &= - r \left(1+\varepsilon_{1,n}\right) \eta_{c,n} + \mathcal{O}_{1,n} \tau_n,\\
\eta_{c,n+1} &= \left(1 + \varepsilon_{2,n}\right) \eta_{s,n} - \frac{(1+r)}{2} \cos\overline{\theta} \left(1 + \varepsilon_{3,n}\right) \eta_{c,n} + \mathcal{O}_{2,n} \tau_n,
\end{array}
\right.
\end{equation}
\noindent
where $\left( \varepsilon_{i,n} \right)_n$ are vanishing sequences, $\left( \mathcal{O}_{i,n} \right)_n$ are bounded sequences, and $\cos\overline{\theta}$ is the limit of the sequence $\left( \omega_{1,n}\cdot\omega_{2,n} \right)_n$ according to Proposition \ref{PROPOSS3.3AsympCompa__I__}. Since the eigenvalues of the limiting matrix
\begin{align}
\label{EQUATSS3.4MatriLimitLinea}
A = \begin{pmatrix} 0 & -r \\ 1 & -\frac{(1+r)}{2}\cos\overline{\theta} \end{pmatrix}
\end{align}
are given, when they are real, by the expression:
\begin{align*}
\frac{\text{Tr}(A)\pm\sqrt{\text{Tr}^2(A)-4\det(A)}}{2}
\end{align*}
and so these eigenvalues are bounded in absolute value by
\begin{align*}
\lambda_\pm = \frac{\vert \text{Tr}(A) \vert + \sqrt{\text{Tr}^2(A)-4\det(A)} }{2} \leq \vert \text{Tr}(A) \vert = \left\vert \frac{(1+r)}{2}\cos\overline{\theta} \right\vert < 1
\end{align*}
(because $\det(A) = r > 0$). In the case when they are complex, that is when $\text{Tr}^2(A)-4\det (A) < 0$, they are given by
\begin{align*}
\lambda_\pm = \frac{\text{Tr}(A) \pm i \sqrt{4\det (A) - \text{Tr}^2(A)}}{2},
\end{align*}
so that their squared modulus is
\begin{align*}
\vert \lambda_\pm \vert^2 = \frac{1}{4} \left[ \text{Tr}^2(A) + 4\det (A) - \text{Tr}^2(A) \right] = \det (A) = r < 1.
\end{align*}
In all the cases, we observe that the spectral radius of the limiting matrix $A$ is strictly smaller than $1$. Therefore, there exists a norm on $\mathbb{C}^2$ such that the induced operator norm of all the matrices $A_n$, for $n$ large enough, is strictly smaller than $1$, where
\begin{align*}
A_n = \begin{pmatrix} 0 & -r(1+\varepsilon_{1,n}) \\ 1+\varepsilon_{2,n} & -\frac{(1+r)}{2} (1+\varepsilon_{3,n}) \cos\overline{\theta} \end{pmatrix} .
\end{align*}
As a consequence, since \eqref{EQUATSS3.4EvoluVNormLinea} implies by immediate recursion that
\begin{align*}
\begin{pmatrix}
\eta_{s,n} \\ \eta_{c,n}
\end{pmatrix}
=
\left[ \prod_{k=0}^{n-1} A_k \right] \begin{pmatrix} \eta_{s,0} \\ \eta_{c,0} \end{pmatrix}
+
\sum_{k=0}^{n-1} \tau_k \left[ \prod_{j=k+1}^{n-1} A_j \right] \begin{pmatrix} \mathcal{O}_{1,k} \\ \mathcal{O}_{2,k} \end{pmatrix},
\end{align*}
we can deduce that the normal components $\eta_{s,n}$ and $\eta_{c,n}$ both converge to zero as $n$ goes to infinity, exponentially fast.
\end{proof}

\begin{remar}
Note that in the proof above we did not need to require any extra assumption on the spectrum of the matrix $A$ (such as the realness of the eigenvalues).
\end{remar}

\noindent
Thanks to Proposition \ref{PROPOSS3.4CLineConveNorma}, we are now in position to conclude the study of the convergence of the different variables of the dynamical system \eqref{EQUATSS2.2IterationSystm1}, in the case of a nearly-linear collapse.

\begin{propo}[Convergence of the tangential components of the relative velocities for the nearly-linear collapse]
\label{PROPOSS3.4CLineConveTangt}
Let $r \in\ ]0,1[$ be a positive real number smaller than $1$, and let us consider a system of three inelastic particles \circled{0}, \circled{1} and \circled{2} evolving according to the $r$-inelastic hard sphere flow \eqref{EQUATSS2.1_Loi_de_Newton_}, \eqref{EQUATSS2.1VitesPost-Colli}, on a time interval $[0,\tau^*[$, and undergoing a nearly-linear inelastic collapse at time $\tau^* > 0$.\\
Then the sequences of the tangential components $\left(W_1^\perp\right)_{n\in\mathbb{N}}$ and $\left(W_2^\perp\right)_{n\in\mathbb{N}}$ of the relative velocities $W_1 = v_1-v_0$ and $W_2 = v_2-v_0$ are converging.
\end{propo}

\begin{proof}
Let us prove the result for the first relative velocity. According to \eqref{EQUATSS2.2IterationSystm1}, when the collision between the particles \circled{0} and \circled{2} takes place, the tangential component of the first relative velocity is modified according to:
\begin{align*}
W_{1,n+1}^\perp &= W_{1,n} - \frac{\eta_{1,n} +\tau_n \vert W_{1,n} \vert^2}{(1+d_{n+1})} + \frac{(1+r)}{2}\left((1+d_n)\eta_{2,n} + \tau_n \vert W_{2,n} \vert^2 \right) \left[ (\omega_{1,n+1}\cdot\omega_{2,n+1}) \omega_{1,n+1} - \omega_{2,n+1} \right] \\
&= W^\perp_{1,n} + \eta_{1,n} \omega_{1,n} - \frac{\eta_{1,n} +\tau_n \vert W_{1,n} \vert^2}{(1+d_{n+1})} \\
&\hspace{40mm} +\frac{(1+r)}{2}\left((1+d_n)\eta_{2,n} + \tau_n \vert W_{2,n} \vert^2 \right) \left[ (\omega_{1,n+1}\cdot\omega_{2,n+1}) \omega_{1,n+1} - \omega_{2,n+1} \right].
\end{align*}
According to Proposition \ref{PROPOSS3.2ConveInterTemps}, the series of the time intervals $\tau_n$ between two consecutive collisions is converging. On the other hand, according to Proposition \ref{PROPOSS3.4CLineConveNorma}, the series of the normal components are also summable. Finally, according to Proposition \ref{PROPOSS3.2ConveAngleDista}, we know that the sequence of the distances $\left(d_n\right)_n$ is vanishing and that the sequences of the angular parameters $\left(\omega_{i,n}\right)_n$ are converging. Therefore, we can write the iteration satisfied by the tangential component of the first relative velocity as:
\begin{align*}
W_{1,n+1}^\perp = W_{1,n}^\perp + \delta_n,
\end{align*}
where $\sum_{n\geq0} \delta_n$ is a summable series since we have:
\begin{align*}
\vert \delta_n \vert \leq \vert \eta_{1,n} \vert + \vert \eta_{1,n} \vert + \tau_n \mathcal{E}_0 + (1+r) \left( 2 \vert \eta_{2,n} \vert + \tau_n \mathcal{E}_0 \right),
\end{align*}
where $\mathcal{E}_0$ denoted the kinetic energy at the initial time.\\
Therefore, we conclude that the series $\sum_{n \geq 0} \left( W_{1,n+1}^\perp - W_{1,n}^\perp \right)$ is summable, that is, the sequence $\left( W_{1,n}^\perp\right)_{n\geq0}$ is converging.\\
The argument is exactly the same for the tangential component of the second relative velocity, and the conclusion follows accordingly.
\end{proof}

\subsection{The results of Zhou and Kadanoff concerning the nearly-linear collapse}
\label{SSSCT3.4.2LineaZhou-Kadan}

For the sake of completeness, let us recall the results in the literature concerning the nearly-linear collapse in dimension $d \geq 2$. In \cite{ZhKa996}, Zhou and Kadanoff obtained two necessary conditions concerning the nearly-linear collapse, that write as follows.
%, both of them relying on a careful study of the spectrum of the matrix $A$, defined in \eqref{EQUATSS3.4MatriLimitLinea}, corresponding to the evolution law of the velocities given by \eqref{EQUATSS3.4EvoluVNormLinea}, in the case of the infinite repetition of the pairs of collisions \circled{0}-\circled{1}, \circled{0}-\circled{2}. In particular, Zhou and Kadanoff obtained bounds on the maximal angle of the linear configuration of the three particles \circled{1}-\circled{0}-\circled{2} at the collapsing time.\\
Let $r \in\ ]0,1[$ be a positive real number smaller than $1$, and let us consider a system of three inelastic particles \circled{0}, \circled{1} and \circled{2} evolving according to the $r$-inelastic hard sphere flow \eqref{EQUATSS2.1_Loi_de_Newton_}, \eqref{EQUATSS2.1VitesPost-Colli}, on a time interval $[0,\tau^*[$, undergoing a nearly-linear inelastic collapse at time $\tau^* > 0$, and such that the angle between the pairs of particles \circled{0}-\circled{1} and \circled{0}-\circled{2} converges to the limit value $\theta$ at the collapsing time $\tau^*$.
\begin{itemize}
\item The existence of such a nearly-linear inelastic collapse is possible only if:
\begin{align}
\label{EQUATSS3.4ZKCl1CondiExist}
-\cos \theta \geq \frac{4\sqrt{r}}{1+r} \cdotp
\end{align}
\item In addition, such a nearly-linear inelastic collapse is stable (with respect to perturbations of the initial data leading to such a collapse) only if:
\begin{align}
\label{EQUATSS3.4ZKCl2CondiStabi}
-\cos \theta > \frac{2r^{1/3}(1+r^{1/3})}{1+r} \cdotp
\end{align}
\end{itemize}
\noindent
In particular, according to \eqref{EQUATSS3.4ZKCl1CondiExist}, in the limit $r \rightarrow 0$ of the very small restitution parameter (that is, in the case when the system dissipates a lot of kinetic energy), nearly-linear inelastic collapses such that the particles form any obtuse angle in the limit can  a priori exist. On the other hand, considering the constraint $\theta = \pi$, which corresponds to an inelastic collapse such that the particles are perfectly aligned (or, in other words, such that the particles, at the final time of the inelastic collapse, all lie in a one dimensional space), provides a bound from above on $r$, namely $4\sqrt{r} \leq 1+r$, or again $r^2 - 14r + 1 \geq 0$. Since $r \leq 1$, this is equivalent to $r \leq r_\text{exist.}$ with:
\begin{align}
r_\text{exist.} = 7 - 4\sqrt{3} \simeq 0.07179677.
\end{align}
Concerning the second condition \eqref{EQUATSS3.4ZKCl2CondiStabi}, note that the function $r \mapsto 2r^{1/3}(1+r^{1/3})/(1+r)$ is increasing on $[0,1]$, is equal to $0$ at $r=0$, and equal to $2$ as $r=1$. Therefore, there exists a unique critical value $r_c$ of the restitution parameter such that for $r \geq r_c$, $2r^{1/3}(1+r^{1/3})/(1+r) \geq 1$. For such restitution parameters, there exists no angle $\theta$ that could fulfill the inequality \eqref{EQUATSS3.4ZKCl2CondiStabi}, so that no elastic collapse can be stable for that range of restitution coefficients. The critical value is given by $1 = 2r_c^{1/3}(1+r^{1/3})/(1+r_c)$, or again $1 - r_c^{1/3} + r_c^{2/3} = 2r_c^{1/3}$ so that $r_\text{stabi.}^{1/3} = (3 - \sqrt{5})/2$, and then in the end:
\begin{align}
r_\text{stabi.} = \frac{1}{8}\left(27 - 27\sqrt{5} + 45 - 5\sqrt{5}\right) = \frac{72 - 32\sqrt{5}}{8} = 9 - 4\sqrt{5} \simeq 0.05572809.
\end{align}

\noindent
The collapses in \cite{ZhKa996} are obtained as self-similar solutions of a simplified, one-dimensional version of the complete dynamical system \eqref{EQUATSS2.2IterationSystm1}, corresponding to fixed points of some mapping, which turns out to be an homography \eqref{EQUATSS3.4ZKCl1CondiExist} provides the existence of such fixed points, and \eqref{EQUATSS3.4ZKCl2CondiStabi} ensures that the Ansatz $\tau \ll \eta$ holds true, which allows in turn to obtain formally the one-dimensional version of the complete dynamical system.

\subsection{Explicit construction of a stable set of initial configurations leading to the collapse}

The previous results do not describe the basin of attraction of the collapsing final configurations, with a given angle. In particular, it seems particularly delicate to deduce, considering an initial data leading to a nearly-linear collapse, what will be the limiting angle of the three particles at the collapsing time $\tau^*$. Let us address a weaker version of this problem by constructing explicitly a set of the phase space that has a non trivial interior, composed with initial configurations of systems that all experience a nearly-linear collapse, with a final angle that is prescribed, up to an error under control.\\
To the best of our knowledge, this is the first explicit description of a stable set of initial data leading to collapse, in dimension strictly larger than $1$.

\begin{theor}
\label{THEORSS4.1ExistenCollapse}
Let $r$ be a positive number such that
\begin{align}
\label{EQUATSS4.1ThExiCondi_r_1_}
0 < r < r_\text{stabi}
\end{align}
with $r_\text{stabi} = 9-4\sqrt{5}$, let $\theta_0 \in\, ]\pi/2,\pi]$ be an angle such that:
\begin{align}
\label{EQUATSS4.1ThExiCondiTheta}
0 < -\cos\theta_0 < \frac{2r^{1/3}(1+r^{1/3})}{1+r},
\end{align}
and let $\delta_\theta$ be any positive real number.\\
Then, there exists a set $\mathcal{ZK}$ of initial configurations of three inelastic particles \circled{0}, \circled{1}, \circled{2}, with a non trivial interior, such that the trajectory $Z(t)$ issued from any initial configuration $Z_0$ of the set $\mathcal{ZK}$ experiences a nearly-linear collapse in finite time, with \circled{0} as a central particle. In addition, the final angle $\overline{\theta}(Z_0)$, at the time of the inelastic collapse, formed by the external particles \circled{1} and \circled{2} around \circled{0}, is such that:
\begin{align}
\left\vert (-\cos\overline{\theta}) - (-\cos\theta_0) \right\vert \leq \delta_\theta.
\end{align}
\end{theor}

\begin{remar}
It is possible to describe explicitly the set of initial configurations mentioned in the theorem.
\end{remar}

\begin{proof}[Proof of Theorem \ref{THEORSS4.1ExistenCollapse}]
Before starting the construction of an initial datum leading to a collapse, let us first make some preliminary observations.\\
We define $\alpha_0$ as the positive quantity:
\begin{align}
\alpha_0 = \frac{(1+r)}{2}\left( - \cos\theta_0 \right),
\end{align}
smaller than $1$. The assumption \eqref{EQUATSS4.1ThExiCondi_r_1_} ensures that the mapping:
\begin{align}
\varphi_1 \mapsto \frac{r}{\alpha_0 - \varphi_1}
\end{align}
has the two positive fixed points
\begin{align}
F_1:\varphi_{1,\alpha_0}^\pm = \frac{\alpha_0 \pm \sqrt{\alpha_0^2 - 4r}}{2},
\end{align}
where only $\varphi_{1,\alpha_0}^-$, the smaller, is stable, that is, we have:
\begin{align}
\label{EQUATSS4.1DerivFonctF_1P-}
\left\vert F_1' \left( \varphi_{1,\alpha_0}^- \right) \right\vert = \left\vert \frac{r}{\left( \alpha_0 - \varphi_{1,\alpha_0}^- \right)^2} \right\vert < 1. 
\end{align}
By definition, $\varphi_{1,\alpha_0}^- < \alpha_0/2$, let us denote
\begin{align}
\delta_1 = \frac{\alpha_0}{2} - \varphi_{1,\alpha_0}^- > 0.
\end{align}
Since $\alpha_0 < 1$ and $r > 0$, let us observe in particular that
\begin{align}
-1 < \varphi_{1,\alpha_0}^- - \alpha_0 < 0.
\end{align}
We denote by $\delta_2$ and $\delta_3$ the two positive quantities defined as:
\begin{align}
\delta_2 = \delta_2(r,\theta_0) = \frac{1}{2} \left( \alpha_0 - \varphi_{1,\alpha_0}^- \right) \hspace{3mm} \text{and} \hspace{3mm} \delta_3 = \delta_3(r,\theta_0) = \frac{1}{2} \left( 1 - \left(\alpha_0 - \varphi_{1,\alpha_0}^- \right) \right).
\end{align}
We introduce also the positive real number $0 < C_\eta < 1$, defined by means of:
\begin{align}
\label{EQUATSS4.1Quantite_C_eta_}
C_\eta = 1 - \delta_3.
\end{align}
According to \eqref{EQUATSS4.1DerivFonctF_1P-}, let us denote by $h_4$ the positive quantity defined as:
\begin{align}
h_4 = h_4(r,\theta_0) = \frac{1}{2} \left( 1 - \frac{r}{\left( \alpha_0 - \varphi_{1,\alpha_0}^- \right)^2} \right),
\end{align}
and let us define the positive quantity $\delta_4 = \delta_4(r,\theta_0)$ as:
\begin{align}
\frac{r}{\left( \alpha_0 - (\varphi_{1,\alpha_0}^-+\delta_4) \right)^2} = 1 - h_4.
\end{align}
In the same way, the second assumption \eqref{EQUATSS4.1ThExiCondiTheta} ensures that the positive quantity:
\begin{align}
\frac{ \varphi_{1,\alpha_0}^- }{ \left( \alpha_0 - \varphi_{1,\alpha_0}^- \right)^2}
\end{align}
is strictly smaller than $1$. Let us denote by $h_5$ the positive quantity defined as:
\begin{align}
h_5 = h_5(r,\theta_0) = \frac{1}{2} \left( 1 - \frac{ \varphi_{1,\alpha_0}^- }{ \left( \alpha_0 - \varphi_{1,\alpha_0}^- \right)^2} \right),
\end{align}
and let us define the positive quantity $\delta_5 = \delta_5(r,\theta_0)$ as:
\begin{align}
\frac{ \varphi_{1,\alpha_0}^- }{ \left( \alpha_0 - (\varphi_{1,\alpha_0}^- + \delta_5) \right)^2} = 1 - h_5.
\end{align}
Finally, we choose two positive numbers $0 < V_0 < V_1$ satisfying:
\begin{align}
\frac{V_1}{V_0} = \frac{1-h_5/2}{1-h_5} \cdotp
\end{align}

\noindent
Now, let us introduce some notation. We assume that the initial configuration of the system is such that the particles \circled{0} and \circled{1} are in contact, and just collided. The next collision to take place is between \circled{0} and \circled{2}, and between \circled{0} and \circled{1} right after, and so on. Therefore \circled{0} is the central particle (that is, involved in each collision), while the respective roles of \circled{1} and \circled{2} are alternating between spectator and colliding.\\
If we write the configuration of the system at the time of the $n$-th collision (assuming that the first collision to take place at a positive time is between \circled{0} and \circled{2}), such that the time between the $n$-th and the $(n+1)$-th collision is denoted by $\tau_n$, and if we use the subscripts $*_{s,n}$ to denote the particle that just collided at the $n$-th collision (so that this particle will be spectator for the next collision) and $*_{c,n}$ to denote the particle that will collide just after the $n$-th collision, we have for the initial configuration:
\begin{align}
\label{EQUATSS4.1DonneInitiColSp}
\left\{
\begin{array}{rl}
W_{1,0} &= W_{s,0},\\
W_{2,0} &= W_{c,0},\\
x_{1,0} - x_{0,0} &= \omega_{1,0} = \omega_{s,0},\\
x_{2,0} - x_{0,0} &= (1+d_0) \omega_{2,0} = (1+d_0) \omega_{c,0}.
\end{array}
\right.
\end{align}
and we can rewrite the equations \eqref{EQUATSS2.2IterationSystm1} of the dynamics as:
\begin{align}
\label{EQUATSS4.1SysteDynamColSp} % Écrit en utilisant la description collision-spectateur
\left\{
\begin{array}{rl}
W_{c,n+1} &= W_{s,n} - \frac{(1+r)}{2} W_{c,n} \cdot \omega_{s,n+1} \omega_{s,n+1},\\
W_{s,n+1} &= W_{c,n} - (1+r) W_{c,n} \cdot \omega_{s,n+1} \omega_{s,n+1},\\
\omega_{c,n+1} &= \frac{1}{1+d_{n+1}} \left( \omega_{s,n} + \tau_n W_{s,n} \right),\\
\omega_{s,n+1} &= (1+d_n) \omega_{c,n} + \tau_n W_{c,n},\\
d_{n+1} &= \sqrt{1 + 2 \eta_{s,n} \tau_n + \vert W_{s,n} \vert^2 \tau_n^2 } - 1,
\end{array}
\right.
\end{align}
with
\begin{align}
\tau_n = \frac{(1+d_n)}{\vert W_{c,n} \vert^2} (-\eta_{c,n}) \left[ 1 - \sqrt{1-\zeta_n} \right] \hspace{3mm} \text{and} \hspace{3mm} \zeta_n = \frac{d_n(2+d_n)}{(1+d_n)^2} \frac{\vert W_{c,n} \vert^2}{\eta_{c,n}^2},
\end{align}
and where $\eta_{c,n}$ and $\eta_{s,n}$ denote respectively:
\begin{align}
\eta_{c,n} = W_{c,n} \cdot \omega_{c,n} \hspace{3mm} \text{and} \hspace{3mm} \eta_{s,n} = W_{s,n} \cdot \omega_{s,n}.
\end{align}
With these notations, the necessary and sufficient conditions for the collapse to take place are:
\begin{align}
\forall n\in\mathbb{N}, \hspace{3mm} \eta_{c,n} < 0 \hspace{3mm} \text{and} \hspace{3mm} \zeta_n < 1
\end{align}
(which ensures that the particles are initially approaching after the $n$-th collision, and that they will indeed collide), together with:
\begin{align}
\forall n\in\mathbb{N}, \forall t \in [0,\tau_n], \hspace{3mm} \vert (1+d_n)\omega_{c,n} - \omega_{s,n} + t\left( W_{c,n} - W_{s,n} \right) \vert > 1,
\end{align}
which ensures that just after a collision between \circled{0} and \circled{1} (respectively \circled{0} and \circled{2}), a collision between \circled{1} and \circled{2} cannot happen before the next scheduled collision, between \circled{0} and \circled{2} (respectively \circled{0} and \circled{1}).\\
We need now to compute some evolution laws. Let us start with the normal component $\eta_{c,n}$ of the colliding pair. For all $n \geq 1$ we have:
\begin{align}
\label{EQUATSS4.1EvoluCompoNormC}
\eta_{c,n+1} &= W_{c,n+1} \cdot \omega_{c,n+1} \nonumber\\
&= \frac{1}{1+d_{n+1}} \eta_{s,n} - \frac{(1+r)}{2}(1+d_n)\sqrt{1-\zeta_n} \left( -\omega_{s,n+1}\cdot\omega_{c,n+1} \right)(-\eta_{c,n}) + \frac{1}{1+d_{n+1}} \tau_n \vert W_{s,n} \vert^2.
\end{align}
For the spectator pair we have:
\begin{align}
\eta_{s,n+1} &= W_{s,n+1}\cdot\omega_{s,n+1} \nonumber\\
&= r(1+d_n) \sqrt{1-\zeta_n} (-\eta_{c,n}).
\end{align}
The square of the norms of the relative velocities evolve as:
\begin{align}
\label{EQUATSS4.1EvoluNormeVitRl} % Normes des Vitesses Relatives
\vert W_{c,n+1} \vert^2 &= \vert W_{s,n} \vert^2 - (1+r) \left( W_{c,n}\cdot \omega_{s,n+1} \right)\left( W_{s,n}\cdot\omega_{s,n+1} \right) + \frac{(1+r)^2}{2} \left( W_{c,n}\cdot\omega_{s,n+1} \right)^2,
\end{align}
and
\begin{align}
\vert W_{s,n+1} \vert^2 &= \vert W_{c,n} \vert^2 + (r^2-1)\left( W_{c,n}\cdot \omega_{s,n+1} \right)^2.
\end{align}
Finally, let us write the evolution of the angle $\theta_n$ between $\omega_{c,n}$ and $\omega_{s,n}$. We have:
\begin{align}
\label{EQUATSS4.1EvoluCosinAngle}
-\cos\theta_{n+1} &= \frac{1+d_n}{1+d_{n+1}} \left(-\cos\theta_n\right) \nonumber\\
&\hspace{5mm} - \frac{1}{1+d_{n+1}} \tau_n \left(\omega_{s,n}\cdot W_{c,n}\right) - \frac{1+d_n}{1+d_{n+1}} \tau_n \left(\omega_{c,n} \cdot W_{s,n}\right) - \frac{1}{1+d_{n+1}} \tau_n^2 \left(W_{s,n}\cdot W_{c,n}\right).
\end{align}
Let us now describe explicitly initial data leading to the collapse. We define first:
\begin{align}
\label{EQUATSS4.1DefinDeltxDelty}
\delta_x = \delta_y = \text{min} \left( \delta_1,\delta_2,\delta_3,\delta_4,\delta_5 \right),
\end{align}
then we introduce:
\begin{align}
\vert x_0 \vert = \text{min} \left( (1-h_4)\delta_x, \frac{V_0}{V_1}\frac{\delta_2^2 h_5}{12} \right),
\end{align}
and for any positive real number $\delta_\theta > 0$, which will be used to measure the maximal variation of the angles during the collapse, we define:
\begin{align}
\vert \Delta \theta \vert = \text{min} \left( \frac{\vert \cos\theta_0 \vert}{2}, \frac{h_4}{8}\delta_x, \delta_\theta \right).
\end{align}
We are now in position to describe the initial data. We introduce:
\begin{align}
\label{EQUATSS4.1DefinBorne_Eta_}
\overline{\eta} = \text{min} \left( 1, \frac{(\sqrt{2}-1)}{16}\frac{V_0}{V_1^{1/2}}, \frac{V_0}{33V_1^{1/2}} \vert \Delta\theta \vert, \frac{(1-C_\eta)(V_1-V_0)}{12(V_1^{1/2}+2)}, \frac{V_0}{V_1^{1/2}}\left( \frac{2}{15\alpha_0} + \frac{h_4\delta_2}{6} + \frac{h_4}{12\alpha_0} + \frac{h_4}{24} \right)\delta_x \right),
\end{align}
\begin{align}
\label{EQUATSS4.1DefinBorne_Zeta}
\overline{\zeta} = \text{min} \left( 2(\sqrt{2}-1), \frac{V_0 h_5}{3V_1(\alpha_0+8V_1/V_0)}\delta_2^2, \frac{1}{16}\left( 1 + 4\delta_2 + \frac{V_0}{V_1} \right)h_4\delta_x \right),
\end{align}
and finally:
\begin{align}
\label{EQUATSS4.1DefinBorne_d_0_}
\overline{d} = \overline{d}\left((-\eta_{c,0})\right) = \text{min} \left( \frac{\sqrt{2}-1}{4}, \frac{\vert \Delta\theta \vert}{3}, \frac{1}{5}(1+\delta_2 h_4)\delta_x, \frac{1}{4}\left(1 + 2\delta_2\right)h_4\delta_x, \frac{\overline{\zeta}}{2V_1}(\eta_{c,0})^2 \right),
\end{align}
noting the bound $\overline{d}$ is a function of the initial normal component $-\eta_{c,0}$, which corresponds to a uniform bound on the initial Zhou-Kadanoff parameter $\zeta_0$.\\
We will then consider initial data such that:
\begin{align}
\label{EQUATSS4.1DescrDoIniColla}
\left\{
\begin{array}{ll}
&\text{the particles \circled{0} and \circled{1} are in contact, in a post-collisional configuration},\\
&\eta_{2,0} = \eta_{c,0} < 0,\\
&\vert \eta_{2,0} \vert \leq \overline{\eta},\\
&\left\vert \frac{\eta_{s,0}}{(-\eta_{c,0})} - \varphi_{1,\alpha_0}^- \right\vert \leq \vert x_0 \vert,\\
&d_0 \leq \overline{d}\left((-\eta_{c,0})\right),\\
&V_0 + \frac{1}{3}(V_1-V_0) \leq \vert W_{c,0} \vert^2,\vert W_{s,0} \vert^2 \leq V_1 - \frac{1}{3}(V_1-V_0),\\
&\omega_{c,0} \cdot \omega_{s,0} = \cos\theta_0.
\end{array}
\right.
\end{align}
\noindent
Let us now show that such initial configurations generate trajectories that eventually collapse, proceeding by recursion. More precisely, we will show that the following properties:
\begin{itemize}
\item the negativity of the normal component of the relative velocity of the colliding pair:
\begin{align}
\label{EQUATSS4.1Recurrence_Cnd1}
\eta_{c,n} < 0,
\end{align}
\item for all $k \leq n$:
\begin{align}
\label{EQUATSS4.1Recurrence_Cnd2}
\vert \eta_{c,k} \vert \leq \overline{\eta} \hspace{3mm} \text{and if} \hspace{3mm} n\geq 1, \hspace{3mm} \vert \eta_{c,k} \vert \leq C_\eta \vert \eta_{c,k-1} \vert,
\end{align}
where $\overline{\eta}$ is defined in \eqref{EQUATSS4.1DefinBorne_Eta_} and $C_\eta$ in \eqref{EQUATSS4.1Quantite_C_eta_},
\item the smallness of the Zhou-Kadanoff parameter:
\begin{align}
\label{EQUATSS4.1Recurrence_Cnd3}
\zeta_n < 1,
\end{align}
\item and in addition
\begin{align}
\label{EQUATSS4.1Recurrence_Cnd4}
\zeta_n \leq \overline{\zeta},
\end{align}
where $\overline{\zeta}$ is defined in \eqref{EQUATSS4.1DefinBorne_Zeta},
\item the condition preventing a collision between the particles \circled{1} and \circled{2}:
\begin{align}
\label{EQUATSS4.1Recurrence_CndT}
\forall \tau \in [0,\tau_n], \hspace{3mm} \left\vert (1+d_n)\omega_{c,n} - \omega_{s,n} + \tau(W_{c,n}-W_{s,n}) \right\vert > 1,
\end{align}
\item the boundedness, from above \emph{and} below, of the norms of the relative velocities:
\begin{align}
\label{EQUATSS4.1Recurrence_Cnd5}
V_0 \leq \vert W_{1,n} \vert^2, \vert W_{2,n} \vert^2 \leq V_1,
\end{align}
\item and in addition, if $n$ is odd:
\begin{align}
\label{EQUATSS4.1Recurrence_Cnd6}
\vert W_{c,n} \vert^2 - \vert W_{s,0} \vert^2 \leq 4(V_1^{1/2}+2) \left( \sum_{k=0}^{n-1} C_\eta^k \right)(-\eta_{2,0}) \hspace{0.5mm} \text{and} \hspace{0.5mm} \vert W_{s,n} \vert^2 - \vert W_{c,0} \vert^2 \leq 4(V_1^{1/2}+2) \left( \sum_{k=0}^{n-1} C_\eta^k \right)(-\eta_{2,0})
\end{align}
and if $n$ is even:
\begin{align}
\label{EQUATSS4.1Recurrence_Cnd7}
\vert W_{c,n} \vert^2 - \vert W_{c,0} \vert^2 \leq 4(V_1^{1/2}+2) \left( \sum_{k=0}^{n-1} C_\eta^k \right)(-\eta_{2,0}) \hspace{0.5mm} \text{and} \hspace{0.5mm} \vert W_{s,n} \vert^2 - \vert W_{s,0} \vert^2 \leq 4(V_1^{1/2}+2) \left( \sum_{k=0}^{n-1} C_\eta^k \right)(-\eta_{2,0})
\end{align}
\item the boundedness of the variation of the angle between the particles:
\begin{align}
\label{EQUATSS4.1Recurrence_Cnd8}
\left\vert \left(-\cos\theta_{n+1}\right) - \left(-\cos\theta_0\right) \right\vert \leq d_0 + 11\frac{V_1^{1/2}}{V_0}\left(-\eta_{2,0}\right) + \frac{V_1^{1/2}}{V_0 C_\eta} \left[ 3+11C_\eta \right] \left( \sum_{k=1}^n C_\eta^k \right) (-\eta_{2,0}),
\end{align}
\item and finally, the explicit description of the stable region, namely, with the bound on the distance between the ratio of the normal components and the unstable fixed point $\varphi_{1,\alpha_0}^-$ on the one hand:
\begin{align}
\label{EQUATSS4.1Recurrence_Cnd9}
\left\vert \frac{\eta_{s,n}}{\left(-\eta_{c,n}\right)} - \varphi_{1,\alpha_0}^- \right\vert \leq \delta_x,
\end{align}
where $\delta_x$ is defined in \eqref{EQUATSS4.1DefinDeltxDelty},
\item and on the other hand, the bound on the distance between the ratio of the consecutive normal components of the colliding pairs and the quantity $\alpha_0 - \varphi_{1,\alpha_0}^-$:
\begin{align}
\label{EQUATSS4.1Recurrence_Cnd10}
\left\vert \frac{\eta_{c,n+1}}{\eta_{c,n}} - \left( \alpha_0 - \varphi_{1,\alpha_0}^- \right) \right\vert \leq \delta_y,
\end{align}
where $\delta_y$ is defined in \eqref{EQUATSS4.1DefinDeltxDelty},
\end{itemize}
holding for the $n$-th collision, propagate and hold true for the $(n+1)$-th collision when the initial datum is chosen carefully. These properties allow in particular to deduce that the collapse takes place.\\
\newline
We proceed first to check that the conditions \eqref{EQUATSS4.1Recurrence_Cnd1}-\eqref{EQUATSS4.1Recurrence_Cnd10} hold for the initial datum.\\
By assumption, $\eta_{2,0} = \eta_{c,0} < 0$, so \eqref{EQUATSS4.1Recurrence_Cnd1} holds. By construction, the first inequality in \eqref{EQUATSS4.1Recurrence_Cnd2} holds true, while the second is an empty condition.\\
Concerning the Zhou-Kadanoff parameter $\zeta_0$, by definition we have:
\begin{align}
\label{EQUATSS4.1BorneParam_ZK0_}
\zeta_0 = \frac{d_0(2+d_0)}{(1+d_0)^2}\frac{\vert W_{c,0} \vert^2}{(\eta_{c,0})^2} \leq 2d_0 \frac{\vert W_{c,0} \vert^2}{(\eta_{c,0})^2},
\end{align}
using the fact that the quantity $x(2+x)/(1+x)^2$ is always bounded from above by $2$ for $x \geq 0$. In the end, since we assumed that $\vert W_{c,0} \vert^2 \leq V_1 - \frac{1}{3}(V_1-V_0) \leq V_1$, we deduce that \eqref{EQUATSS4.1Recurrence_Cnd3} and \eqref{EQUATSS4.1Recurrence_Cnd4} hold concerning the Zhou-Kadanoff parameter $\zeta_0$ according to our assumptions on $\eta_{2,0}$ and $d_0$.\\
Concerning the conditions on the norms of the relative velocities, \eqref{EQUATSS4.1Recurrence_Cnd5} holds by construction, and \eqref{EQUATSS4.1Recurrence_Cnd6} is trivial since $n=0$.\\
Concerning the condition \eqref{EQUATSS4.1Recurrence_CndT}, we start with observing that the following crude bound on $\tau_0$ holds:
\begin{align*}
\tau_0 \leq 2 \frac{(-\eta_{c,0})}{\vert W_{c,0} \vert^2} \left[ 1 - \sqrt{1-\zeta_0} \right] \leq 2\frac{(-\eta_{c,0})}{V_0} \zeta_0.
\end{align*}
Since by assumption $\cos\theta_0 < 0$, the angle formed by the particles \circled{1} and \circled{2} around \circled{0} is obtuse, so that:
\begin{align*}
\vert \omega_{1,0} - \omega_{2,0} \vert \geq \sqrt{2},
\end{align*}
and then for all $\tau \leq \tau_0$:
\begin{align*}
\left\vert (1+d_0)\omega_{c,0} - \omega_{s,0} + \tau(W_{c,0}-W_{s,0}) \right\vert \geq \sqrt{2} - d_0 - 2V_1^{1/2}\tau_0,
\end{align*}
which is strictly larger than $1$ according to the assumptions on $d_0$ and $(-\eta_{c,0})$.\\
The condition \eqref{EQUATSS4.1Recurrence_Cnd8} on the variation of the angle comes from the evolution law \eqref{EQUATSS4.1EvoluCosinAngle}, together with the a priori bound on the distance $d_1$ between the particles after the first collision:
\begin{align}
\label{EQUATSS4.1BorneDista_d_1_}
d_1 = \sqrt{1 + 2\eta_{s,0}\tau_0 + \vert W_{s,0} \vert^2 \tau_0^2} - 1 &\leq \eta_{s,0} \tau_0 + \frac{1}{2} \vert W_{s,0} \vert^2 \tau_0^2 \nonumber\\
&\leq \frac{3}{2} V_1^{1/2} \tau_0 \leq 3 \frac{V_1^{1/2}}{V_0}(-\eta_{c,0})
\end{align}
(using in particular the upper and lower bounds \eqref{EQUATSS4.1Recurrence_Cnd5} on the norms of the relative velocities, and that $(-\eta_{c,0})$ is small enough to have $V_1 \tau_0^2 \leq V_1^{1/2}\tau_0 \leq 1$) which provides:
\begin{align*}
\left\vert (-\cos\theta_1) - (-\cos\theta_0) \right\vert &\leq \left\vert \frac{1+d_0}{1+d_1} - 1 \right\vert (-\cos\theta_0) + \tau_0 \vert W_{c,0} \vert + 2 \tau_0 \vert W_{s,0} \vert + \tau_0^2 \vert W_{s,0} \vert \cdot \vert W_{c,0} \vert \\
&\leq d_0 + d_1 + 4 V_1^{1/2} \tau_0 \\
&\leq d_0 + 11 \frac{V_1^{1/2}}{V_0}(-\eta_{c,0}).
\end{align*}
Finally, the difference between the ratio of the normal components \eqref{EQUATSS4.1Recurrence_Cnd9} holds true by assumption for $n=0$. As for the ratio of the consecutive normal components of the colliding pairs, according to the evolution law \eqref{EQUATSS4.1EvoluCompoNormC} we find:
\begin{align*}
\frac{\eta_{c,1}}{\eta_{c,0}} = \frac{1}{1+d_1}\frac{\eta_{s,0}}{\eta_{c,0}} + \frac{(1+r)}{2}(1+d_0)\sqrt{1-\zeta_0}\left(-\omega_{s,1}\cdot\omega_{c,1}\right) + \frac{1}{1+d_1} \frac{\tau_0}{\eta_{c,0}} \vert W_{s,0} \vert^2.
\end{align*}
Rearranging we find:
\begin{align*}
\frac{\eta_{c,1}}{\eta_{c,0}} - \alpha_0 + \varphi_{1,\alpha_0}^- &= \left[ \frac{1}{1+d_1} - 1 \right] \frac{\eta_{s,0}}{\eta_{c,0}} + \left[ \varphi_{1,\alpha_0}^- - \frac{\eta_{s,0}}{(-\eta_{c,0})} \right] + \left[ (1+d_0)\sqrt{1-\zeta_0} - 1 \right] \frac{(1+r)}{2} (-\omega_{s,1}\cdot\omega_{c,1}) \\
&\hspace{15mm} + \left( \frac{(1+r)}{2} (-\omega_{s,1}\cdot\omega_{c,1}) - \alpha_0 \right) + \frac{1}{1+d_1} \frac{\tau_0}{\eta_{c,0}} \vert W_{s,0} \vert^2.
\end{align*}
By assumption on $x_0$, we have $\vert \eta_{s,0}/\eta_{c,0} \vert \leq \alpha_0/2$, and so we can estimate the first term using the bound \eqref{EQUATSS4.1BorneDista_d_1_} for $d_1$. The second term corresponds exactly to $x_0$, which is smaller in absolute value than $(1-h_4)\delta_x$, the third term is bounded from above by $\text{Max}(d_0,2\zeta_0)$, and the fourth term is estimated using the bound on the variation of the angle that we just obtained. Concerning the fifth term, we need a finer estimate on $\tau_0$ than the one we used above, namely we write:
\begin{align*}
\tau_0 \leq 2 \frac{(-\eta_{c,0})}{\vert W_{c,0} \vert^2} \zeta_0 \leq 4\frac{d_0}{(-\eta_{c,0})},
\end{align*}
where we used the upper bound \eqref{EQUATSS4.1BorneParam_ZK0_} on $\zeta_0$. Since all the five terms of the decomposition, except $\vert x_0 \vert$, are smaller than $h_4 \delta_x/4$, we deduce that the bound \eqref{EQUATSS4.1Recurrence_Cnd10} holds initially.\\
All the conditions \eqref{EQUATSS4.1Recurrence_Cnd1}-\eqref{EQUATSS4.1Recurrence_Cnd10} are satisfied for $n=0$, the initialization step of our proof by recursion is complete.\\
\newline
Let us now prove by recursion that the conditions \eqref{EQUATSS4.1Recurrence_Cnd1}-\eqref{EQUATSS4.1Recurrence_Cnd8} hold for all the collisions.\\
Let us then assume that the conditions \eqref{EQUATSS4.1Recurrence_Cnd1}-\eqref{EQUATSS4.1Recurrence_Cnd8} hold for a certain integer $n$, corresponding to the $n$-th collision of the particles starting from the initial configuration $\left(\omega_{1,0},W_{1,0},d_0,\omega_{2,0},W_{2,0}\right)$. In particular, let us observe that the conditions \eqref{EQUATSS4.1Recurrence_Cnd1}, \eqref{EQUATSS4.1Recurrence_Cnd3} ensure that a collision takes place after the $n$-th collision, and the condition \eqref{EQUATSS4.1Recurrence_CndT} implies in addition that the next collision involves the central particle \circled{0}. We can then compute the post-collisional configuration of the system right after the $(n+1)$-th collision, using the equations \eqref{EQUATSS4.1SysteDynamColSp}.\\
Let us start to show that the estimates \eqref{EQUATSS4.1Recurrence_Cnd6} and \eqref{EQUATSS4.1Recurrence_Cnd7} on the norms of the relative velocities hold for $n+1$. Since the quantity $W_{c,n}\cdot\omega_{s,n+1}$ can be written as
\begin{align*}
W_{c,n} \cdot \omega_{s,n+1} = (1+d_n) \eta_{c,n} + \tau_n \vert W_{c,n} \vert^2 = -(1+d_n) \sqrt{1-\zeta_n}(-\eta_{c,n}),
\end{align*}
according to \eqref{EQUATSS4.1EvoluNormeVitRl} we find:
\begin{align*}
\vert W_{c,n+1} \vert^2 - \vert W_{s,n} \vert^2 &\leq 2 \vert W_{s,n} \vert \cdot \vert W_{c,n} \cdot \omega_{s,n+1} \vert + 2 \vert W_{c,n} \cdot \omega_{s,n+1} \vert^2 \\
&\leq 2 \vert V_1 \vert^{1/2} (1+d_n) (-\eta_{c,n}) + 2 (1+d_n)^2 (-\eta_{c,n})^2 \\
&\leq 4 \left( \vert V_1 \vert^{1/2} + 2 \right) (-\eta_{c,n}),
\end{align*}
using the bound \eqref{EQUATSS4.1Recurrence_Cnd5} on the norm of $W_{s,n}$, and \eqref{EQUATSS4.1Recurrence_Cnd2}, implying on the one hand $(-\eta_{c,n})^2 \leq (-\eta_{c,n}) \leq 1$, and on the other hand, for $n\geq1$:
\begin{align*}
d_n \leq \vert \eta_{s,n} \vert \tau_n + \frac{1}{2} \vert W_{s,n} \vert^2 \tau_n^2 \leq V_1^{1/2} \tau_n + \frac{1}{2} V_1 \tau_n^2 \leq \frac{3}{2} V_1^{1/2} \tau_n \leq 3 \frac{V_1^{1/2}}{V_0} (-\eta_{c,n})
\end{align*}
because
\begin{align}
\label{EQUATSS4.1_Controle_Tau_n}
\tau_n \leq 2 \frac{(-\eta_{c,n})}{\vert W_{c,n} \vert^2} \zeta_n \leq 2 \frac{(-\eta_{c,n})}{V_0}
\end{align}
so that $V_1 \tau_n^2 \leq V_1^{1/2} \tau_n \leq 1$. In the same way, for $\vert W_{s,n+1} \vert^2$ we find
\begin{align*}
\vert W_{s,n+1} \vert^2 - \vert W_{c,n} \vert^2 \leq (1+d_n)^2 (-\eta_{c,n})^2 \leq 4 (-\eta_{c,n}).
\end{align*}
Let us assume that $n+1$ is even, and let us consider $\vert W_{c,n+1} \vert^2$: we obtain together with \eqref{EQUATSS4.1Recurrence_Cnd6} and \eqref{EQUATSS4.1Recurrence_Cnd2}:
\begin{align*}
\vert W_{c,n+1} \vert^2 - \vert W_{c,0} \vert^2 &= \left( \vert W_{c,n+1} \vert^2 - \vert W_{s,n} \vert^2 \right) + \left( \vert W_{s,n} \vert^2 - \vert W_{c,0} \vert^2 \right) \\
&\leq 4\left( V_1^{1/2} + 2 \right) (-\eta_{c,n}) + 4\left(V_1^{1/2} + 2 \right) \left( \sum_{k=0}^{n-1} C_\eta^k \right)(-\eta_{2,0}) \\
&\leq 4\left(V_1^{1/2} + 2 \right) \left( \sum_{k=0}^{n} C_\eta^k \right)(-\eta_{2,0}),
\end{align*}
so that \eqref{EQUATSS4.1Recurrence_Cnd7} holds for $n+1$. The cases concerning $W_{s,n+1}$ or $n+1$ odd are identical.\\
In addition, since $C_\eta < 1$, we have $\sum_{k=0}^n C_\eta^k < (1-C_\eta)^{-1} < +\infty$ for all $n \in \mathbb{N}$, the initial norms $\vert W_{1,0} \vert^2 = \vert W_{s,0} \vert^2$ and $\vert W_{2,0} \vert^2 = \vert W_{c,0} \vert^2$ satisfy $V_0 + \frac{1}{3}(V_1-V_0) \leq \vert W_{c,0} \vert^2, \vert W_{s,0} \vert^2 \leq V_1 - \frac{1}{3}(V_1-V_0)$, and since $(-\eta_{2,0})$ is chosen in particular smaller than $(1-C_\eta)(V_1-V_0)\cdot\left(12 (V_1^{1/2} + 2)\right)^{-1}$, we deduce that the inequalities \eqref{EQUATSS4.1Recurrence_Cnd5} holds also for $n+1$.\\
\newline
We can now consider the first key estimates \eqref{EQUATSS4.1Recurrence_Cnd1} and \eqref{EQUATSS4.1Recurrence_Cnd2} concerning the normal component $\eta_{c,n+1}$ of the colliding pair right after the $(n+1)$-th collision. The equation \eqref{EQUATSS4.1EvoluCompoNormC} can be rewritten as:
\begin{align}
\label{EQUATSS4.1DefinReste_y_n_}
\frac{\eta_{c,n+1}}{-\eta_{c,n}} = \varphi_{1,\alpha_0}^- - \alpha_0 + y_n,
\end{align}
where $y_n$ is a remainder term, expected to be small at each iteration, and so we have the identity:
\begin{align}
\label{EQUATSS4.1Lien_x_n_&_y_n_}
y_n &= \left( \frac{1}{1+d_{n+1}} - 1 \right) \left( \frac{\eta_{s,n}}{(-\eta_{c,n})} \right)
- \frac{(1+r)}{2} \left[ (1+d_n)\sqrt{1-\zeta_n} - 1 \right] (-\cos\theta_{n+1}) \\
&\hspace{10mm} - \frac{(1+r)}{2} \left[ (-\cos\theta_{n+1}) - (-\cos\theta_0) \right]
+ \frac{1}{1+d_{n+1}} \vert W_{s,n} \vert^2 \frac{\tau_n}{(-\eta_{c,n})}
+ x_n,
\end{align}
where we define $x_n$ as:
\begin{align}
\label{EQUATSS4.1DefinReste_x_n_}
\frac{\eta_{s,n}}{(-\eta_{c,n})} = \varphi_{1,\alpha_0}^- + x_n,
\end{align}
also expected to remain small. Let us note though that we cannot expect to have $x_n$, and then $y_n$, vanishing.\\
By \eqref{EQUATSS4.1Recurrence_Cnd10}, we have in particular that
\begin{align*}
\vert y_n \vert \leq \delta_2 \hspace{3mm} \text{and} \hspace{3mm} \vert y_n \vert \leq \delta_3.
\end{align*}
Therefore, we have:
\begin{align*}
-1 < -1 + \delta_3 \leq \varphi_{1,\alpha_0}^- - \alpha_0 + y_n \leq - \delta_2 < 0,
\end{align*}
so that, first, $\eta_{c,n+1}$ is negative, and we have in addition $\vert \eta_{c,n+1} \vert \leq C_\eta \vert \eta_{c,n} \vert$: the estimates \eqref{EQUATSS4.1Recurrence_Cnd1} and \eqref{EQUATSS4.1Recurrence_Cnd2} are proved for $n+1$. In particular, we know that right after the $(n+1)$-th collision, the colliding pair of particles is indeed in a pre-collisional configuration. It remains to check, on the one hand that $\zeta_{n+1} < 1$ to be sure that a collision will eventually take place, and on the other hand that the collision cannot take place between \circled{1} and \circled{2}.\\
Let us consider the Zhou-Kadanoff parameter $\zeta_{n+1}$, computed from the post-collisional configuration of the system, right after the $(n+1)$-th collision. We start with a universal a priori bound on $\zeta_{n+1}$:
\begin{align*}
\zeta_{n+1} &= \frac{d_{n+1}(2+d_{n+1})}{(1+d_{n+1})^2} \frac{\vert W_{c,n+1} \vert^2}{(\eta_{c,n+1})^2} \\
&\leq 2 \vert W_{c,n+1} \vert^2 \frac{d_{n+1}}{(\eta_{c,n+1})^2} \\
&\leq 2 \vert W_{c,n+1} \vert^2 \frac{\vert \eta_{s,n} \vert \tau_n}{(\eta_{c,n+1})^2} + \vert W_{c,n+1} \vert^2 \vert W_{s,n} \vert^2 \frac{\tau_n^2}{(\eta_{c,n+1})^2} \\
&\leq 2 \frac{\vert W_{c,n+1} \vert^2}{\vert W_{c,n} \vert^2} \frac{\eta_{s,n}(-\eta_{c,n})\left[1-\sqrt{1-\zeta_n}\right]}{(\eta_{c,n+1})^2}(1+d_n) + \vert W_{c,n+1} \vert^2 \vert W_{s,n} \vert^2 \frac{\tau_n^2}{(\eta_{c,n+1})^2}
\end{align*}
using in the end the explicit expression of the collision time $\tau_n$. Rearranging and bounding $\tau/(-\eta_{c,n})$ by $4\zeta_n^2/\vert W_{c,n} \vert^4$ in the second term we find:
\begin{align*}
\zeta_{n+1} \leq 2 \frac{\vert W_{c,n+1} \vert^2}{\vert W_{c,n} \vert^2} \frac{\left( \frac{\eta_{s,n}}{(-\eta_{c,n})} \right)}{\left( \frac{(-\eta_{c,n+1})}{(-\eta_{c,n})} \right)^2} \left[1-\sqrt{1-\zeta_n}\right](1+d_n) + \frac{\vert W_{c,n+1} \vert^2 \vert W_{s,n} \vert^2}{\vert W_{c,n} \vert^4} \frac{4}{\left( \frac{(-\eta_{c,n+1})}{(-\eta_{c,n})} \right)^2} \zeta_n^2
\end{align*}
which suggest the following decomposition, using the inequality
\begin{align*}
\left[1 - \sqrt{1-\zeta_n}\right] \leq \frac{1}{2}\zeta_n + \frac{1}{4}\zeta_n^2,
\end{align*}
that holds as soon as $\zeta_n \leq 2(\sqrt{2} - 1)$:
\begin{align*}
\zeta_{n+1} &\leq \frac{\vert W_{c,n+1} \vert^2}{\vert W_{c,n} \vert^2} \frac{\varphi_{1,\alpha_0}^-}{\left( \varphi_{1,\alpha_0}^- - \alpha_0 + y_n \right)^2} \zeta_n + \frac{\vert W_{c,n+1} \vert^2}{\vert W_{c,n} \vert^2} \frac{1}{\left( \varphi_{1,\alpha_0}^- - \alpha_0 + y_n \right)^2} x_n \zeta_n \\
& \hspace{5mm} + \frac{\vert W_{c,n+1} \vert^2}{\vert W_{c,n} \vert^2} \frac{\left( \frac{\eta_{s,n}}{(-\eta_{c,n})} \right)}{\left( \frac{(-\eta_{c,n+1})}{(-\eta_{c,n})} \right)^2} \zeta_n^2 + \frac{\vert W_{c,n+1} \vert^2 \vert W_{s,n} \vert^2}{\vert W_{c,n} \vert^4} \frac{4}{\left( \frac{(-\eta_{c,n+1})}{(-\eta_{c,n})} \right)^2} \zeta_n^2,
\end{align*}
where only the first term is expected not to be a small factor in front of $\zeta_n$. Indeed, \eqref{EQUATSS4.1Recurrence_Cnd9} implies in particular that $\vert x_n \vert \leq \left(V_0/V_1\right)\left(\delta_2^2 h_5\right)/12$, so that we have for the second term:
\begin{align*}
\left\vert \frac{\vert W_{c,n+1} \vert^2}{\vert W_{c,n} \vert^2} \frac{1}{\left( \varphi_{1,\alpha_0}^- - \alpha_0 + y_n \right)^2} x_n \right\vert \leq \frac{V_1}{V_0} \frac{1}{\delta_2^2} x_n \leq \frac{1}{12}h_5,
\end{align*}
using the estimates \eqref{EQUATSS4.1Recurrence_Cnd6}, \eqref{EQUATSS4.1Recurrence_Cnd7} on the norms $\vert W_{c,n+1} \vert^2$, $\vert W_{c,n} \vert^2$ and \eqref{EQUATSS4.1Recurrence_Cnd10}, implying in particular, since $\vert y_n \vert$ is small enough, that $\vert \varphi_{1,\alpha_0}^- - \alpha_0 + y_n \vert \geq \delta_2$.\\
Concerning the third term, we have:
\begin{align*}
\left\vert \frac{\vert W_{c,n+1} \vert^2}{\vert W_{c,n} \vert^2} \frac{\left( \frac{\eta_{s,n}}{(-\eta_{c,n})} \right)}{\left( \frac{(-\eta_{c,n+1})}{(-\eta_{c,n})} \right)^2} \right\vert \leq \frac{V_1}{V_0} \frac{\alpha_0/2}{\delta_2^2},
\end{align*}
using in addition $\vert x_n \vert \leq \delta_1$, so that $\eta_{s,n}/(-\eta_{c,n}) \leq \alpha_0/2$, while for the fourth term:
\begin{align*}
\left\vert \frac{\vert W_{c,n+1} \vert^2 \vert W_{s,n} \vert^2}{\vert W_{c,n} \vert^4} \frac{4}{\left( \frac{(-\eta_{c,n+1})}{(-\eta_{c,n})} \right)^2} \right\vert \leq \frac{V_1^2}{V_0^2} \frac{4}{\delta_2^2},
\end{align*}
so that we see that for:
\begin{align*}
\frac{V_1}{2 V_0 \delta_2^2} \left( \alpha_0 + \frac{8 V_1}{V_0} \right) \zeta_n \leq \frac{h_5}{6},
\end{align*}
which is implied by \eqref{EQUATSS4.1Recurrence_Cnd4}, we obtain for the three remainder terms:
\begin{align*}
\left\vert \frac{\vert W_{c,n+1} \vert^2}{\vert W_{c,n} \vert^2} \frac{1}{\left( \varphi_{1,\alpha_0}^- - \alpha_0 + y_n \right)^2} x_n + \frac{\vert W_{c,n+1} \vert^2}{\vert W_{c,n} \vert^2 \left( \frac{(-\eta_{c,n+1})}{(-\eta_{c,n})} \right)^2} \left[ \frac{\eta_{s,n}}{(-\eta_{c,n})} + 4\frac{\vert W_{s,n} \vert^2}{\vert W_{c,n} \vert^2} \right] \zeta_n \right\vert \leq \frac{h_5}{12} + \frac{h_5}{6} = \frac{h_5}{4} \cdotp
\end{align*}
Finally, for the main term, by construction of the bounds $V_0$ and $V_1$ on the one hand, and by the estimate \eqref{EQUATSS4.1Recurrence_Cnd10} on $\vert y_n \vert$ on the other hand, we have:
\begin{align*}
\frac{\vert W_{c,n+1} \vert^2}{\vert W_{c,n} \vert^2} \frac{\varphi_{1,\alpha_0}^-}{\left( \varphi_{1,\alpha_0}^- - \alpha_0 + y_n \right)^2} \leq \frac{V_1}{V_0}(1-h_5) \leq 1- \frac{h_5}{2}.
\end{align*}
In summary we find:
\begin{align*}
\zeta_{n+1} \leq \left( 1 - \frac{h_5}{2} \right) \zeta_n + \frac{h_5}{4} \zeta_n = \left( 1 - \frac{h_5}{4} \right) \zeta_n,
\end{align*}
so that we obtained \eqref{EQUATSS4.1Recurrence_Cnd3} and \eqref{EQUATSS4.1Recurrence_Cnd4} for $n+1$. In particular, together with the condition \eqref{EQUATSS4.1Recurrence_Cnd1} on $\eta_{c,n+1}$, after the $(n+1)$-th collision, we are now certain that another collision will take place. Let us now verify that the $(n+2)$-th collision cannot be between \circled{1} and \circled{2}, that is, let us check that \eqref{EQUATSS4.1Recurrence_CndT} holds also for $n+1$. The distance between these two particles, right after the $(n+1)$-th collision, and until any further collision, is given by:
\begin{align*}
\left\vert (1+d_{n+1})\omega_{c,n+1} - \omega_{s,n+1} + \tau(W_{c,n+1}-W_{s,n+1}) \right\vert,
\end{align*}
where $\tau$ denotes the time measured from the $(n+1)$-th collision. But then, for $\tau \leq \tau_{n+1}$ we have:
\begin{align*}
\left\vert (1+d_{n+1})\omega_{c,n+1} - \omega_{s,n+1} + \tau(W_{c,n+1}-W_{s,n+1}) \right\vert &\geq \left\vert \omega_{c,n+1} - \omega_{s,n+1} \right\vert - d_{n+1} - 2 V_1^{1/2} \tau \\
&\geq \left\vert \omega_{c,n+1} - \omega_{s,n+1} \right\vert - \frac{3}{2} V_1^{1/2} \tau_n - 2 V_1^{1/2} \tau_{n+1} \\
&\geq \left\vert \omega_{c,n+1} - \omega_{s,n+1} \right\vert - 3 \frac{V_1^{1/2}}{V_0}(-\eta_{c,n}) - 4 \frac{V_1^{1/2}}{V_0}(-\eta_{c,n+1}).
\end{align*}
Since the estimate \eqref{EQUATSS4.1Recurrence_Cnd8}, together with the upper bound \eqref{EQUATSS4.1Recurrence_Cnd2} on the normal component $(-\eta_{c,n})$, imply that $\cos\theta_{n+1} < 0$, we deduce that:
\begin{align*}
\vert \omega_{c,n+1} - \omega_{s,n+1} \vert \geq \sqrt{2},
\end{align*}
so that the upper bound \eqref{EQUATSS4.1Recurrence_Cnd2} implies that the distance between the particles \circled{1} and \circled{2}, between the $(n+1)$-th collision and the next, is always larger than $\frac{1}{2}\left( \sqrt{2} - 1 \right)$, that is, the $(n+2)$-th collision cannot involve the pair \circled{1}-\circled{2}. As a consequence, the system will experiment a $(n+2)$-th collision, involving the central particle \circled{0}.\\
\newline
Now that the $(n+2)$-th collision is certain, involving the right pair of particles, we know that the time interval $\tau_{n+1}$ between the $(n+1)$th and $(n+2)$-th collisions is well-defined. Let us now turn to the estimate \eqref{EQUATSS4.1Recurrence_Cnd8} of the variation of the angles $\theta_n$.\\
Our starting point is the evolution law \eqref{EQUATSS4.1EvoluCosinAngle}. We find:
\begin{align*}
\left\vert \left(-\cos\theta_{n+2}\right) - \left(-\cos\theta_{n+1}\right) \right\vert &= \left\vert \frac{1+d_{n+1}}{1+d_{n+2}} - 1 \right\vert (-\cos\theta_{n+1}) + \frac{1}{1+d_{n+2}} \vert \omega_{s,n+1} \cdot W_{c,n+1} \vert \tau_{n+1} \\
&\hspace{5mm} + \frac{1+d_{n+1}}{1+d_{n+2}} \vert \omega_{c,n+1} \cdot W_{s,n+1} \vert \tau_{n+1} - \frac{1}{1+d_{n+2}} \vert W_{s,n+1} \cdot W_{c,n+1} \vert \tau_{n+1}^2 \\
&\leq d_{n+1} + d_{n+2} + V_1^{1/2}\tau_{n+1} + 2 V_1^{1/2} \tau_{n+1} + V_1 \tau_{n+1}^2.
\end{align*}
$\tau_{n+1}$ being bounded from above by $2(-\eta_{c,n+1})/V_0$ (because of \eqref{EQUATSS4.1_Controle_Tau_n}), which is small enough, we deduce:
\begin{align*}
\left\vert \left(-\cos\theta_{n+2}\right) - \left(-\cos\theta_{n+1}\right) \right\vert 
&\leq 3 \frac{V_1^{1/2}}{V_0}(-\eta_{c,n}) + 3 \frac{V_1^{1/2}}{V_0}(-\eta_{c,n+1}) + 8 \frac{V_1^{1/2}}{V_0} (-\eta_{c,n+1}).
\end{align*}
Using finally the bound \eqref{EQUATSS4.1Recurrence_Cnd2}, holding for $n+1$, and \eqref{EQUATSS4.1Recurrence_Cnd8} holding for $n$, we deduce that this last estimate holds also for $n+1$.\\
\newline
To complete the investigation of the estimates \eqref{EQUATSS4.1Recurrence_Cnd1}-\eqref{EQUATSS4.1Recurrence_Cnd10} for the $(n+1)$-th collision, it remains only to check \eqref{EQUATSS4.1Recurrence_Cnd9} and \eqref{EQUATSS4.1Recurrence_Cnd10}.\\
Let us recall that $x_n$ is defined in \eqref{EQUATSS4.1DefinReste_x_n_}. Therefore we find:
\begin{align*}
\vert x_{n+1} \vert &= \left\vert \frac{\eta_{s,n+1}}{(-\eta_{c,n+1})} - \varphi_{1,\alpha_0}^- \right\vert = \left\vert \frac{r(1+d_n) \sqrt{1-\zeta_n}}{\frac{(-\eta_{c,n+1})}{(-\eta_{c,n})}} - \varphi_{1,\alpha_0}^- \right\vert \\
&\leq \left\vert \frac{r\left[(1+d_n) \sqrt{1-\zeta_n} - 1\right]}{\frac{(-\eta_{c,n+1})}{(-\eta_{c,n})}} \right\vert + \left\vert \frac{r}{\alpha_0-\varphi_{1,\alpha_0}^--y_n} - \varphi_{1,\alpha_0}^- \right\vert.
\end{align*}
Since $\vert y_n \vert \leq \delta_2$, the first term is bounded by:
\begin{align*}
\left\vert \frac{r\left[(1+d_n) \sqrt{1-\zeta_n} - 1\right]}{\frac{(-\eta_{c,n+1})}{(-\eta_{c,n})}} \right\vert \leq \frac{\text{Max}(d_n,2\zeta_n)}{\delta_2},
\end{align*}
and the bound is itself smaller than $(h_4 \delta_x)/2$ by the assumptions on the initial $(-\eta_{2,0})$ and $\zeta_0$ (and here it is crucial to have $(h_4\delta_x)/2$, and not only $h_4 \delta_x$), together with the uniform estimates \eqref{EQUATSS4.1Recurrence_Cnd2} and \eqref{EQUATSS4.1Recurrence_Cnd4} on $(-\eta_{c,n})$ and $\zeta_n$.\\
Concerning the main term, which is the second one, since we have also $\vert y_n \vert \leq \delta_4$, we deduce:
\begin{align*}
\left\vert \frac{r}{\alpha_0-\varphi_{1,\alpha_0}^--y_n} - \varphi_{1,\alpha_0}^- \right\vert \leq (1-h_4) \vert y_n \vert \leq (1-h_4) \delta_x.
\end{align*}
In the end, we have:
\begin{align}
\label{EQUATSS4.1Cntrlx_n+1Ameli} %% Contrôle sur x_{n+1}, amélioré
\vert x_{n+1} \vert \leq \frac{1}{2} h_4 \delta_x + (1-h_4)\delta_x = \left( 1 - \frac{h_4}{2} \right) \delta_x,
\end{align}
that is, the inequality \eqref{EQUATSS4.1Recurrence_Cnd9} propagates to the rank $n+1$. Concerning the estimate \eqref{EQUATSS4.1Recurrence_Cnd10} on $\vert y_{n+1} \vert$, let us recall that we proved that the $(n+2)$-th collision is takes place, so in particular $d_{n+2}$ and $\cos\theta_{n+2}$ are well-defined, and we have by a direct triangular inequality:
\begin{align*}
\vert y_{n+1} \vert \leq d_{n+2} \frac{\eta_{s,n+1}}{(-\eta_{c,n+1})} + \text{Max}(d_{n+1},2\zeta_{n+1}) + \vert (-\cos\theta_{n+2}) - (-\cos\theta_0) \vert + V_1 \frac{\tau_{n+1}}{(-\eta_{c,n+1})} + \vert x_{n+1} \vert.
\end{align*}
In particular, since we just obtained that $\vert x_{n+1} \vert$ is smaller than $\delta_x$, we have $\eta_{s,n+1}/(-\eta_{c,n+1}) \leq \alpha_0/2$, the four first terms are bounded by $(h_4\delta_x)/8$ using the estimates on $-\eta_{c,n}$ and $\zeta_n$, and using to conclude the improved inequality \eqref{EQUATSS4.1Cntrlx_n+1Ameli}, we obtain:
\begin{align*}
\vert y_{n+1} \vert \leq \frac{h_4}{8}\delta_x + \frac{h_4}{8}\delta_x + \frac{h_4}{8}\delta_x + \frac{h_4}{8}\delta_x + \left( 1 - \frac{h_4}{2} \right) \delta_x = \delta_x,
\end{align*}
and so the estimate \eqref{EQUATSS4.1Recurrence_Cnd10} is also true for $n+1$.
\newline
In summary, assuming that the $n$-th collision takes place, involving the central particle \circled{0}, and that the conditions \eqref{EQUATSS4.1Recurrence_Cnd1}-\eqref{EQUATSS4.1Recurrence_Cnd10} hold for the $n$-th collision, then the $(n+1)$-th collision takes place, involving again the central particle \circled{0}, and the estimates \eqref{EQUATSS4.1Recurrence_Cnd1}-\eqref{EQUATSS4.1Recurrence_Cnd10} hold true for the $(n+1)$-th collision. Since these estimates were initially verified, we deduce that infinitely many collisions take place, all of them involving the central particle \circled{0}, and we have therefore exhibited an explicit trajectory performing an inelastic collapse.
\end{proof}

\begin{remar}
In the previous proof we had to use the fact that the ratio of the norms $\vert W_{c,n} \vert/\vert W_{s,n} \vert$ is not much larger than $1$. This restriction simplified the proof and enabled to consider a single collision at each iteration of the recursion. Most likely, it is possible to remove such an assumption, but the proof becomes more intricate.\\
The restriction on $\vert W_{c,n} \vert/\vert W_{s,n} \vert$ looks natural at the first glance, because if the relative velocities of the two external particles are too different, one would expect that the system of particles separates after few collisions. However it is not the case in general: the proof indicates that such a ratio can be arbitrarily, up to choose initial distances between the particles small enough. We recover the justification of the ``flat surface approximation'' described in \cite{ZhKa996}. In such a regime, the particles are close enough, so that the geometry (i.e., the curvature) of the particles plays no role, neither as the tangential components of the relative velocities, and the only relevant information turns out to be the normal components, describing how fast the tangent planes of the particles, orthogonal to the contact lines, are approaching each other.
\end{remar}

\section{Sequence of collisions involving the three pairs of particles: the triangular collapse}
\label{SSECTIO3.5SuiteColli_012_}

In this section we will consider a system of three inelastic particles performing an inelastic collapse, and we will assume that all the three pairs of particles \circled{0}-\circled{1}, \circled{0}-\circled{2} and \circled{1}-\circled{2} are involved infinitely many times in the collisions. To the best of our knowledge, this is the first study of a collapse with such a collision order. Such a collapse is a priori possible, and we will see that it is essentially different from the nearly-linear collapse, concerning its geometrical characteristics.

\begin{defin}[Triangular inelastic collapse]
\label{DEFINSS3.5CollapseTriangu}
Let $r \in\ ]0,1[$ be a positive real number smaller than $1$, and let us consider a system of three inelastic particles \circled{0}, \circled{1} and \circled{2} evolving according to the $r$-inelastic hard sphere flow \eqref{EQUATSS2.1_Loi_de_Newton_}, \eqref{EQUATSS2.1VitesPost-Colli}, on a time interval $[0,\tau^*[$, and undergoing an inelastic collapse at time $\tau^* > 0$.\\
We say that the system experiences a \emph{triangular collapse} if the sequence of collisions presents infinitely many collisions between all the three pairs of particles \circled{0}-\circled{1}, \circled{0}-\circled{2}, and \circled{1}-\circled{2}.
\end{defin}

\begin{remar}
At this step, the order of collisions is not prescribed. As for the nearly-linear collapse, the name ``triangular collapse'' will become self-evident in what follows.
\end{remar}
\noindent
First, we will obtain the final geometry of the system of particles, at the time of the inelastic collapse. From that geometry, we will deduce that the eventual order of the collisions is prescribed. Finally, from this prescribed order, we will present a formal argument strongly suggesting that such an inelastic collapse is not stable, that is, not stable under perturbation of the initial data in the phase space (if it exists).

\subsection{Geometry of the triangular collapse}
\label{SSSCT3.5.1GeomeTriagColla}

Considering a system of three inelastic particles performing an inelastic collapse such that all the three pairs are colliding infinitely many times before the collapsing time $\tau^*$, we see that Proposition \ref{PROPOSS3.2ConveAngleDista} enables immediately to deduce the final geometry of the system of particles at $\tau^*$.

\begin{propo}[Geometry of the triangular collapse]
\label{PROPOSS3.5GeomeTriagColla}
Let $r \in\ ]0,1[$ be a positive real number smaller than $1$, and let us consider a system of three inelastic particles \circled{0}, \circled{1} and \circled{2} evolving according to the $r$-inelastic hard sphere flow \eqref{EQUATSS2.1_Loi_de_Newton_}, \eqref{EQUATSS2.1VitesPost-Colli}, on a time interval $[0,\tau^*[$, and undergoing a triangular inelastic collapse at time $\tau^* > 0$.\\
Then, the three particles \circled{0}, \circled{1} and \circled{2} of the system are in contact at the collapsing time $\tau^*$, that is, we have:
\begin{align}
\vert \overline{\omega}_1 \vert = 1, \hspace{5mm} \vert \overline{\omega}_2 \vert = 1 \hspace{5mm} \text{and} \hspace{5mm} \vert \overline{\omega}_2 - \overline{\omega}_1 \vert = 1.
\end{align}
\end{propo}
\noindent
In the present case, the three particles of the system are in contact at the collapsing time $\tau^*$. Therefore, the centers of the particles form an equilateral triangle, hence the name of the collapse given in Definition \ref{DEFINSS3.5CollapseTriangu}. We will see that this geometry has important implications.

\subsection{Order of the collisions of the triangular collapse}
\label{SSSCT3.4.2OrdreTriagColla}

In this section, we will see that the final geometry of a system of particles experiencing a triangular collapse will eventually experience a periodic sequence of collisions.

\begin{theor}[Order of collisions of the triangular collapse]
\label{PROPOSS3.5OrdreTriagColla}
Let $r \in\ ]0,1[$ be a positive real number smaller than $1$, and let us consider a system of three inelastic particles \circled{0}, \circled{1} and \circled{2} evolving according to the $r$-inelastic hard sphere flow \eqref{EQUATSS2.1_Loi_de_Newton_}, \eqref{EQUATSS2.1VitesPost-Colli}, on a time interval $[0,\tau^*[$, and undergoing a triangular inelastic collapse at time $\tau^* > 0$.\\
Then, up to relabelling the particles, the sequence of collisions before the collapsing time $\tau^*$ becomes eventually the infinite repetition of the three collisions \circled{0}-\circled{1}, \circled{0}-\circled{2}, \circled{1}-\circled{2}, in that order.
\end{theor}

\begin{proof}
We consider a collapsing system of three particles, so without loss of generality, we can assume that a collision between the particles \circled{0} and \circled{1} takes place infinitely many times, and such that such a collision is followed immediately after by a collision between the particles \circled{0} and \circled{2}. We have then infinitely many pairs of consecutive collisions of type \circled{0}-\circled{1},\circled{0}-\circled{2}. Let us index by $\varphi(n)$ the collision of type \circled{0}-\circled{1}, taking place at $t_{\varphi(n)}$, followed immediately by the collision $\varphi(n)+1$, of type \circled{0}-\circled{2}, at time $t_{\varphi(n)+1}$.\\
Using the notations introduced in Section \ref{SECTION__2EcrirSystmDynam}, let us denote by $W_1$ and $W_2$ the relative velocities of the respective pairs of particles \circled{0}-\circled{1} and \circled{0}-\circled{2} at time $t_{\varphi(n)}^+$, so that in particular the pair \circled{0}-\circled{1} is in a post-collisional configuration. We denote also $x_1\left(t_{\varphi(n)}\right) - x_0\left(t_{\varphi(n)}\right) = \omega_1\left(t_{\varphi(n)}\right)$, in short, by $\omega_1$, and similarly:
\begin{align*}
\displaystyle{ \frac{x_2\left(t_{\varphi(n)}\right) - x_0\left(t_{\varphi(n)}\right)}{\left\vert x_2\left(t_{\varphi(n)}\right) - x_0\left(t_{\varphi(n)}\right) \right\vert}} = \omega_2\left(t_{\varphi(n)}\right) = \omega_2.
\end{align*}
As before, the normal components of $W_1$ and $W_2$, respectively along $\omega_1$ and $\omega_2$, are denoted by $\eta_1$ and $\eta_2$. Let us denote by $W_1'$ and $W_2'$ the corresponding relative velocities at time $t_{\varphi(n)+1}^+$, just after the collision $\varphi(n)+1$ takes place. Accordingly, we denote
\begin{align*}
\frac{x_1\left(t_{\varphi(n)+1}\right) - x_0\left(t_{\varphi(n)+1}\right)}{\left\vert x_1\left(t_{\varphi(n)+1}\right) - x_0\left(t_{\varphi(n)+1}\right) \right\vert} = \omega_1\left(t_{\varphi(n)+1}\right)
\end{align*}
by $\omega_1'$ and $x_2\left(t_{\varphi(n)+1}\right) - x_0\left(t_{\varphi(n)+1}\right) = \omega_2\left(t_{\varphi(n)+1}\right)$ by $\omega_2'$. Then, we denote by $\eta'_1$, $\eta'_2$ the respective normal components of $W_1'$ and $W_2'$ (along the directions $\omega'_1$ and $\omega_2'$ respectively).\\
After introducing these notations, we can compute $\eta_1'$. In particular, in order to have a collision of type \circled{0}-\circled{1} following immediately the $(\varphi(n)+1)$-th collision (of type \circled{0}-\circled{2}), it is necessary to have
\begin{align*}
\eta_1' < 0.
\end{align*}
But using the expression of the post-collisional velocities obtained after a collision of type \circled{0}-\circled{2}, we obtain:
\begin{align*}
\eta_1' = W_1'\cdot \omega_1' = \left( W_1 - \frac{(1+r)}{2} \left( W_2\cdot\omega_2'\right) \omega_2' \right) \cdot \omega_1',
\end{align*}
and using now the explicit expressions of the angular parameters at the time $t_{\varphi(n)+1}$ of the $\varphi(n)+1$-th collision:
\begin{align*}
\omega_1' = \frac{x_1\left(t_{\varphi(n)+1}\right) - x_0\left(t_{\varphi(n)+1}\right)}{\vert x_1\left(t_{\varphi(n)+1}\right) - x_0\left(t_{\varphi(n)+1}\right) \vert} = \frac{x_1\left(t_{\varphi(n)}\right) - x_0\left(t_{\varphi(n)} \right) + \tau W_1}{\vert x_1\left(t_{\varphi(n)}\right) - x_0\left(t_{\varphi(n)} \right) + \tau W_1 \vert} = \frac{\omega_1 + \tau W_1}{1 + d'},
\end{align*}
and
\begin{align*}
\omega_2' = \frac{x_2\left(t_{\varphi(n)+1}\right) - x_0\left(t_{\varphi(n)+1}\right)}{\vert x_2\left(t_{\varphi(n)+1}\right) - x_0\left(t_{\varphi(n)+1}\right) \vert} = x_2\left(t_{\varphi(n)}\right) - x_0\left(t_{\varphi(n)} \right) + \tau W_2 = (1+d)\omega_2 + \tau W_2,
\end{align*}
(where we denoted by $\tau$ the difference $\tau_{\varphi(n)}$ between the times $t_{\varphi(n)}$ and $t_{\varphi(n)+1}$ of the $\varphi(n)$-th and $\left(\varphi(n)+1\right)$-th collision), we get:
\begin{align}
\label{EQUATSS3.5Eta1_Explicite_}
\eta_1' = \frac{1}{1+d'}\left( \eta_1 + \tau \vert W_1 \vert^2 \right) - \frac{(1+r)}{2} \left( (1+d)\eta_2 + \tau \vert W_2 \vert^2 \right) \omega_1'\cdot\omega_2'.
\end{align}
Now, the term $\eta_1 + \tau \vert W_1 \vert^2$ is positive, because at time $t_{\varphi(n)}$ a collision of type \circled{0}-\circled{1} took place, and so $\eta_1$ is positive. As for the second term, using also the explicit expression \eqref{EQUATSS2.2TempsColliTauV3} of $\tau$, which is the key observation, we find:
\begin{align*}
(1+d)\eta_2 + \tau \vert W_2 \vert^2 &= (1+d)\eta_2 + \frac{(1+d)(-\eta_2)}{\vert W_2 \vert^2}\left[ 1 - \sqrt{1 - \zeta} \right] \vert W_2 \vert^2 \\
&= + (1+d) \sqrt{1-\zeta} \eta_2.
\end{align*}
In the end, we need to conclude about the sign of the term:
\begin{align*}
- \frac{(1+r)}{2} (1+d) \sqrt{-\zeta} \left(\omega_1'\cdot\omega_2'\right) \eta_2.
\end{align*}
But by assumption at time $t_{\varphi(n)+1}^+$ the system is in a pre-collisional configuration, in the sense that the next collision that will take place is of type \circled{0}-\circled{2}, so that $-\eta_2 > 0$. Now, according to Proposition \ref{PROPOSS3.2ConveAngleDista}, the sequences $\left( \omega_1\left(t_{\varphi(n)}\right)\right)_n$ and $\left( \omega_2\left(t_{\varphi(n)}\right)\right)_n$ are converging such that $\omega_1\left(t_{\varphi(n)}\right) \cdot \omega_2\left(t_{\varphi(n)}\right)$ converges towards $1/2$. Therefore, for $n$ large enough we have $\omega_1'\cdot\omega_2' \simeq \frac{1}{2}$, and so in particular $\omega_1'\cdot\omega_2' > 0$, which gives:
\begin{align}
\label{EQUATSS3.5ConclusionEta1>}
\eta_1' > 0,
\end{align}
and in conclusion if a collision takes place after the $\left(\varphi(n)+1\right)$-th collision, it cannot be of type \circled{0}-\circled{1}. Then, since in addition at this time the particles \circled{0} and \circled{2} are in a pre-collisional configuration, the next collision is necessarily of type \circled{1}-\circled{2}.\\
Since the labelling of the particles we used throughout the proof was arbitrary, the result of Theorem \ref{PROPOSS3.5OrdreTriagColla} follows.

\end{proof}

\begin{remar}
The inequality \eqref{EQUATSS3.5ConclusionEta1>} is obtained from the explicit expression \eqref{EQUATSS3.5Eta1_Explicite_}, which is completely exact. Therefore, the conclusion holds in full generality.
\end{remar}
\noindent
Theorem \ref{PROPOSS3.5OrdreTriagColla} has several direct consequences, that we will now describe. First of all, we can deduce directly that the normal components of the relative velocities of the three pairs of particles are all vanishing in the regime of the collapse, because Theorem \ref{PROPOSS3.5OrdreTriagColla} implies that the triangular collapse has finite maximal gaps in the sense of Definition \ref{DEFINSS3.2FonctCompt_Gaps} (because eventually, for any type, the gap between two consecutive collisions of the same type is exactly $3$).

\begin{propo}[Convergence of the normal components of the relative velocities for the triangular collapse]
\label{PROPOSS3.5CTriaConveNorma}
Let $r \in\ ]0,1[$ be a positive real number smaller than $1$, and let us consider a system of three inelastic particles \circled{0}, \circled{1} and \circled{2} evolving according to the $r$-inelastic hard sphere flow \eqref{EQUATSS2.1_Loi_de_Newton_}, \eqref{EQUATSS2.1VitesPost-Colli}, on a time interval $[0,\tau^*[$, and undergoing a triangular inelastic collapse at time $\tau^* > 0$.\\
Then, for all $(i,j) \in\ \left\{(0,1),(0,2),(1,2)\right\}$ we have:
\begin{align}
\eta_{(i,j),n} \xrightarrow[n \rightarrow +\infty]{} 0,
\end{align}
where $\eta_{(i,j),n}$ denotes the normal component of the relative velocity between the particles \circled{i} and \circled{j}, that is $\eta_{(i,j),n} = W_{(i,j)}\cdot \omega_{(i,j)} = (v_j-v_i)\cdot\frac{\left( x_j-x_i \right)}{\vert x_j-x_i \vert}$.
\end{propo}

\subsection{Formal study of the triangular collapse with the matrix of collisions}

Since we know that the order of the collisions is eventually prescribed for the triangular collapse, we can also compute how the relative velocities evolve asymptotically with the help of the collision matrices.\\
Let us be more explicit: we can assume without loss of generality that the sequence of collisions eventually becomes \circled{0}-\circled{1}, \circled{0}-\circled{2}, \circled{1}-\circled{2},$\dots$. We will write the matrices corresponding to these three collisions, and consider their product.\\
\newline
Starting from the relative velocities $\left( W_1(t_n),W_2(t_n) \right) = \left(W_1,W_2\right) = \left(v_1-v_0,v_2-v_0\right)$, let us assume that the first collision, of type \circled{0}-\circled{1}, takes place at time $t_n$, with angular parameters $\omega_1 = \omega_1(n) = \frac{\left(x_1(t_n)-x_0(t_n)\right)}{\vert x_1(t_n)-x_0(t_n) \vert}$ and $\omega_2 = \omega_2(n) = \frac{\left(x_2(t_n)-x_0(t_n)\right)}{\vert x_2(t_n)-x_0(t_n) \vert}$. The post-collisional relative velocities $\left( W_1(t_{n+1}),W_2(t_{n+1}) \right) = \left( W_1',W_2' \right)$ write:
\begin{align}
\label{EQUATSS3.5VitesColli_0-1_}
\left\{
\begin{array}{rl}
W_1' &= W_1 -(1+r) W_1\cdot\omega_1\omega_1, \\
W_2' &= W_2 - \frac{(1+r)}{2} W_1\cdot\omega_1\omega_1,
\end{array}
\right.
\end{align}
which can be rewritten with the help of a $2d\times 2d$ matrix as:
\begin{align*}
\begin{pmatrix} W_1' \\ W_2' \end{pmatrix} = A_{\mathfrak{a},n} \begin{pmatrix} W_1 \\ W_2 \end{pmatrix}
\end{align*}
with
\begin{align}
\label{EQUATSS3.5DefinMatrix_A_a}
A_{\mathfrak{a},n} = A_{\mathfrak{a},n}(\omega_1) = \begin{pmatrix} I_d - (1+r) \omega_1\otimes\omega_1 & 0 \\ -\frac{(1+r)}{2} \omega_1\otimes\omega_1 & I_d \end{pmatrix},
\end{align}
where $I_d$ denotes the $d\times d$ identity matrix, and where $u \otimes v$ denotes the linear mapping (and by extension, its matrix) defined as $ \left( u\otimes v\right)(x) = \left(u \cdot x\right) v$, for two fixed vectors $u,v \in\ \mathbb{R}^d$.\\
In the same way, at time $t_{n+1}$, the second collision, of type \circled{0}-\circled{2}, takes place with angular parameters $\omega_1' = \omega_1(n+1) = \frac{\left(x_1(t_{n+1})-x_0(t_{n+1})\right)}{\vert x_1(t_{n+1})-x_0(t_{n+1}) \vert}$ and $\omega_2' = \omega_2(n+1) = \frac{\left(x_2(t_{n+1})-x_0(t_{n+1})\right)}{\vert x_2(t_{n+1})-x_0(t_{n+1}) \vert}$. After this collision the post-collisional relative velocities $\left( W_1(t_{n+2}),W_2(t_{n+2}) \right) = \left( W_1'',W_2'' \right)$ write:
\begin{align*}
\begin{pmatrix} W_1'' \\ W_2'' \end{pmatrix} = A_{\mathfrak{b},n+1} \begin{pmatrix} W_1' \\ W_2' \end{pmatrix}
\end{align*}
with
\begin{align}
\label{EQUATSS3.5DefinMatrix_A_b}
A_{\mathfrak{b},n+1} = A_{\mathfrak{b},n+1}(\omega_2') = \begin{pmatrix} I_d & - \frac{(1+r)}{2} \omega_2'\otimes\omega_2' \\ 0 & I_d - (1+r) \omega_2'\otimes\omega_2' \end{pmatrix}.
\end{align}
Finally, the third collision, of type \circled{1}-\circled{2}, takes place at time $t_{n+2}$ with angular parameter $\omega_3'' = \omega_3(n+2) = \frac{\left(x_2(t_{n+2})-x_1(t_{n+2})\right)}{\vert x_2(t_{n+2})-x_1(t_{n+2}) \vert}$. In this case, the post-collisional velocities write:
\begin{align*}
%\label{EQUATSS3.5VitesColli_0-1_}
\left\{
\begin{array}{rl}
v_0''' &= v_0,\\
v_1''' &= v_1 -\frac{(1+r)}{2} \left(v_1''-v_2''\right)\cdot\omega_3''\omega_3'', \\
v_2''' &= v_2 + \frac{(1+r)}{2} \left(v_1''-v_2''\right)\cdot\omega_3''\omega_3'',
\end{array}
\right.
\end{align*}
so that in this case
\begin{align*}
\left\{
\begin{array}{rcl}
W_1''' &= v_1'''-v_0''' &= W_1'' - \frac{(1+r)}{2} \left( W_1''-W_2'' \right)\cdot\omega_3'' \omega_3'',\\
W_2''' &= v_2'''-v_0''' &= W_2'' + \frac{(1+r)}{2} \left(W_1''-W_2''\right)\cdot\omega_3'' \omega_3'',
\end{array}
\right.
\end{align*}
which provides the matrix equation:
\begin{align*}
\begin{pmatrix} W_1''' \\ W_2''' \end{pmatrix} = A_{\mathfrak{c},n+2} \begin{pmatrix} W_1'' \\ W_2'' \end{pmatrix}
\end{align*}
with
\begin{align}
\label{EQUATSS3.5DefinMatrix_A_c}
A_{\mathfrak{c},n+2} = A_{\mathfrak{c},n+2}(\omega_3'') = \begin{pmatrix} I_d - \frac{(1+r)}{2}\omega_3''\otimes\omega_3'' & \frac{(1+r)}{2} \omega_3''\otimes\omega_3'' \\ \frac{(1+r)}{2} \omega_3''\otimes\omega_3'' & I_d - \frac{(1+r)}{2} \omega_3''\otimes\omega_3'' \end{pmatrix}.
\end{align}
Therefore, the relative velocities $\left(W_{1,3m+n},W_{2,3m+n}\right)$ are obtained as:
\begin{align}
\begin{pmatrix} W_1(t_{3m+n}) \\ W_2(t_{3m+n}) \end{pmatrix} = \prod_{k=0}^{m-1} \left[A_{\frak{c},3k+2}A_{\frak{b},3k+1}A_{\frak{a},3k}\right] \begin{pmatrix} W_1(t_n) \\ W_2(t_n) \end{pmatrix},
\end{align}
and then it is clear that studying the matrix:
\begin{align}
\mathfrak{A}_k = A_{\frak{c},3k+2}A_{\frak{b},3k+1}A_{\frak{a},3k}
\end{align}
is of central importance to understand the behaviour of the velocities of the system of particles at the regime of collapse. In particular, we will seek the spectrum of this matrix, following the method already adopted, for instance, in \cite{McYo991} and \cite{ZhKa996}.\\
However, there is a difficulty coming from the fact that these matrices depend on the index $k \in\ \mathbb{N}$, through the angular parameters $\omega_1$, $\omega_2'$ and $\omega_3''$. The behaviour of these angular parameters encodes the non linearity of the problem, and therefore it would require a careful study to obtain rigorous results. We will not perform such a study in the present work, but according to Proposition \ref{PROPOSS3.5GeomeTriagColla}, the angular parameters are converging as $k$ goes to infinity, and we have:
\begin{align*}
\omega_1(3k+n) \xrightarrow[k \rightarrow +\infty]{} \omega_{1,\infty},\hspace{3mm} \omega_2(3k+n+1) \xrightarrow[k \rightarrow +\infty]{} \omega_{2,\infty},\hspace{3mm} \omega_3(3k+2+n) \xrightarrow[k \rightarrow +\infty]{} \omega_{3,\infty}
\end{align*}
for some unitary vectors $\omega_{1,\infty}$, $\omega_{2,\infty}$ and $\omega_{3,\infty}$ such that:
\begin{align}
\omega_{1,\infty} + \omega_{3,\infty} = \omega_{2,\infty},
\end{align}
translating the fact that the geometry of the collapsing system of particles is an equilateral triangle. Therefore, the sequence of matrices $\left(\mathfrak{A}_k\right)_k$ converges towards a limiting matrix, that we will now study in full detail.

\begin{defin}[Limiting velocity matrix $\mathfrak{A}_\infty$ of the triangular collapse]
\label{DEFINSS3.5MatriLimitVites}
Let $r \in\ ]0,1[$ be a positive real number smaller than $1$, and let us consider a system of three inelastic particles \circled{0}, \circled{1} and \circled{2} evolving according to the $r$-inelastic hard sphere flow \eqref{EQUATSS2.1_Loi_de_Newton_}, \eqref{EQUATSS2.1VitesPost-Colli}, on a time interval $[0,\tau^*[$, and undergoing a triangular inelastic collapse at time $\tau^* > 0$, and such that the order of the collisions becomes eventually the infinite repetition of the triple \circled{0}-\circled{1}, \circled{0}-\circled{2}, \circled{1}-\circled{2}, with the $(3n + n_0)$-th collisions of type \circled{0}-\circled{1} taking place at $t_{3n+n_0}$, the $(3n + 1 + n_0)$-th collisions of type \circled{0}-\circled{2} at $t_{3n+1+n_0}$ and the $(3n + 2 + n_0)$-th collisions of type \circled{1}-\circled{2} at $t_{3n+2+n_0}$, for some $n_0 \in\ \mathbb{N}$, and all $n \in\ \mathbb{N}$.\\
Let us denote by $\omega_{1,\infty}, \omega_{2,\infty}, \omega_{3,\infty} \in\ \mathbb{S}^{d-1}$ the three respective limits of the sequences $\left(x_1(t_{3n+n_0}) - x_0(t_{3n+n_0})\right)_n$, $\left(x_2(t_{3n+1+n_0}) - x_0(t_{3n+1+n_0})\right)_n$ and $\left(x_2(t_{3n+2+n_0}) - x_1(t_{3n+2+n_0})\right)_n$.\\
We define the matrix $\mathfrak{A}_\infty$ as:
\begin{align}
\mathfrak{A}_\infty = A_\mathfrak{c}(\omega_{3,\infty}) A_\mathfrak{b}(\omega_{2,\infty}) A_\mathfrak{a}(\omega_{1,\infty}),
\end{align}
where $A_\mathfrak{a}$, $A_\mathfrak{b}$ and $A_\mathfrak{c}$ are defined using respectively \eqref{EQUATSS3.5DefinMatrix_A_a}, \eqref{EQUATSS3.5DefinMatrix_A_b} and \eqref{EQUATSS3.5DefinMatrix_A_c}, replacing respectively $\omega_1$, $\omega_2'$ and $\omega_3''$ by $\omega_{1,\infty}$, $\omega_{2,\infty}$ and $\omega_{3,\infty}$.\\
The matrix $\mathfrak{A}_\infty$ will be called the \emph{limiting velocity matrix} of the system.
\end{defin}

\noindent
There is already a first property of the limiting velocity matrix that allows to reduce drastically the complexity of the study of the spectrum.

\begin{propo}[Restriction of the limiting velocity matrix to the plane $\text{span}\left(\omega_{1,\infty},\omega_{2,\infty}\right)$]
Let us consider the sub-vector space $E$ of $\mathbb{R}^d\times\mathbb{R}^d$ defined as:
\begin{align}
E = \left\{ (x,y) \in\ \mathbb{R}^d\times\mathbb{R}^d\ /\ x\cdot \omega_{1,\infty} = x\cdot\omega_{2,\infty} = 0,y\cdot \omega_{1,\infty} = y\cdot\omega_{2,\infty} = 0 \right\},
\end{align}
Then, the matrix $\mathfrak{A}_\infty$ restricted to $E$ is the identity matrix.
\end{propo}
\noindent
We can therefore study the matrix $\mathfrak{A}_\infty$, restricted to the following vector space:
\begin{align}
\label{EQUATSS3.5DefinSs-Ev__F__}
F = \text{span}\left(\omega_{1,\infty},\omega_{2,\infty}\right) \times \text{span}\left(\omega_{1,\infty},\omega_{2,\infty}\right) \subset \mathbb{R}^d\times\mathbb{R}^d,
\end{align}
which is a four-dimensional space.

\begin{theor}[Spectrum of the limiting velocity matrix $\overline{\mathfrak{A}}_\infty$]
\label{THEORSS3.5SpectMatricVitLm}
Let us consider the restriction $\overline{\mathfrak{A}}_\infty$ of the limiting velocity matrix $\mathfrak{A}_\infty$ to the space $F \subset \mathbb{R}^d\times\mathbb{R}^d$, defined in \eqref{EQUATSS3.5DefinSs-Ev__F__}.\\
The image of the space $F$ by the matrix $\mathfrak{A}_\infty$ is contained into $F$ itself. $\overline{\mathfrak{A}}_\infty$ is therefore a $4\times 4$ matrix, that writes
\begin{align}
\label{EQUATSS3.5MatrixAlim_ResF}
\overline{\mathfrak{A}}_\infty = \begin{pmatrix}
\frac{(-r^3+5r^2-59r-1)}{64} & \frac{\sqrt{3}(r+1)}{8} & \frac{(r^2-4r-5)}{16} & - \frac{\sqrt{3}(r+1)}{8} \\
\frac{\sqrt{3}(r^3+3r^2+11r+9)}{64} & \frac{(5-3r)}{8} & -\frac{\sqrt{3}(r^2+4r+3)}{16} & \frac{(3r+3)}{8} \\
\frac{(-r^3+17r^2+13r-5)}{64} & \frac{\sqrt{3}(r+1)}{8} & \frac{(r^2-16r-1)}{16} & -\frac{\sqrt{3}(r+1)}{8} \\
\frac{\sqrt{3}(-r^3+r^2+13r+11)}{64} & \frac{(3r+3)}{8} & \frac{\sqrt{3}(r^2-1)}{16} & \frac{(5-3r)}{8}
\end{pmatrix}
\end{align}
in the basis $(\omega_{1,\infty},0),(\omega_{1,\infty}^\perp,0),(0,\omega_{2,\infty}),(0,\omega_{2,\infty}^\perp)$, where $\omega_{1,\infty}^\perp$ is the unitary vector orthogonal to $\omega_{1,\infty}$ in the plane $\text{span}\left(\omega_{1,\infty},\omega_{2,\infty}\right)$ and such that $\omega_{1,\infty}^\perp \cdot \omega_{2,\infty} > 0$, and $\omega_{2,\infty}^\perp$ is the unitary vector orthogonal to $\omega_{2,\infty}$ in the plane $\text{span}\left(\omega_{1,\infty},\omega_{2,\infty}\right)$ and such that $\omega_{2,\infty}^\perp \cdot \omega_{1,\infty} > 0$.\\
The polynomial characteristic $\chi_{\overline{\mathfrak{A}}_\infty}(\lambda)$ of the matrix $\overline{\mathfrak{A}}_\infty$ writes:
\begin{align}
\label{EQUATSS3.5Alim_PolynCarac}
\chi_{\overline{\mathfrak{A}}_\infty}(\lambda) = (\lambda-1) Q_{\overline{\mathfrak{A}}_\infty}(\lambda),
\end{align}
where:
\begin{align}
\label{EQUATSS3.5Alim_DetaiPCarQ}
Q_{\overline{\mathfrak{A}}_\infty}(\lambda) = \lambda^3 + \frac{(r^3-9r^2+171r-11)}{64} \lambda^2 + \frac{-11r^3 + 171 r^2 -9r +1}{64} \lambda + r^3.
\end{align}
For any $r \in\ ]0,1[$, the polynomial $Q_{\overline{\mathfrak{A}}_\infty}$ has a single real root $\lambda_0$ which satisfies:
\begin{align}
\label{EQUATSS3.5Alim_ContrRacRL} % Contrôle de la Racine Réelle
-1 < -r < \lambda_0 < -r^3 < 0.
\end{align}
The polynomial $Q_{\overline{\mathfrak{A}}_\infty}$ has in addition two complex conjugated roots $\lambda_i^\pm$ which satisfy:
\begin{align}
\label{EQUATSS3.5Alim_ContrRacCx} % Contrôle des Racines Complexes
0 < \vert \lambda_0 \vert < \vert \lambda_i^\pm \vert < 1. 
\end{align}
In particular, the spectrum of the restricted limiting velocity matrix $\overline{\mathfrak{A}}_\infty$ is composed with the four eigenvalues:
\begin{align}
1,\lambda_0,\lambda_i^+,\lambda_i^-,
\end{align}
all of geometric multiplicity equal to $1$, and all except $1$ contained in the open unit disk of the complex plane. In addition, the eigenspace $E_1$ associated to the eigenvalue $\lambda = 1$ is:
\begin{align}
\label{EQUATSS3.5Alim_EspPp_L=1_} % Espace Propre associé à la valeur propre Lambda = 1
\text{span}\left(0,1,0,1\right)
\end{align}
(written in the basis $(\omega_{1,\infty},0),(\omega_{1,\infty}^\perp,0),(0,\omega_{2,\infty}),(0,\omega_{2,\infty}^\perp)$).
\end{theor}

\begin{proof}
Let us start with proving \eqref{EQUATSS3.5MatrixAlim_ResF}. In the basis $(\omega_{1,\infty},0),(\omega_{1,\infty}^\perp,0),(0,\omega_{2,\infty}),(0,\omega_{2,\infty}^\perp)$ of the vector space $F$, the matrices $A_\mathfrak{a}(\omega_{1,\infty})$, $A_\mathfrak{b}(\omega_{2,\infty})$ and $A_\mathfrak{c}(\omega_{3,\infty})$  can be written as:
\begin{align}
A_\mathfrak{a}(\omega_{1,\infty}) = \begin{pmatrix}
-r & 0 & 0 & 0 \\
0 & 1 & 0 & 0 \\
-\frac{(1+r)}{4} & 0 & 1 & 0 \\
\frac{\sqrt{3}(1+r)}{4} & 0 & 0 & 1
\end{pmatrix},\hspace{5mm}
A_\mathfrak{b}(\omega_{2,\infty}) = \begin{pmatrix}
1 & 0 & -\frac{(1+r)}{4} & 0 \\
0 & 1 & -\frac{\sqrt{3}(1+r)}{4} & 0 \\
0 & 0 & -r & 0 \\
0 & 0 & 0 & 1
\end{pmatrix}
\end{align}
and
\begin{align}
A_\mathfrak{a}(\omega_{3,\infty}) = \begin{pmatrix}
\frac{(7-r)}{8} & \frac{\sqrt{3}(1+r)}{8} & -\frac{(1+r)}{8} & -\frac{\sqrt{3}(1+r)}{8} \\
\frac{\sqrt{3}(1+r)}{8} & \frac{(5-3r)}{8} & \frac{\sqrt{3}(1+r)}{8} & \frac{3(1+r)}{8} \\
-\frac{1+r}{8} & \frac{\sqrt{3}(1+r)}{8} & \frac{(7-r)}{8} & -\frac{\sqrt{3}(1+r)}{8} \\
-\frac{\sqrt{3}(1+r)}{8} & \frac{3(1+r)}{8} & -\frac{\sqrt{3}(1+r)}{8} & \frac{(5-3r)}{8}
\end{pmatrix}.
\end{align}
We can then perform directly the computation of the product $A_\mathfrak{c}(\omega_{3,\infty}) A_\mathfrak{b}(\omega_{2,\infty}) A_\mathfrak{a}(\omega_{1,\infty})$, which provides \eqref{EQUATSS3.5MatrixAlim_ResF}.\\
\newline
The computation of the characteristic polynomial of $\overline{\mathfrak{A}}_\infty$ is presented in the Appendix (see Appendix \ref{APPENSS___PolynChara_Alim}).\\
\newline
The fact that $\lambda = 1$ is an eigenvalue is clear, either as a consequence of the characteristic polynomial (that one can obtain by a brute force method), of the form $(\lambda-1)Q_{\overline{\mathfrak{A}}_\infty}$, or by direct inspection, by computing the image of $\left(0,1,0,1\right)$ by the matrix, which enables also to deduce that this vector is contained in the eigenspace associated to $\lambda = 1$.\\
\newline
Let us now turn to the study of the polynomial of third degree $Q_{\overline{\mathfrak{A}}_\infty}$.\\
First, the second derivative of $Q_{\overline{\mathfrak{A}}_\infty}$ is zero only for $\lambda = \lambda_c$, where
\begin{align}
\lambda_c = \frac{-r^3+9r^2-171r+11}{6\cdot32}\cdotp
\end{align}
Since the graph of $\lambda \mapsto Q_{\overline{\mathfrak{A}}_\infty}(\lambda)$ is a cubic with limits $\pm\infty$ when $\lambda \rightarrow \pm\infty$, the first derivative $Q'_{\overline{\mathfrak{A}}_\infty}$ of $Q_{\overline{\mathfrak{A}}_\infty}$ is a convex parabola, that reaches its minimum at $\lambda_c$. The minimum of the derivative is:
\begin{align}
Q'_{\overline{\mathfrak{A}}_\infty}(\lambda_c) &= - \frac{r^6}{12288} + \frac{3r^5}{2048} - \frac{141r^4}{4096} + \frac{247r^3}{3072} + \frac{1131r^2}{4096} + \frac{339r}{2048} + \frac{71}{12288} \nonumber\\
&= - \frac{(r+1)^2}{12288} \left[ (r+1)^4 - 24(r+1)^3 + 528(r+1)^2 - 2880(r+1) +2304 \right].
\end{align}
It is then possible to show that the function $q:s \mapsto s^4 -24s^3 +528s^2 -2880s +2304$ is behaving as a convex parabola, that is, with limits $+\infty$ as $s$ goes to $\pm\infty$, and such that the function is first strictly decreasing on an interval of the form $]-\infty,s_0]$, and then becomes strictly increasing on $[s_0,+\infty[$. Indeed, its second derivative is the quadratic function:
\begin{align}
q'':s \mapsto 12s^2 - 6\cdot24 s + 2\cdot 528,
\end{align}
of discriminant:
\begin{align}
\Delta = (6\cdot 24)^2 -4\cdot12\cdot2\cdot528 = -29952 < 0.
\end{align}
But now since:
\begin{align}
q(1) = 1-24+528-2880+2304 = -71 \hspace{3mm} \text{and} \hspace{3mm} q(2) = 2^4 -24\cdot 8 + 528\cdot 4 -2880\cdot 2 +2304 = -1520,
\end{align}
we can deduce that
\begin{align}
q(s) < 0 \ \forall s \in\ [1,2],
\end{align}
or again
\begin{align}
(r+1)^4 - 24(r+1)^3 + 528(r+1)^2 - 2880(r+1) + 2304 < 0 \ \forall r \in\ [0,1],
\end{align}
and so we deduce that the minimum $Q'_{\overline{\mathfrak{A}}_\infty}(\lambda_c)$ of the derivative of $Q_{\overline{\mathfrak{A}}_\infty}$ is strictly positive, and therefore the cubic $\lambda \mapsto Q_{\overline{\mathfrak{A}}_\infty}(\lambda)$ is always strictly increasing, and so it has a unique real root. The proof of the uniqueness of $\lambda_0$ is complete. We deduce in addition that there exist also two complex conjugated roots of $Q_{\overline{\mathfrak{A}}_\infty}$.\\
\newline
Let us now prove the inequalities \eqref{EQUATSS3.5Alim_ContrRacRL} on $\lambda_0$.\\
We start with the proof that $-r < \lambda_0$. This comes from the fact that
\begin{align}
Q_{\overline{\mathfrak{A}}_\infty}(-r) = r\frac{(r^4+2r^3-2r-1)}{64}\cdotp
\end{align}
Observing that $r^4+2r^3-2r-1$ vanishes for $r=1$ we can factorize to get:
\begin{align}
Q_{\overline{\mathfrak{A}}_\infty}(-r) = \frac{r}{64}(r-1)(r+1)^3 < 0 \ \forall r \in\ ]0,1[.
\end{align}
Since $Q_{\overline{\mathfrak{A}}_\infty}$ is strictly increasing we deduce that $-r < \lambda_0$.\\
Let us now turn to the estimate $\lambda_0 < -r^3$ of \eqref{EQUATSS3.5Alim_ContrRacRL}. We have:
\begin{align}
Q_{\overline{\mathfrak{A}}_\infty}(-r^3) = \frac{r^3}{64}\left(-63r^6 -9r^5 +171r^4 -171r^2 +9r +63\right).
\end{align}
Observing again that $-63r^6 -9r^5 +171r^4 -171r^2 +9r +63$ is vanishing for $r=1$, as well as for $r=-1$, we can factorize to get:
\begin{align}
Q_{\overline{\mathfrak{A}}_\infty}(-r^3) = -\frac{9r^3}{64}(r-1)(r+1)^3(7r^2-13r+7)
\end{align}
where the quadratic polynomial is always strictly positive. Therefore $Q_{\overline{\mathfrak{A}}_\infty}(-r^3) > 0$, which proves $\lambda_0 < -r^3$, and so the proof of \eqref{EQUATSS3.5Alim_ContrRacRL} is complete.\\
\newline
Let us now turn to \eqref{EQUATSS3.5Alim_ContrRacCx}. This estimate is obtained using the Vi\`{e}te's formula linking the roots of a polynomial with its coefficients. In particular, denoting by $\lambda_i^\pm$ the two complex conjugated roots of $Q_{\overline{\mathfrak{A}}_\infty}$, we have:
\begin{align}
(-\lambda_0) \cdot \lambda_i^+ \cdot \lambda_i^- = \vert \lambda_0 \vert \cdot \vert \lambda_i^\pm \vert^2 = r^3,
\end{align}
which provides, combined with \eqref{EQUATSS3.5Alim_ContrRacRL}, first:
\begin{align}
\vert \lambda_i^\pm \vert^2 = \frac{r^3}{\vert \lambda_0 \vert} < 1,
\end{align}
and second:
\begin{align}
\frac{\vert \lambda_i^\pm \vert^2}{\vert \lambda_0 \vert^2} = \frac{r^3}{\vert \lambda_0 \vert^3} > 1,
\end{align}
so that \eqref{EQUATSS3.5Alim_ContrRacCx} is proved.\\
\newline
Finally, since the $4\times4$ matrix $\overline{\mathfrak{A}}_\infty$ has four distinct eigenvalues, each of the associated eigenspaces has dimension $1$, which enables to obtain \eqref{EQUATSS3.5Alim_EspPp_L=1_}, determining in particular completely the eigenspace associated to $\lambda = 1$.
The proof of Theorem \ref{THEORSS3.5SpectMatricVitLm} is complete.
\end{proof}

\begin{remar}
It is interesting to note the similarity between the polynomial $Q_{\overline{\mathfrak{A}}_\infty}$, described in \eqref{EQUATSS3.5Alim_DetaiPCarQ}, and the different polynomials obtained in \cite{CDKK999}, where the authors study in much detail the inelastic collapse of a one-dimensional system of four inelastic particles. In particular, as in \cite{CDKK999}, note the intriguing symmetry linking the coefficients $a_1(r)$, $a_2(r)$ of first and second degrees of $Q_{\overline{\mathfrak{A}}_\infty}$, that are themselves polynomials in $r$, and such that
\begin{align*}
a_1(1/r) = \frac{-11+171r-9r^2+r^3}{64r^3} = \frac{a_2(r)}{r^3} \cdotp
\end{align*}
\end{remar}
\noindent
We assumed that we have a system of particles experiencing a triangular collapse, such that the order of the collisions becomes eventually the infinite repetition of \circled{0}-\circled{1}, \circled{0}-\circled{2}, \circled{1}-\circled{2}. Let us denote, as above, by $n_0$ the index of the collision such that this periodic sequence starts. In order to have the existence of such a triangular collapse, we need to have a configuration of the system, just before time $t_{n_0}$, such that the particles \circled{0} and \circled{1} are in a pre-collisonal configuration. Then, at time $t_{n_0}$, when the collision \circled{0}-\circled{1}, the particles \circled{0} and \circled{2} have now to be in a pre-collisional configuration. These two conditions are imposing:
\begin{align*}
\overline{\eta}_1 &< 0,\\
\overline{\eta}_2 &< \frac{(1+r)}{4} \eta_1,
\end{align*}
where $\overline{\eta}_1$ and $\overline{\eta}_2$ are the normal components of the relative velocities between \circled{0} and \circled{1}, and \circled{0} and \circled{2}, at time $t_{n_0}^-$.\\
Therefore, these two conditions define a cone $\mathcal{C}^1$ in the four-dimensional space $F$ on which the matrix $\overline{\mathfrak{A}}_\infty$ is acting, and since we need to be able to iterate this matrix and still recover a configuration of particles that remains in a pre-collisional configuration for the next triple of collisions, the image of the initial datum $(\eta_1,w_1,\eta_2,w_2)$ (written in the basis $\left(\omega_{1,\infty},0),(\omega_{1,\infty}^\perp,0),(0,\omega_{2,\infty}),(0,\omega_{2,\infty}^\perp)\right)$ by $\overline{\mathfrak{A}}_\infty$, as well as all its iterates, have to remain in this cone. This idea is already in \cite{McYo991} and \cite{ZhKa996}, where other limiting velocity matrices, associated to the linear collapse, are considered.\\
\newline
In our case, we can also provide another condition, which restricts even more the cone of admissible velocities.\\
Indeed, in the basis $(\omega_{1,\infty},0),(\omega_{1,\infty}^\perp,0),(0,\omega_{2,\infty}),(0,\omega_{2,\infty}^\perp)$, let us denote by $(x,y,z,t)$ an initial datum written in coordinates, and $(x',y',z',t')$ its image by the matrix $\overline{\mathfrak{A}}_\infty$. Assuming that $(x,y,z,t)$ is in the cone $\mathcal{C}_1$, for $r$ small enough we have
\begin{align}
\frac{r^2-4r-5}{16} < 0 \hspace{3mm} \text{and} \hspace{3mm} \frac{r^2-16r-1}{16} < 0
\end{align}
and so
\begin{align}
x' &> \frac{(-r^3+5r^2-59r-1)}{64}x + \frac{(1+r)}{64}(r^2-4r-5)x + \frac{\sqrt{3}}{8}(r+1)(y-t)\\
&> \frac{2r^2-68r-6}{64}x + \frac{\sqrt{3}}{8}(r+1)(y-t),
\end{align}
and in the same way
\begin{align}
\label{EQUATSS3.5Minor_z'__Triag}
z' &> \frac{(-r^3+17r^2+13r-5)}{64}x + \sqrt{3} \frac{(r+1)}{8}(y-t) + \frac{(1+r)}{64}(r^2-16r-1)x\\
&> \frac{2r^2-4r-6}{64}x + \sqrt{3}\frac{(r+1)}{8}(y-t).
\end{align}
In particular, for $r$ small enough, \eqref{EQUATSS3.5Minor_z'__Triag} imposes \begin{align}
y-t < 0
\end{align}
in order to have $z' < 0$. This defines then a restricted cone $\mathcal{C}_2$ as:
\begin{align}
(x,y,z,t) \in \mathcal{C}_2 \Longleftrightarrow
\left\{
\begin{array}{rl}
x &< 0,\\
z &< \frac{(1+r)}{4} x,\\
y &< t.
\end{array}
\right.
\end{align}
Then, we see that the initial datum $(x,y,z,t)$, as well as all of its iterates under multiplication by $\overline{\mathfrak{A}}_\infty$, have to remain in this cone $\mathcal{C}_2$.\\
We can now conclude formally about the stability of the triangular collapse: any sequence of iterations $\left(x_n,y_n,z_n,t_n\right)_n = \left( \overline{\mathfrak{A}}_\infty(x_0,y_0,z_0,t_0)\right)_n$ converges to a point $X_\infty \in \text{span}(0,1,0,1)$, exponentially fast. In addition, among the eigenvalues in the unit disk, the complex conjugated $\lambda_i^\pm$ are dominating, these complex eigenvalues are inducing a rotation-dilatation is some plane, so that eventually the iterations are spiraling around $X_\infty$. It is possible to show that the plan in which the spiral is taking place has to be transverse to the boundary of the cone $\mathcal{C}_2$, so that, considering \emph{only the limiting matrix}, that is, considering only the \emph{linearization} of the full dynamics of the velocities of the triangular collapse, there is no chance to find a sequence of iterations $\left(x_n,y_n,z_n,t_n\right)_n$ that starts and remains in the restricted cone $\mathcal{C}_2$.\\
\newline
These last considerations suggest strongly that the triangular collapse, if it exists, should be extremely unstable. In particular, Theorem \ref{THEORSS3.5SpectMatricVitLm} implies that if the triangular collapse exists, it has to be associated to sequences of iterations $\left(x_n,y_n,z_n,t_n\right)_n$ constrained to evolve in a fixed manifold, that has to be of dimension $1$.\\
However, the whole discussion of this last paragraph can only be considered as formal. To prove or disprove the last statements concerning the iterations of the relative velocities, one would need to study the non-linearity of the complete dynamics, that is, to take into account the effect of the variation of the angular parameters, that induces perturbations of the collision matrices $A_\mathfrak{a}$, $A_\mathfrak{b}$ and $A_\mathfrak{c}$. Actually, the question if such a triangular collapse actually exists is not tackled neither: we only provided necessary conditions that such a triangular collapse has to fulfil, if ever it exists. The study of the full dynamical system should also allow to conclude about the existence of such a collapse, and most probably one should be able to provide an explicit construction of self-similar data corresponding to a system of inelastic particles performing a triangular collapse.

\section{Conclusion and perspective}

In this work we revisited the classical problem of the collapse of inelastic particles. We proved new general results concerning the behaviour of systems of three inelastic particles in the collapsing regime, in particular about the converging quantities and the possible order of collisions between the particles. These new results on converging quantities enabled us to produce a rigorous study of the system of particles in the case of the nearly-linear collapse. As an application, we constructed a set of positive measure in the phase space, composed with initial configurations of three particles, leading all to the nearly-linear collapse, in a stable way. We discovered also another possible order of collisions, leading to the hypothetical triangular collapse, and we studied completely the limiting velocity matrix of such a collapse. In addition, we proved that there is no other possible order of collisions for a system of three collapsing particles.\\ %We also considered several levels of simplification of the full dynamical system governing the movement of the particles, from the evolutions of the leading order terms, towards the simplifying Ansatz of the low energy limit. We proved the stability of the Zhou and Kadanoff regime in a non trivial domain of the phase space for the leading order system using local arguments, and performed numerical investigations to observe qualitative behaviours for the whole domain in which the orbits of the leading order system are globally well-defined. We proved these qualitative behaviours for the orbits of the simplified system in the low energy limit.\\
\newline
However, an important number of open questions remain. In the companion paper \cite{DoVeArt} we derived formally the evolution equations of the leading order terms of the full dynamical system, for which we proved the stability of the ZK-regime in a non trivial region, and we studied numerically the orbits of the system of the leading order terms. The behaviour of these orbits allowed us to state conjectures, and we were able to prove these conjectures in a particular regime. Proving these conjectures in the general case would achieve the characterization of the ZK-regime. We hope to solve these conjectures in the future, and prove that way an Alexander's theorem for systems of three particles.\\
On a more general ground, the results concerning the converging quantities and orders of collisions we obtained for the system of three inelastic particles can most probably be extended to a larger number of particles. It would be for example interesting to study the possible geometries at the time of the collapse, and the eventual order of collisions, for systems composed with four or more particles. In turn, we hope that such results would allow to obtain an Alexander's theorem for systems composed with an arbitrary number of particles.\\
Concerning the model, we considered only \emph{strictly} positive restitution coefficients $r \in\ ]0,1[$. The limit $r=0$ seems to be completely unstudied so far. There is no mathematical result concerning inelastic collapse in this case. It might be possible to obtain an Alexander theorem more easily for this model. Nevertheless, the case $r=0$ is known to be a singular limit of the cases $r>0$.\\
Finally, concerning the triangular collapse, we studied in some detail necessary conditions for such a collapse to happen. It would be interesting to have an explicit description of a system of three inelastic particles realizing a triangular collapse, if it ever can be realized, for it would be the first purely non-linear (that is, which cannot be realized as a perturbation of a one-dimensional collapsing figure) collapsing figure of inelastic particles ever observed.%\\
%Finally, the results concerning our double-collision mapping, which is the evolution law of the leading order terms of the full dynamical system, are rather partial. As a matter a fact, we were able to carry on only a perturbative study of trivial equilibria and stable manifold. The main limitation came from the extreme technicality of the computations, although completely elementary. It is likely that a perturbation argument can provide with the existence of the Zhou and Kadanoff equilibrium for the general double-collision mapping, at least in the case when the difference $b-a$ is small. There is also some hope coming from the recent development of the computer assisted proofs in analysis: we conjectured the existence of stable manifolds, maybe the most modern methods can help to conclude rigorously here.\\

\begin{appendices}

\section{Characteristic polynomial of the matrix $\overline{\mathfrak{A}}_\infty$}
\label{APPENSS___PolynChara_Alim}

This section is devoted to describe how to obtain the characteristic polynomial of the limiting velocity matrix $\overline{\mathfrak{A}}_\infty$ (\eqref{EQUATSS3.5Alim_PolynCarac} and \eqref{EQUATSS3.5Alim_DetaiPCarQ} page \pageref{EQUATSS3.5Alim_PolynCarac}), associated to the triangular collapse, introduced in Definition \ref{DEFINSS3.5MatriLimitVites}, and written in coordinates in \eqref{EQUATSS3.5MatrixAlim_ResF} page \pageref{EQUATSS3.5MatrixAlim_ResF}.\\
\newline
Starting from:
\begin{align*}
\overline{\mathfrak{A}}_\infty = \begin{pmatrix}
\frac{(-r^3+5r^2-59r-1)}{64} & \frac{\sqrt{3}(r+1)}{8} & \frac{(r^2-4r-5)}{16} & - \frac{\sqrt{3}(r+1)}{8} \\
\frac{\sqrt{3}(r^3+3r^2+11r+9)}{64} & \frac{(5-3r)}{8} & -\frac{\sqrt{3}(r^2+4r+3)}{16} & \frac{(3r+3)}{8} \\
\frac{(-r^3+17r^2+13r-5)}{64} & \frac{\sqrt{3}(r+1)}{8} & \frac{(r^2-16r-1)}{16} & -\frac{\sqrt{3}(r+1)}{8} \\
\frac{\sqrt{3}(-r^3+r^2+13r+11)}{64} & \frac{(3r+3)}{8} & \frac{\sqrt{3}(r^2-1)}{16} & \frac{(5-3r)}{8}
\end{pmatrix}
\end{align*}
we compute the determinant of such a matrix with the usual manipulations on the rows and the columns, using the multilinearity of the determinant. Performing first:
\begin{align*}
L_1' = L_1-L_3 \hspace{5mm} \text{and} \hspace{5mm} L_2' = L_2+L_4,
\end{align*}
we obtain
\begin{align*}
\det\left(\lambda I_4 - \overline{\mathfrak{A}}_\infty\right) = \begin{vmatrix}
\lambda + \frac{(12r^2+72r-4)}{64} & 0 & -\lambda + \frac{(-12r+4)}{16} & 0 \\
\frac{\sqrt{3}(-4r^2-24r-20)}{64} & \lambda-1 & \frac{\sqrt{3}(4r+4)}{16} & \lambda-1 \\
\frac{(r^3-17r^2-13r+5)}{64} & -\frac{\sqrt{3}(r+1)}{8} & \lambda + \frac{(-r^2+16r+1)}{16} & \frac{\sqrt{3}(r+1)}{8} \\
\frac{\sqrt{3}(r^3-r^2-13r-11)}{64} & -\frac{(3r+3)}{8} & -\frac{\sqrt{3}(r^2-1)}{16} & \lambda + \frac{(-5+
3r)}{8}
\end{vmatrix},
\end{align*}
then applying $L_4' = L_4 - \sqrt{3}L_3$ provides
\begin{align*}
\det\left(\lambda I_4 - \overline{\mathfrak{A}}_\infty\right) = \begin{vmatrix}
\lambda + \frac{(12r^2+72r-4)}{64} & 0 & -\lambda + \frac{(-12r+4)}{16} & 0 \\
\frac{\sqrt{3}(-4r^2-24r-20)}{64} & \lambda-1 & \frac{\sqrt{3}(4r+4)}{16} & \lambda-1 \\
\frac{(r^3-17r^2-13r+5)}{64} & -\frac{\sqrt{3}(r+1)}{8} & \lambda + \frac{(-r^2+16r+1)}{16} & \frac{\sqrt{3}(r+1)}{8} \\
\frac{\sqrt{3}(16r^2-16)}{64} & 0 & -\sqrt{3}\lambda + \frac{\sqrt{3}(-16r)}{16} & \lambda -1
\end{vmatrix},
\end{align*}
and finally $C_4' = C_4+C_2$ enables to write:
\begin{align*}
\det\left(\lambda I_4 - \overline{\mathfrak{A}}_\infty\right) &= \begin{vmatrix}
\lambda + \frac{(12r^2+72r-4)}{64} & 0 & -\lambda + \frac{(-12r+4)}{16} & 0 \\
\frac{\sqrt{3}(-4r^2-24r-20)}{64} & \lambda-1 & \frac{\sqrt{3}(4r+4)}{16} & 2(\lambda-1) \\
\frac{(r^3-17r^2-13r+5)}{64} & -\frac{\sqrt{3}(r+1)}{8} & \lambda + \frac{(-r^2+16r+1)}{16} & 0 \\
\frac{\sqrt{3}(16r^2-16)}{64} & 0 & -\sqrt{3}\lambda + \frac{\sqrt{3}(-16r)}{16} & \lambda -1
\end{vmatrix} \\
&= (\lambda-1) \begin{vmatrix}
\lambda + \frac{(12r^2+72r-4)}{64} & 0 & -\lambda + \frac{(-12r+4)}{16} & 0 \\
\frac{\sqrt{3}(-4r^2-24r-20)}{64} & \lambda-1 & \frac{\sqrt{3}(4r+4)}{16} & 2 \\
\frac{(r^3-17r^2-13r+5)}{64} & -\frac{\sqrt{3}(r+1)}{8} & \lambda + \frac{(-r^2+16r+1)}{16} & 0 \\
\frac{\sqrt{3}(16r^2-16)}{64} & 0 & -\sqrt{3}\lambda + \frac{\sqrt{3}(-16r)}{16} & 1
\end{vmatrix}.
\end{align*}
We can finally develop the last determinant to obtain \eqref{EQUATSS3.5Alim_PolynCarac} and \eqref{EQUATSS3.5Alim_DetaiPCarQ} page \pageref{EQUATSS3.5Alim_PolynCarac}.\\
\end{appendices}
\newline
\noindent
\textbf{Acknowledgements.} The authors are grateful to E. Caglioti, I. Gallagher, B. Lods, M. Pulvirenti, C. Saffirio and S. Simonella for many stimulating discussions concerning the topic of the present article. The authors gratefully acknowledge the financial support of the Hausdorff Research Institute for Mathematics (Bonn) through the collaborative research center The mathematics of emerging effects (CRC 1060, Project-ID 211504053), and the Deutsche Forschungsgemeinschaft (DFG, German Research Foundation).

%Thanks: Chiara, Isabelle and Bertrand Lods (that tried Alexander a lot of times) Emanuele Caglioti (interesting discussion in Zürich)

%\title{Collapse of inelastic hard spheres in dimension $d \geq 2$}
%\author{Th\'eophile Dolmaire}
%\author{Juan J. L. Vel\'azquez}
%\affil{Institute for Applied Mathematics, University of Bonn}

% \thanks{, e-mail: \texttt{dolmaire@iam.uni-bonn.de}, \texttt{velazquez@iam.uni-bonn.de}}, }

E-mail address: \texttt{dolmaire@iam.uni-bonn.de}, \texttt{velazquez@iam.uni-bonn.de}.

\end{document}